\documentclass[11pt]{article}
\usepackage[a4paper,margin=1in]{geometry}
\usepackage{amsmath,amssymb,amsthm,mathtools,bm}
\providecommand{\ket}[1]{\left|#1\right\rangle}
\providecommand{\bra}[1]{\left\langle#1\right|}
\DeclareMathOperator{\Tr}{Tr}
\usepackage{enumitem}
\usepackage{booktabs}
\usepackage{graphicx}
\usepackage{tikz}
\usetikzlibrary{decorations.pathreplacing}
\usepackage[colorlinks=true,linkcolor=blue!55!black,citecolor=blue!55!black,%
            urlcolor=blue!55!black]{hyperref}

\theoremstyle{plain}
\newtheorem{theorem}{Theorem}
\newtheorem{proposition}{Proposition}
\newtheorem{lemma}{Lemma}
\newtheorem{corollary}{Corollary}
\newtheorem*{thmA}{Theorem A}
\newtheorem*{thmB}{Theorem B}
\newtheorem*{thmC}{Theorem C}

\theoremstyle{definition}
\newtheorem{definition}{Definition}
\newtheorem{remark}{Remark}

\DeclareMathOperator{\spec}{spec}

\DeclareMathOperator{\arctanh}{arctanh}
\DeclareMathOperator{\arccosh}{arccosh}

\newcommand{\Do}{\Delta_{\mathrm o}}

\newcommand{\deferproof}[1]{%
  \par\smallskip\noindent\emph{Proof:} Appendix~\ref{#1}.\par\smallskip}

\title{Hamiltonian Eigenvalue Transformation\\ by Tridiagonal Gadgets}
\author{%
  Arthur Braida\thanks{\texttt{arthur.braida@irif.fr}}\\
  \small Universit\'e Paris Cit\'e, CNRS, IRIF, France, \\ \small Alfréd Rényi Institute of Mathematics, Budapest, Hungary
  \and
  Joseph Cunningham\thanks{\texttt{joseph.cunningham@u-bordeaux.fr}}\\
  \small Universit\'e de Bordeaux, CNRS, LaBRI, France
  \and
  J\'er\'emie Roland\thanks{\texttt{jeremie.roland@ulb.be}}\\
  \small Universit\'e libre de Bruxelles, QuIC, Belgium%
}
\date{\today}

\begin{document}
\maketitle

\begin{abstract}
An analog device implements a local Hamiltonian $H$, but a polynomial $P(H)$ is
in general not local, so the device cannot implement it. We carry $P(H)$, to any
prescribed accuracy, as the action of a single time-independent local
Hamiltonian on an explicitly described invariant subspace. Short chains of
ancilla qubits are attached to $H$. A chain of $2m$ sites has a unique isolated
eigenvalue that is an analytic function of the input vanishing to order exactly
$2m$, because the input must cross the chain and come back before it can shift
the energy at the far end. Chains of different lengths therefore form a
triangular family, and a weighted sum of them reproduces a prescribed polynomial
term by term. Even chains give the even part and odd chains the odd part, so
any polynomial is reached. We prove uniqueness, analyticity and coefficient
bounds uniform in the chain length for inputs $H$ of norm below any $r<1$. 
Where the circuit model pays degree $2\ell$ in
$2\ell$ sequential oracle calls, the result here is one Hamiltonian of locality
two more than that of $H$, which pays the degree in
$O(\ell^{2}+\log^{2}(1/\varepsilon))$ ancillas and in energy scale. As an
application we filter a marked eigenstate of a local Hamiltonian. Composing a
synthesised square along the Chebyshev doubling identity works, but its energy
scale grows quasi-polynomially with the degree. Iterating instead the exact
eigenvalue branch of the two-site chain $k$ times gives a filter on $k+1$
ancilla qubits whose error decays geometrically in $k$, at an energy scale
polynomial in the inverse passband width and independent of $k$, the locality
growing by one per stage. In adiabatic optimisation with a rank-one driver, the construction replaces that driver, the one non-local ingredient, by a Hamiltonian with $O(\log n)$ ancillas and $O(\log n)$-body terms, at an energy scale polynomial in $n$. Simulations including every ancilla, up to $n=20$, reproduce the spectrum and the anticrossing gap of the ideal algorithm on a simple cost function.
\end{abstract}

\section{Introduction}\label{sec:intro}

Suppose you have an analog device. It implements a local Hamiltonian, you can
tune its couplings, and you can let it evolve. What you would like to run on it
is an algorithm, and continuous-time algorithms are often designed around operators
chosen for their analytic convenience rather than for their locality. The
prototypical one is analog search by quantum adiabatic
evolution~\cite{farhi1998analog,farhi2000adiabatic}, a model universal for
quantum computation~\cite{aharonov2007adiabatic}. The analog search algorithm interpolates towards
$-\ket m\!\bra m$, a rank-one projector on $2^{n}$ dimensions. Its spectral gap
can be studied analytically and the schedule that follows it is
optimal~\cite{RolandLocalAQC,Braida2025unstructured}. None of these results applies as such to hardware, because the projector is as
far from local as an operator can be and no device implements it.

What a device does give you is a local Hamiltonian $H_1$ whose ground state is
the marked string $\ket m$, an Ising Hamiltonian for instance. The projector is a
function of that Hamiltonian, namely the polynomial taking the value $1$ at the
ground energy and nearly $0$ on the rest of the bounded spectrum. So the distance between
what is analysed and what can be built is exactly one question.

\emph{Can we find a local Hamiltonian whose restriction to an invariant subspace is 
equal to some prescribed polynomial of the original Hamiltonian?}

We answer it by attaching small ancilla chains to the input
and computing exactly, rather than perturbatively, which
function of the input the isolated eigenvalue of such a chain is. The answer is
a polynomial family rich enough to synthesise any target, and the operator that
results is time independent and local. The transformation is chosen once, by a
list of chain lengths and real weights $\alpha_j$, and is then written into an operator
rather than sequenced in time.
Figure~\ref{fig:search} is that answer applied to the case above: on the encoded
subspace, and numerically in the full space with every ancilla for $n\le20$, the
local Hamiltonian our gadgets build reproduces the spectrum of analog search and
its gap at the anticrossing.

In comparison, in the circuit model, the quantum singular value transformation
(QSVT)~\cite{gilyen2019qsvt} applies an essentially arbitrary bounded polynomial
to the singular values of a block-encoded matrix, and recovers quantum
search~\cite{grover1996}, phase estimation~\cite{kitaev1995phase}, matrix
inversion~\cite{harrow2009linear} and Hamiltonian
simulation~\cite{berry2015taylor} as special
cases~\cite{low2017qsp,low2019qubitization,martyn2021grand}. Its mechanism is
interleaving, so a polynomial of degree $2\ell$ costs $2\ell$ sequential oracle
calls and a constant number of ancillas, which is to say that the degree is paid
in depth. The construction below pays it the other way. Degree $2\ell$
costs $O(\ell^{2})$ ancilla qubits and a larger energy scale, but no sequencing of the
transformation, which is a single operator that one switches on.
Digital and analog sit at opposite ends of the same trade-off, and having the
analog end available is a step towards carrying a guarantee proved for an idealised
Hamiltonian over to one that a device could actually run.
Table~\ref{tab:tradeoff} is the whole account of what each end pays.

\begin{table}[!tb]
\centering
\small
\renewcommand{\arraystretch}{1.3}
\begin{tabular}{@{}lllll@{}}
\toprule
 & \textbf{steps $k$} & \textbf{ancillas} & \textbf{locality}
 & \textbf{time overhead}\\
\midrule
\multicolumn{5}{@{}l}{\emph{a general polynomial of degree $2\ell$, to accuracy
$\varepsilon$}}\\
QSVT
 & & $O(1)$ & a circuit & $2\ell$ oracle calls\\
one pass
 & & $O\bigl(\ell^{2}+\log^{2}\tfrac1\varepsilon\bigr)$ & $\varkappa+2$
 & $(\Delta+2)\Lambda_L$\\
\addlinespace
\multicolumn{5}{@{}l}{\emph{the projector $\ket m\!\bra m$, to accuracy $\eta$, at
passband edge $d$}}\\
Chebyshev filter via recursion
 & $\bigl\lceil\log_2\bigl(\tfrac{\sigma}{2d}\ln\tfrac8\eta\bigr)\bigr\rceil$
 & $O\bigl(k\,(k+\log\tfrac1\eta)^{2}\bigr)$ & $\varkappa+2k+2$
 & $e^{O(k\,(k+\log\frac1\eta))}$\\
single block recursion
 & $\bigl\lceil\log_\varrho\tfrac{\sigma^2}{2d^2\eta}\bigr\rceil$
 & $k+1$ & $\varkappa+k+1$
 & $\bigl(\tfrac\sigma d\bigr)^{2\log_\varrho\theta}$\\
\bottomrule
\end{tabular}
\caption{What each route to a polynomial of a Hamiltonian pays. The two blocks
measure different errors and are not comparable row for row. Upper block: the
degree $2\ell$ is given, the coefficients are $O(1)$,
$\varepsilon=\|P(H)-P^{e}_L(H)\|$, the difference between the target and synthesised 
polynomials, and $\Lambda_L=\sum_{j\le L}|\alpha_j|$ the total weight 
of the combined gadget. Lower
block: the target is the projector, $\eta=\|\widehat R(H_1)-\ket m\!\bra m\|$ at
passband edge $d$, and $k$ is the number of doubling steps
that~\eqref{eq:Lneed} and Proposition~\ref{prop:single}(iii) require. The time
overhead of QSVT is its number of sequential oracle calls, so the circuit model
pays the degree in depth. The constructions of this paper need no sequencing,
since each is a single operator, and pay the degree in width, locality and
energy. Their time overhead is the factor by which a device with a fixed energy
budget stretches every evolution time, namely the norm of the Hamiltonian
normalised to its encoded action, stated at fixed $\Delta/\sigma$ and up to
constant factors, with $\theta,\varrho$ as in~\eqref{eq:Qmap}
and~\eqref{eq:rho}. Section~\ref{sec:single} compares the two
filters. The entries are Corollary~\ref{cor:L_needed},
Theorem~\ref{thm:realisation}, Lemma~\ref{lem:filter},
Propositions~\ref{prop:filter}, \ref{prop:energy} and~\ref{prop:single}, and
Corollary~\ref{cor:energy}.}
\label{tab:tradeoff}
\end{table}

For practical application on machines, locality is of great importance. The Hamiltonians one can build, on ultracold atoms,
on Rydberg arrays, on superconducting
hardware~\cite{bloch2012simulations,browaeys2020rydberg,daley2022practical}, are
local and few-body, and often stoquastic, so the reachable computational power
is constrained~\cite{bravyi2008stoquastic}.
The construction keeps the input's locality up to an additive constant per
layer, which is what makes the answer above worth having.

\paragraph{The route: gadgets, used forwards.}
Perturbative gadgets were introduced to simulate a target Hamiltonian:
one designs a large-detuning ancilla system whose low-energy effective
Hamiltonian reproduces a prescribed
interaction~\cite{kempe2006complexity,oliveira2008complexity,jordan2008perturbative,bravyi2008gadgets,biamonte2008realizable,cao2015gadgets,cao2015nostrong},
the error being controlled perturbatively in the inverse detuning. We use the
same mechanism in another way. Rather than asking which effective
interaction a gadget produces, we ask which function of the input a
gadget's isolated eigenvalue is, and we compute that function exactly instead of
to leading order. The answer, for the chain of
Section~\ref{sec:gadget}, is an analytic function of the input whose Taylor
expansion begins at a prescribed order. Prescribed leading order is exactly
what makes a family triangular, hence complete. The whole construction works in particular
because we can study the exact expansion to any given order. Triangularity is the
statement that the coefficients of $x^{2},\dots,x^{2m-2}$ in $\lambda_0^{(m)}$
are zero rather than small, since the synthesis inverts a triangular system, and
a family whose leading orders are only approximately separated gives no such
system. The synthesis error, in turn, is a sum of everything the chains
contribute beyond the target degree, so it needs a bound on the expansion
coefficients $c^{(m)}_k$ of $\lambda_0^{(m)}$. 

\subsection{The gadget}\label{sec:intro:gadget}

The object is an open chain of $2m$ sites with couplings alternating between the
input $x$ and $1$, and a detuning $-\Delta$ on the last site
(Figure~\ref{fig:chain}). It is a finite Su--Schrieffer--Heeger
chain~\cite{su1979solitons,asboth2016short} with intracell hopping $x$,
intercell hopping $1$ and one boundary defect. Replacing the scalar $x$ by an
input Hamiltonian $H$ makes the intracell bonds operator-valued, and the whole
object is a Hamiltonian on $2m$ ancilla qubits whose locality is two more than
that of $H$, so a $\varkappa$-local input gives a $(\varkappa+2)$-local gadget
(Lemma~\ref{lem:local}).

\begin{figure}[t]
\centering
\begin{tikzpicture}[
  every node/.style={font=\small},
  site/.style={circle,draw,thick,fill=white,inner sep=0pt,minimum size=7mm},
  sink/.style={circle,draw,thick,fill=black!12,inner sep=0pt,minimum size=7mm},
  xbond/.style={very thick},
  ubond/.style={thin}]
  \def\s{1.75}
  \foreach \i in {1,...,5}{\node[site]  (n\i) at (\s*\i,0) {\scriptsize $\i$};}
  \node[sink] (n6) at (\s*6,0) {\scriptsize $6$};
  \draw[xbond] (n1)--(n2) node[midway,above=2pt] {$x$};
  \draw[ubond] (n2)--(n3) node[midway,above=2pt] {$1$};
  \draw[xbond] (n3)--(n4) node[midway,above=2pt] {$x$};
  \draw[ubond] (n4)--(n5) node[midway,above=2pt] {$1$};
  \draw[xbond] (n5)--(n6) node[midway,above=2pt] {$x$};
  \node at (\s*6,-0.85) {$-\Delta$};
  \draw[decorate,decoration={brace,mirror,amplitude=5pt}]
       (\s*1-0.35,-1.35) -- (\s*6+0.35,-1.35)
       node[midway,below=6pt] {$2m$ ancilla sites, one excitation};
\end{tikzpicture}
\caption{The even gadget $E_{2m}(x)$ for $m=3$. Thick bonds carry the input $x$
and become the operator $H$ when the scalar is replaced by a Hamiltonian. Thin
bonds are unit couplings, and the last site carries the detuning $-\Delta$. The
quantity of interest is the eigenvalue branch $\lambda_0^{(m)}(x)$ that vanishes
at $x=0$.}
\label{fig:chain}
\end{figure}
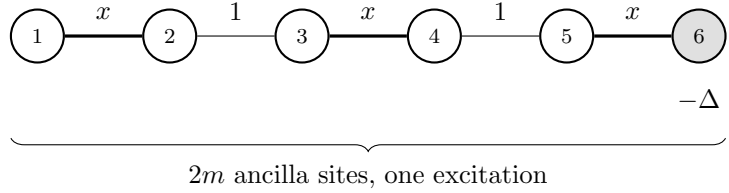

The chain has an eigenvalue $\lambda_0^{(m)}(x)$ that vanishes at $x=0$, and the
single fact the whole paper turns on is the validity of the following expansion in 
a neighbourhood of $x=0$:
\begin{equation}
\lambda_0^{(m)}(x)\;=\;\frac{x^{2m}}{\Delta}\;+\;O\bigl(x^{2m+2}\bigr).
\label{eq:intro_branch}
\end{equation}
The function $\lambda_0^{(m)}$ vanishes at $x=0$ to order exactly $2m$, with
$c^{(m)}_0=1/\Delta\neq0$. The picture behind it is that the input must traverse
the whole chain and come back before it can shift the energy at the far end. Chains 
of different lengths therefore produce eigenvalues with
different leading orders, and that is what makes the family complete. Attaching
several chains to the same input, with a real weight $\alpha_j$ on each, adds their
eigenvalues, and the weights can be chosen one after another so that the leading
orders reproduce a prescribed polynomial term by term. Because the chains occupy
disjoint sets of ancillas they do not interact, so the whole assembly is again a
single Hamiltonian of the same locality as one chain, and it carries a function that 
agrees with the target up to a tail
geometrically small in the number of chains, on a subspace described in
Section~\ref{sec:realisation}.
That subspace is invariant, so a computation prepared inside it and evolved
under the operator stays inside it.

The chains of Figure~\ref{fig:chain} have an even number of sites and produce
even functions of the input. Chains of odd length, with a slightly different
boundary, produce odd ones, and everything below transposes to them with two
modifications, which Appendix~\ref{app:odd} carries out. Any polynomial is
therefore reached by splitting it into its two parities and using chains of both
kinds, the constant and linear terms being added directly as multiples of $I$
and of $H$. To keep the exposition in one piece, the body of the paper treats the
even family.

\subsection{Results}\label{sec:intro:results}

Making this picture rigorous rests on two properties of the branch. It must be
unique, so that each chain has a well-defined eigenvalue and the
synthesis is a well-posed linear system, and it must be analytic on a
controlled complex disc, so that the resulting series converges for a given maximal
admissible input norm. Theorem~A proves both, Theorem~B turns them
into the synthesis of an arbitrary polynomial, and Theorem~C applies the
construction to a filter for a target state.

\begin{thmA}[The branch, Sections~\ref{sec:gadget} and~\ref{sec:branch}]
Fix a target radius $r\in(0,1)$ and let $\Delta\ge\Delta_\star(r)$, a threshold
given in closed form by one quadratic inequality $($Corollary~\ref{cor:Dmin},
equation~\eqref{eq:summary}$)$ and satisfying $\Delta_\star(r)\sim(1-r)^{-2}$.
Then for every $m\ge1$ the chain $E_{2m}(x)$ has, for each $|x|<r$, a
\emph{unique} eigenvalue in a disc around the origin, that eigenvalue is an
analytic function of $x$ on the whole complex disc $\{|x|<r\}$
$($Theorem~\ref{thm:onestep}$)$, and its Taylor series is
\[
\lambda_0^{(m)}(x)=\sum_{k\ge0}c^{(m)}_kx^{2m+2k},
\qquad
c^{(m)}_0=\frac1\Delta,
\quad
c^{(m)}_1=-\frac1\Delta\ \ (m\ge2),
\quad
c^{(m)}_k=0 \ \ (1<k<m),
\]
with $\bigl|c^{(m)}_k\bigr|\le C_r\big/\bigl(\Delta r^{2k}\bigr)$ and
$C_r=2/(1-r^{2})$, uniformly in $m$ $($Theorem~\ref{thm:series},
equation~\eqref{eq:summary_ck}$)$. The leading order is exactly $2m$, which is
what makes chains of different lengths a triangular family.
\end{thmA}

\begin{thmB}[Any polynomial, as a local Hamiltonian, Section~\ref{sec:synthesis}]
Let $\Delta\ge\Delta_\star(r)$, let $H$ be Hermitian and $\varkappa$-local with
$\|H\|<r$, let $P$ be an even polynomial with $P(0)=0$ and largest coefficient $M_P$ in
modulus, and let $\varepsilon>0$. There are real weights $\alpha_j$, and
$L=O\bigl(\max\{\deg P,\ \ln\tfrac{M_P}\varepsilon\}\bigr)$ of them, the
constant depending only on $r$ and $\|H\|/r$,
such that the single time-independent Hamiltonian
$\Sigma=\sum_{j}\alpha_jE_{2j}(H)$, on $O(L^{2})$ ancilla qubits and
$(\varkappa+2)$-local, satisfies
\[
\Sigma\,W \;=\; W\,P^{e}_L(H),
\qquad
\bigl\|P(H)-P^{e}_L(H)\bigr\|\;\le\;\varepsilon ,
\]
for the explicit isometry $W$ that sends each eigenvector of $H$ to itself
tensored with the branch eigenvector of every chain at that eigenvalue, and whose
image is therefore invariant under $\Sigma$
$($Theorem~\ref{thm:realisation}, Corollary~\ref{cor:L_needed}$)$. The error
decays geometrically,
$\limsup_{L\to\infty}\bigl\|P(H)-P^{e}_L(H)\bigr\|^{1/L}\le\|H\|^{2}/r^{2}$, so
the admissible input norm is strictly bounded by the disc radius $r$. 
Odd chains supply the odd part
$($Appendix~\ref{app:odd}$)$ and commute with the even ones, so every polynomial
is reached by splitting it into its two parities, its constant and linear terms
being added directly as $P(0)I$ and $P'(0)H$.
\end{thmB}

\begin{thmC}[A filter for a marked state, Section~\ref{sec:filter}]
Let $H_1$ be $\varkappa$-local with $\spec H_1\subset [-\sigma,-\delta]\cup\{0\}\cup[\delta,\sigma]$ and
$\ket m$ its unique $0$-energy state, and let $0<d\le\delta$, $d<\sigma$ and
$\Delta>0$. Applying the two-site gadget to $H_1$ and then recursively iterating it $k$ times,
with fixed affine constants in between, gives a $(\varkappa+k+1)$-local
Hamiltonian $\Sigma_k$ on $k+1$ ancilla qubits whose encoded action
$\widehat R_k(H_1)$ satisfies
\[
\bigl\|\widehat R_k(H_1)-\ket m\!\bra m\bigr\|\;\le\;\frac{\sigma^2}{2d^2}\,\varrho^{-k} ,
\]
with $\varrho\in(2,4)$ explicit, and whose norm, normalised to that action, is
bounded, at fixed $\Delta/\sigma$, by a polynomial in $\sigma/d$ that does not depend on $k$. The construction is exact and
needs no radius condition $($Proposition~\ref{prop:single}$)$.
\end{thmC}

For analog search on $n$ qubits, where $\sigma/d=n$, accuracy $\eta$ is reached
with $k=\bigl(2\log_2n+\log_2\frac1\eta\bigr)/\log_2\varrho+O(1)$ stages, and
$2/\log_2\varrho=1.007$ at $\Delta=10\sigma$. At constant $\eta$ that is
$O(\log n)$ ancilla qubits and a $(\varkappa+O(\log n))$-local Hamiltonian
(Section~\ref{sec:search}).

Theorem~C is not an instance of Theorem~B. For the filtering application, one could naturally
decompose the Chebyshev filter in the monomial basis and apply the procedure. 
However, to reach the desired accuracy, one would need a degree that scales polynomially
with $n$. Because the coefficients of such a filter are exponentially large with the degree
(Lemma~\ref{lem:MPbig}), the energy scale, in turn, would also explode. 
Another approach, composing a synthesised square along the Chebyshev
doubling identity $T_{2\ell}=2T_\ell^{2}-1$, does work, reaching the Chebyshev filter of
degree $2\ell$ on $O(\log^{3}\ell)$ ancillas, but
every stage pays the weight $\Lambda_L$ of the synthesis and the energy scale still grows as
$\ell^{\log_2(\Lambda_L/\sigma)}$ (Appendix~\ref{app:recursion}). Theorem~C rests on
the observation that the doubling needs only the shape of the square, a map
folding an interval onto itself under which the marked point escapes, and that the
exact branch of the two-site chain has that shape. The price is an error that
decays geometrically in $k$ rather than doubly exponentially, and
Figure~\ref{fig:search} shows that this is enough for analog search.

Theorem~C shows that the exact branch, not only its Taylor polynomial, is a useful primitive. It
reaches Sections~\ref{sec:branch} and~\ref{sec:synthesis} only through
Proposition~\ref{prop:m1}, which is the root of a quadratic, and through the
one-excitation form~\eqref{eq:even-gadget-Hamiltonian-matrix} of
Section~\ref{sec:def}. It uses neither $\Delta_\star(r)$ nor $C_r$, neither the
triangular system nor the truncation tail, and it carries no radius condition at
all, since by Proposition~\ref{prop:m1} the branch of the two-site chain is a
closed-form algebraic function whose only singularities sit at $x=\pm i\Delta/2$,
so on the real axis it is defined everywhere. The general theory earns its place
elsewhere. It is what shows that any polynomial is reachable.

\subsection{Related work}\label{sec:intro:related}

\emph{Polynomial transformations of a spectrum.} QSVT grew out of composite
pulses~\cite{low2016resonant}, quantum signal processing~\cite{low2017qsp} and
qubitisation~\cite{low2019qubitization}. All of them realise $P(A)$ for a
block-encoded $A$ by a circuit of length $\deg P$, with phase factors computed
classically~\cite{haah2019product,dong2021efficient}. Variants transform
eigenvalues rather than singular values~\cite{dong2022qetu,low2024eigenvalue}.
Two works are close to ours in name. \emph{Hamiltonian} singular value
transformation~\cite{lloyd2021hsvt} still runs a circuit of controlled
evolutions. The universal Hamiltonian eigenvalue transformation of Odake,
Kristj\'ansson, Taranto and Murao~\cite{odake2025universal} shares our goal, a
prescribed function $f$ of the eigenvalues of $H$ with its eigenvectors
unchanged, but neither our access model nor our output. There $H$ is unknown and
is given only through its evolution $e^{\pm iHt}$ as a black box, and what is
produced is the evolution $e^{-if(H)t}$, to accuracy $\varepsilon$, by
interleaving the black box with fixed operations on the system and one auxiliary
qubit, with a run time growing as $t^{2}/\varepsilon$. The transformation is a
protocol sequenced in time, and $f(H)$ never appears as a Hamiltonian. Here $H$
is known, and $P(H)$ is the exact action of one time-independent local
Hamiltonian on an invariant subspace, paid for in ancilla qubits and energy
scale. The two are complementary. Theirs applies to an unknown Hamiltonian and
to any sufficiently differentiable function. Ours yields a static operator,
which is what an analog device implements and what allows the transformed
Hamiltonian to enter a larger one as a term, as in Section~\ref{sec:search}. In
both earlier works the Hamiltonian is the resource and a sequence of operations
is the mechanism, whereas here a single time-independent Hamiltonian is the
whole object, and it is in this static sense that we speak of an eigenvalue
transformation. We use no block encoding, and $H$ is Hermitian, so the
transformation acts on eigenvalues. The analog and digital ends are also being joined from the
hardware side, by quantum signal processing on hybrid continuous- and
discrete-variable registers~\cite{liu2024mixed}. That is a different axis from
ours, oscillators against qubits rather than one time-independent operator
against a circuit.

\emph{Perturbative gadgets and Hamiltonian simulation.} The mediator
construction of Kempe--Kitaev--Regev~\cite{kempe2006complexity} and its
descendants~\cite{oliveira2008complexity,jordan2008perturbative,bravyi2008gadgets,biamonte2008realizable,cao2015gadgets,cao2015nostrong}
reproduce a target interaction in a low-energy subspace, with error controlled
by a Schrieffer--Wolff or Feshbach
expansion~\cite{schrieffer1966,feshbach1958,bravyi2011schrieffer}. The theory of
universal Hamiltonians~\cite{cubitt2018universal,piddock2018qudit} classifies
which families can simulate all others in this sense. Our gadget is of the same
physical type but is analysed exactly rather than to a fixed perturbative order,
and the object of interest is an isolated branch rather than a low-energy
effective operator.

\emph{A spectrum as a function of a spectrum.} The closest relative of the
present construction is spectral gap amplification~\cite{somma2013gap}, which
builds from a frustration-free $H$ of gap $\delta$ a second Hamiltonian whose gap
is of order $\sqrt\delta$, a quadratic amplification. That too realises a
prescribed function of a spectrum as a spectrum, by an algebraic rewriting
tailored to one function. What is different here is that the function is a free
parameter, fixed after the fact by a list of chain lengths and real weights. The
price of any scheme of this kind is the interaction strength it asks for, and
\cite{harley2026analogue} makes that point sharply for perturbative gadgets,
whose couplings scale polynomially in the system size and the inverse precision
and which are on that ground physically unrealisable as usually stated. We do not
escape the objection. Propositions~\ref{prop:energy} and~\ref{prop:single} and
Section~\ref{sec:search} say exactly what it costs here.

\emph{Spectral filtering.} Chebyshev filtering is the standard classical tool
for extracting an extremal eigenspace~\cite{zhou2006chebyshev,saad2011eigenvalue},
and its quantum analogues underpin ground-state preparation
algorithms~\cite{poulin2009preparing,ge2019faster,lin2020ground}. Our filter is
the same polynomial, carried by a local Hamiltonian rather than by
a circuit.

\subsection{Organisation}\label{sec:intro:org}

Table~\ref{tab:notation} lists, for reference, the symbols that recur across
sections. Each is also defined where it is first used.

Section~\ref{sec:gadget} defines the gadget and derives its characteristic
polynomial through a transfer matrix. Section~\ref{sec:branch} is the analysis: the real-axis bracket, the
Taylor series, the coefficient bound on a complex disc, and the closed-form
threshold $\Delta_\star(r)$. Section~\ref{sec:synthesis} solves the triangular
system, bounds the growth of its solution and the truncation tail, and assembles
the gadgets into one Hamiltonian. Section~\ref{sec:filter} applies it to a filter
for a marked state. It explains why a steep polynomial cannot be synthesised in
one pass, builds the filter instead by iterating the exact two-site block, and
puts the result into analog search, simulated with every ancilla.
Section~\ref{sec:discussion} discusses how the construction is used, the
one obstruction we have not removed, and the questions it leaves open.
Appendix~\ref{app:recursion} carries out the alternative route, a synthesised
square composed along the Chebyshev doubling identity, Appendix~\ref{app:proofs}
collects the longer proofs and Appendix~\ref{app:odd} the companion odd gadget.

\begin{table}[t]
\centering
\small
\renewcommand{\arraystretch}{1.15}
\begin{tabular}{@{}ll@{\qquad}ll@{}}
\toprule
\multicolumn{2}{@{}l}{\textbf{Gadget}} & \multicolumn{2}{l@{}}{\textbf{Synthesis}}\\
$m$ & half-length of the chain, $2m$ sites & $P$ & even target polynomial, degree $2\ell$\\
$x$ & the input, an eigenvalue of $H$ & $L$ & number of gadgets, $L\ge \ell$\\
$\Delta$ & detuning of the boundary site & $\alpha_j$ & synthesis weights, \eqref{eq:alpha_rec}\\
$E_{2m}(x)$ & the gadget matrix, \eqref{eq:even-gadget-matrix} & $P^{e}_L$ & $\sum_j\alpha_j\lambda_0^{(j)}$\\
$\lambda_0^{(m)}$ & branch vanishing at $x=0$ & $M_P$ & $\max_k|p_{2k}|$, coefficient of $P$\\
$c^{(m)}_k$ & Taylor coefficients of $\lambda_0^{(m)}$ & $\mathcal E_L$ & $P^{e}_L-P$, truncation error\\
$S_k$ & Chebyshev convergents & $\Phi_L$ & growth factor, \eqref{eq:Phi}\\
$R_m$ & $1+S_{m-1}/S_m$, Lemma~\ref{lem:poles} & $\Sigma$ & the realised Hamiltonian\\
\midrule
\multicolumn{2}{@{}l}{\textbf{Analysis}} & \multicolumn{2}{l@{}}{\textbf{Filter}}\\
$r$ & radius of the certified disc & $H_1$ & local input, gap $\delta$, norm $\le\sigma$\\
$\rho$ & Rouch\'e contour radius, $2a_m/\Delta$ & $\ell=2^{k}$ & filter half-degree in Appendix~\ref{app:recursion}\\
$a_m$ & $\sup_{|x|\le r}|M(x)|$, \eqref{eq:ab} & $d$ & passband edge, $d\le\delta$\\
$C_r$ & $2/(1-r^{2})$, coefficient constant & $\eta$ & target filter accuracy\\
$\Delta_\star(r)$ & required detuning, \eqref{eq:summary} & $R_\ell,\ \phi$ & Chebyshev filter of degree $2\ell$, its affine map\\
$\gamma$ & $\|H\|^{2}/r^{2}<1$ & $Q,\ f_i$ & single-block map and stages, Def.~\ref{def:single}\\
$\ket{\xi_j}$ & branch eigenvector of $E_{2j}$ & $\theta,\ \varrho$ & its prefactor and slope, \eqref{eq:Qmap}, \eqref{eq:rho}\\
 & & $\widehat R_k,\ u_0$ & its filter and initial escape, Prop.~\ref{prop:single}\\
 & & $\ket{\Xi_k}$ & encoded ancilla state, Prop.~\ref{prop:single}\\
\bottomrule
\end{tabular}
\caption{Recurring notation. The branch coefficients $c^{(m)}_k$ always carry a
superscript, and the filter values $c_i$ of Appendix~\ref{app:recursion} never
do.
The degree of a target is written $2\ell$ throughout, including the Chebyshev
filter $R_\ell$ of Section~\ref{sec:filter}, the letter $d$ is the passband edge
and never a degree, and the single block, which is not a polynomial, is indexed by
its number $k$ of iterations.}
\label{tab:notation}
\end{table}

\section{The even gadget}\label{sec:gadget}

This section defines the gadget and derives its eigenvalue equation
analytically, so that the behaviour of the target eigenvalue
$\lambda_0^{(m)}(x)$ can be studied rigorously in Section~\ref{sec:branch}. The
chain is tridiagonal, so its characteristic polynomial is generated by a
transfer matrix, and that transfer matrix runs the Chebyshev recurrence.
Section~\ref{sec:branch} uses two outputs of this section: the reduced
eigenvalue equation~\eqref{eq:eig_reduced} and the partial-fraction expansion of
its right-hand side~\eqref{eq:partialfrac}. 

\subsection{Definition and local realisation}\label{sec:def}

We first define the gadget as a $2m\times2m$ matrix of a scalar input $x$, then
as an operator when $x$ is replaced by a Hamiltonian $H$. Lemma~\ref{lem:local}
shows that this operator is carried by a Hamiltonian on $2m$ ancilla qubits, on
their one-excitation sector, whose locality exceeds that of $H$ by two.

For each integer $m\geq 1$, define the $2m\times 2m$ tridiagonal Hermitian matrix
\begin{equation}
E_{2m}(x) =
\begin{pmatrix}
0 & x & 0 & \cdots & 0 \\
x & 0 & 1 & \ddots & \vdots \\
0 & 1 & 0 & \ddots & 0 \\
\vdots & \ddots & \ddots & \ddots & x \\
0 & \cdots & 0 & x & -\Delta
\end{pmatrix},
\label{eq:even-gadget-matrix}
\end{equation}
where the off-diagonal couplings alternate between $x$ and $1$, starting with $x$.

This is a finite Su--Schrieffer--Heeger chain~\cite{su1979solitons,asboth2016short}
with intracell hopping $x$, intercell hopping $1$, and a defect of strength
$-\Delta$ on the last site. We call the couplings the \emph{$x$-bonds} and the
\emph{unit bonds} throughout.

When the scalar coupling $x$ is replaced by an input Hamiltonian $H$, the resulting block matrix $E_{2m}(H)$ can be written as
\begin{multline}
E_{2m}(H) = H\otimes \sum_{k=1}^m \ket{k}\bra{k} \otimes X
+ I\otimes\sum_{k=1}^{m-1} \bigg(\ket{k+1}\bra{k}\otimes \begin{pmatrix} 0 & 1 \\ 0 & 0 \end{pmatrix} + \ket{k}\bra{k+1}\otimes \begin{pmatrix} 0 & 0 \\ 1 & 0 \end{pmatrix}\bigg) \\
+ I\otimes\ket{m}\bra{m} \otimes \begin{pmatrix}0 & 0 \\ 0 & -\Delta\end{pmatrix}. \label{eq:even-gadget-Hamiltonian-matrix}
\end{multline}
This operator acts on the domain of $H$ tensored with the ancilla register $\mathbb{C}^{m}\otimes \mathbb{C}^{2}$:

\begin{lemma}[Local realisation]\label{lem:local}
The operator~\eqref{eq:even-gadget-Hamiltonian-matrix} is the restriction to an invariant subspace of a Hamiltonian on the system and $2m$ ancilla qubits, whose locality is two more than that of $H$. For $m=1$ a single ancilla qubit and one extra unit of locality suffice.
\end{lemma}
\begin{proof}
We can use the following Hamiltonian on $2m$ qubits labelled $0,\dots,2m-1$:
\begin{equation}
-\Delta\, \tfrac{1-Z_{2m-1}}{2}
+\sum_{k=0}^{m-1} H\otimes \tfrac12\bigl(X_{2k}X_{2k+1}+Y_{2k}Y_{2k+1}\bigr)
+\sum_{k=0}^{m-2} \tfrac12\bigl(X_{2k+1}X_{2k+2}+Y_{2k+1}Y_{2k+2}\bigr)
\label{eq:local-form}
\end{equation}
Three observations turn this into the claim. First, every term
of~\eqref{eq:local-form} commutes with the excitation number
$\mathcal N:=\sum_{k}\tfrac{1-Z_k}{2}$, since $\tfrac12(X_aX_b+Y_aY_b)$ does and
$Z_{2m-1}$ does, so the one-excitation sector
$\mathcal S:=\{\mathcal N=1\}$ is invariant. Second,
$\ket{\text{site }i}\mapsto\ket{0\cdots01_i0\cdots0}$, $i=0,\dots,2m-1$, is a
unitary from $\mathbb C^{2m}$ onto $\mathcal S$, and on $\mathcal S$ one has
$\tfrac12(X_aX_b+Y_aY_b)=\ket{0_a1_b}\!\bra{1_a0_b}+\mathrm{h.c.}$, a hopping of
amplitude exactly $1$ between sites $a$ and $b$, while
$\tfrac{1-Z_{2m-1}}{2}=\ket{1}\!\bra{1}_{2m-1}$ is the projector onto the last
site. Under that unitary~\eqref{eq:local-form} is therefore exactly
\eqref{eq:even-gadget-Hamiltonian-matrix}: $x$-bonds on the pairs
$(2k,2k+1)$ carrying $H$, unit bonds on the pairs $(2k+1,2k+2)$, and $-\Delta$ on
the last site. Third, if $H$ is $\varkappa$-local then each summand
$H\otimes\tfrac12(X_{2j}X_{2j+1}+Y_{2j}Y_{2j+1})$ is $(\varkappa+2)$-local, and the
other terms are at most $2$-local, which is the locality claim.

For $m=1$ the matrix $E_2(x)=\bigl(\begin{smallmatrix}0&x\\x&-\Delta\end{smallmatrix}\bigr)$
is an operator on a single qubit, $E_2(x)=x\,X-\Delta\ket1\!\bra1$, so
$E_2(H)=H\otimes X-\Delta\,I\otimes\ket1\!\bra1$, which is $(\varkappa+1)$-local.
\end{proof}

\begin{remark}[Odd gadgets]\label{rem:odd}
Ending the chain on a unit bond instead, at an odd number of sites, and
replacing the boundary detuning by a term proportional to the input, gives a
companion family $O_{2n+1}(x)$ whose branch through the origin is an \emph{odd}
function of $x$, of leading order $2n+1$, and it supplies the triangular basis for
odd targets. The interior of the chain is unchanged, so everything in
Sections~\ref{sec:gadget}--\ref{sec:synthesis} transposes to it with two
modifications, which Appendix~\ref{app:odd} carries out.
\end{remark}

\subsection{The characteristic polynomial}\label{sec:chareq}

We now compute the characteristic polynomial of $E_{2m}(x)$ in closed form. The
detuning only touches the last site, so we first compute the determinants of the
chain without it, which obey a two-step recurrence, and add the boundary at the
end. That recurrence is generated by a $2\times2$ transfer matrix whose powers
are governed by Chebyshev polynomials, and this is what makes the eigenvalue
analytically computable.

Let $\Theta_k(\lambda)=\det(E_k^{\rm int}(x)-\lambda I)$ where $E_k^{\rm int}(x)$ 
is the $k\times k$ leading principal submatrix of $E_{2m}(x)$ (so no non-zero diagonal entries). 
By convention $\Theta_0=1$ and $\Theta_{-1}=0$.

Expanding the determinant along the last row gives two alternating relations ($x$-bond steps couple with off-diagonal $x$, unit-bond steps with $1$):
\begin{equation}
\begin{aligned}
  \Theta_{2j}&= -\lambda \Theta_{2j-1} - x^2 \Theta_{2j-2},
  \qquad j\ge1 \quad\text{($x$-bond)},\\
  \Theta_{2j-1}&= -\lambda \Theta_{2j-2} - \Theta_{2j-3},
  \qquad j\ge1 \quad\text{(unit bond)}.
\end{aligned}
\label{eq:recurrence}
\end{equation}
Combining successive steps:

\begin{lemma}[Transfer matrix]\label{lem:transfer}
\begin{equation}
  \begin{pmatrix}\Theta_{2j}\\\Theta_{2j-1}\end{pmatrix}
  = T
  \begin{pmatrix}\Theta_{2j-2}\\\Theta_{2j-3}\end{pmatrix},
  \quad j\ge1,\qquad
  T=\begin{pmatrix}\lambda^2-x^2 & \lambda \\ -\lambda & -1\end{pmatrix},
\label{eq:transfer}
\end{equation}
with $\Theta_0=1$, $\Theta_{-1}=0$. Hence
\begin{equation}
  \begin{pmatrix}\Theta_{2m}\\\Theta_{2m-1}\end{pmatrix}
  = T^m\begin{pmatrix}1\\0\end{pmatrix}.
\label{eq:Tm_int}
\end{equation}
\end{lemma}
\deferproof{app:transfer}
Computing $\Theta_{2m}$ thus reduces to computing $T^{m}$, and for a $2\times2$
matrix the powers are controlled by the trace and the determinant. The transfer 
matrix $T$ has $\det(T)=x^2$ and $\Tr(T)=\lambda^2-x^2-1=:\tau$. Its eigenvalues 
are the two roots of $\mu^2-\tau\mu+x^2=0$,
\begin{equation}
  \mu_{\pm} = \frac{\tau\pm\sqrt{\tau^2-4x^2}}{2},
  \qquad \mu_+\mu_-=x^2, \quad \mu_++\mu_-=\tau,
\label{eq:mu}
\end{equation}
which are distinct when $\tau^2\neq4x^2$. By the Cayley--Hamilton theorem 
$T^{2}=\tau T-x^{2}I$, so every power of $T$ is a
combination of $T$ and $I$, with coefficients that obey a scalar two-term
recurrence.
\begin{lemma}[Chebyshev identity]\label{lem:ch} Define $S_k$ by the linear recurrence
\begin{equation}
  S_k = \tau S_{k-1} - x^2 S_{k-2}, \quad S_0=0, S_1=1.
\label{eq:Sk_rec}
\end{equation}
Then $S_k=(\mu_+^k-\mu_-^k)/(\mu_+-\mu_-)$ and
\begin{equation}
  T^m = S_m\,T - x^2 S_{m-1}\,I.
\label{eq:CH}
\end{equation}
\end{lemma}
\deferproof{app:ch}
The whole dependence on the length of the chain is therefore carried by the
scalar sequence $S_m$. It remains to add the detuned boundary site, which only
modifies the last row of the determinant.

\begin{proposition}[Eigenvalue equation]\label{prop:eig_eq}
  The characteristic polynomial of $E_{2m}(x)$ is
  \begin{equation}
    \chi_{2m}(\lambda)
    := \det(E_{2m}(x)-\lambda I)
    = S_m\bigl(\lambda^2+\Delta\lambda-x^2\bigr) - x^2 S_{m-1}.
  \label{eq:char_poly}
  \end{equation}
  Setting $\chi_{2m}(\lambda)=0$ gives the reduced eigenvalue equation
  \begin{equation}
    \lambda^2+\Delta\lambda = x^2\,\bigl(\underbrace{1+\frac{S_{m-1}}{S_m}}_{=R_m(\lambda,x)}\bigr).
  \label{eq:eig_reduced}
  \end{equation}
\end{proposition}
\deferproof{app:eig_eq}
The reduced equation has $S_m$ as a denominator, so the zeros of $S_m$ are the
poles that the analysis has to avoid. The
recurrence~\eqref{eq:Sk_rec} is the Chebyshev recurrence, which makes these zeros
explicit. Since $\tau=\lambda^{2}-x^{2}-1$ depends on $\lambda$ only through
$\lambda^{2}$, we work throughout in the variable  

\begin{equation}
w:=\lambda^{2},
\qquad\text{so that}\qquad
\tau=w-x^{2}-1 .
\label{eq:wvar}
\end{equation}
\begin{lemma}[Structure of $S_m$]\label{lem:poles}
Fix $x\ne0$ and regard $S_k$ as a polynomial in $w$ and $x$ through~\eqref{eq:wvar}.
Write $\vartheta_j:=j\pi/m$. Then:
\begin{enumerate}[label=(\roman*)]
\item $S_m=x^{m-1}\,U_{m-1}\!\bigl(\tfrac{\tau}{2x}\bigr)$, where $U_n$ 
is the Chebyshev polynomial of the second kind. In particular $S_m$ is 
monic of degree $m-1$ in $w$, and
  \begin{equation}
  S_m=\prod_{j=1}^{m-1}(w-w_j),
  \qquad
  w_j = 1+x^{2}+2x\cos\vartheta_j=(x-\omega_j)(x-\bar\omega_j),
  \qquad \omega_j:=-e^{\,i \vartheta_j} .
  \label{eq:poles}
  \end{equation}
\item Every zero satisfies
      \begin{equation}
      |w_j| \ge (1-|x|)\,\max\bigl\{\,1-|x|,\ \sin\vartheta_j\,\bigr\},
      \label{eq:polebound2}
      \end{equation}
      and in particular $|w_j|\ge(1-|x|)^{2}$.
      Consequently $S_m$ has no zero in the domain
      \begin{equation}
      \mathcal D := \bigl\{(\lambda,x):\ |\lambda|<1-|x|\bigr\},
      \label{eq:Ddomain}
      \end{equation}
      so $R_m=1+S_{m-1}/S_m$ is defined and analytic there.
\item On $\mathcal D$ the ratio has an explicit positive partial-fraction
      expansion,
      \begin{equation}
      \frac{S_{m-1}}{S_m} = \sum_{j=1}^{m-1}\frac{2\kappa_j}{w-w_j},
      \qquad \kappa_j=\frac{\sin^{2}\vartheta_j}{m} > 0,
      \qquad \sum_{j=1}^{m-1}\kappa_j=\tfrac12 .
      \label{eq:partialfrac}
      \end{equation}
\end{enumerate}
\end{lemma}
\deferproof{app:poles}
Lemma~\ref{lem:poles}(ii) shows in particular that the division
in~\eqref{eq:eig_reduced} is legitimate on all of $\mathcal D$, where the reduced
equation reads $\lambda(\lambda+\Delta)=x^{2}R_m(\lambda,x)$. For real arguments,
the partial-fraction form~\eqref{eq:partialfrac} bounds $R_m$ tightly enough to
bracket the target eigenvalue $\lambda_0^{(m)}$.

\begin{corollary}[The value at $\lambda=0$, and monotonicity]\label{cor:threefold}
Let $|x|<1$ and let $w=\lambda^{2}$ with $(\lambda,x)\in\mathcal D$.
\begin{enumerate}[label=(\roman*)]
\item \emph{(Value at $\lambda=0$.)} $w=0$ gives
      $\frac{S_{m-1}}{S_m}=-\sum_j\frac{2\kappa_j}{w_j}$, and evaluating the
      recurrence directly,
      \begin{equation}
      x^{2}R_m(0,x) = M(x) := \frac{x^{2m}(1-x^{2})}{1-x^{2m}} .
      \label{eq:Mvalue}
      \end{equation}
\item \emph{(Monotonicity.)} For real $x$ and $0\le w<(1-|x|)^{2}$ each
      term $\frac{2\kappa_j}{w-w_j}$ is negative and strictly decreasing in
      $w$, so $w\mapsto R_m$ is strictly decreasing, and
      \begin{equation}
      x^{2}R_m(\lambda,x) \le M(x)
      \qquad\text{with}\qquad
      M-x^{2}R_m  =  x^{2}\,w\sum_{j=1}^{m-1}\frac{2\kappa_j}{w_j\,(w_j-w)}  \ge 0 .
      \label{eq:monotone}
      \end{equation}
\end{enumerate}
\end{corollary}

Finally, the closed form~\eqref{eq:char_poly} contains $x$ only through $x^{2}$,
which gives the parity of the target eigenvalue $\lambda_0^{(m)}$ that we will use 
to restrict to $x\ge0$ and to show that the branch is even.

\begin{proposition}[Parity]\label{prop:parity}
The characteristic polynomial $\chi_{2m}$ depends on $x$ only through $x^{2}$.
Consequently the ordered eigenvalues $\lambda_1(x)\le\cdots\le\lambda_{2m}(x)$ of
$E_{2m}(x)$, $x$ real, are even functions of $x$, and so is any branch singled
out by a condition invariant under $x\mapsto-x$.
\end{proposition}
\deferproof{app:parity}

\section{The target eigenvalue}\label{sec:branch}
This section supplies the second step of the argument sketched in
Section~\ref{sec:intro:gadget}: the eigenvalue $\lambda_0^{(m)}(x)$ of $E_{2m}(x)$
that vanishes at $x=0$ is a single well-defined analytic function of the input,
with every Taylor coefficient controlled, uniformly in the chain length. The
route is: bracket the eigenvalue on the real axis (Section~\ref{sec:bracket}),
read off the first Taylor coefficients from the bracket
(Section~\ref{sec:series}), bound the branch on a complex disc of radius $r$ by one application
of Rouch\'e's theorem (Section~\ref{sec:radius}), and solve the resulting
condition on the detuning in closed form to achieve any $r<1$ (Section~\ref{sec:summary}).

As in Theorem~A, we index the Taylor coefficients of the branch from its leading
order,
\[
\lambda_0^{(m)}(x)=\sum_{k\ge0}c^{(m)}_k\,x^{2m+2k},
\]
so that $c^{(m)}_k$ is the coefficient of $x^{2m+2k}$ and $k$ counts the orders
beyond the leading one. The rest of the paper uses three results of this
section. Theorem~\ref{thm:series} computes the first coefficients exactly:
$c^{(m)}_0=1/\Delta\neq0$ makes the synthesis of polynomials (Section~\ref{sec:synthesis})
triangular, and the run of zeros $c^{(m)}_k=0$ for $1<k<m$ makes it sparse. The
other two hold for a target input radius $r\in(0,1)$ once the detuning exceeds a
threshold $\Delta_\star(r)$, given in closed form in~\eqref{eq:summary}. For
every $x$ in the complex disc $\{|x|<r\}$ the branch is then a simple eigenvalue
of $E_{2m}(x)$, and the only one in a disc around the origin
(Theorem~\ref{thm:onestep}, and Proposition~\ref{prop:m1} for $m=1$), so its
eigenvector, used in Theorem~\ref{thm:realisation}, is well defined up to a
phase. Its coefficients then obey
\[
\bigl|c^{(m)}_k\bigr|\ \le\ \frac{C_r}{\Delta\,r^{2k}}
\quad\text{for every } m\ge1,\ k\ge0,
\qquad
C_r=\frac{2}{1-r^{2}} ,
\]
which is the main bound needed for the synthesis of polynomials in Section~\ref{sec:synthesis}. 

\subsection{A two-sided bracket on the real axis}\label{sec:bracket}
On the real axis the eigenvalue is now trapped between two explicit bounds.
Writing the
eigenvalue equation of Proposition~\ref{prop:eig_eq} as
\begin{equation}
\lambda\,(\lambda+\Delta) = x^{2}R_m(\lambda,x),
\label{eq:eig_Rm}
\end{equation}
Corollary~\ref{cor:threefold} says that the right-hand side equals
$M(x)$ at $\lambda=0$ and only decreases as $\lambda^{2}$ grows. That
alone gives the upper bound, and quantifying the decrease gives the lower
one.
\begin{proposition}[Two-sided bound]\label{prop:bracket}
Let $0<x<1$, put $W:=(1-x)^{2}$ and
$M=M(x)=\frac{x^{2m}(1-x^{2})}{1-x^{2m}}$, and suppose
$\bigl(\tfrac M\Delta\bigr)^{2}<W$. Then $E_{2m}(x)$ has an eigenvalue
$\lambda_0^{(m)}(x)$ with
\begin{equation}
\frac{M}{\Delta}\left(1-\frac{M}{\Delta^{2}}\Bigl(1+\frac{x^{2}}{W\bigl(W-(M/\Delta)^{2}\bigr)}\Bigr)\right) \le \lambda_0^{(m)}(x) \le \frac{M}{\Delta} .
\label{eq:bracket}
\end{equation}
In particular $\lambda_0^{(m)}(x)=\frac{M}{\Delta}\bigl(1+O(x^{2m})\bigr)$ as $x\to0$.
\end{proposition}
\deferproof{app:bracket}

The eigenvalue is also well separated from the rest of the spectrum of the chain.
Write $E_{2m}(x)=E_{2m}(0)+V$, where $V$ carries the $x$-bonds. These bonds are
disjoint, so $\|V\|=|x|$. For $\Delta>1$ the ordered eigenvalues of
$E_{2m}(0)$ are $-\Delta$, then $-1$ with multiplicity $m-1$, then $0$, then
$+1$ with multiplicity $m-1$, and Weyl's
inequality moves each of them by at most
$|x|$. For $x$ small enough, one eigenvalue of $E_{2m}(x)$ therefore lies
in $(-1+|x|,\,1-|x|)$, and since Propositions~\ref{prop:parity}
and~\ref{prop:bracket} give $0\le\lambda_0^{(m)}(x)\le M/\Delta<\tfrac14$ there,
that eigenvalue is the branch. Every other eigenvalue lies at distance at least
$1-|x|-\lambda_0^{(m)}(x)$ from it.

\subsection{The Taylor series}\label{sec:series}

The bracket of Section~\ref{sec:bracket} has relative width $O(x^{2m})$, so it
already pins down every coefficient of the branch below order $x^{4m}$. This
subsection reads them off, and the run of zeros it finds is what makes the
triangular basis sparse.

\begin{theorem}[Even-series structure]\label{thm:series}
The branch $\lambda_0^{(m)}$ is analytic on a ball around $x=0$ and admits there the
convergent expansion
\begin{equation}
\lambda_0^{(m)}(x)=\sum_{k\ge 0} c_k^{(m)} x^{2m+2k},
\qquad c_0^{(m)}=\frac1\Delta .
\end{equation}
If moreover $m\ge2$, then
\[
c_1^{(m)}=-\frac1\Delta,
\qquad
c_k^{(m)}=0 \quad\text{for } 1<k<m,
\qquad
c_m^{(m)}=\frac1\Delta-\frac1{\Delta^3}.
\]
\end{theorem}
\deferproof{app:series}

\subsection{Coefficient bounds on the complex disc}\label{sec:radius}
The series of Theorem~\ref{thm:series} converges on some disc. The size of that disc controls every coefficient through Cauchy's estimate. We first state the Cauchy estimate in the form we can actually prove, then spend the rest of the section producing a radius to feed it.

\begin{lemma}[Cauchy's estimate]\label{thm:cauchy}
Suppose $\lambda_0^{(m)}$ is analytic on $\{|x|<r\}$ and put
\[
\Lambda_m(r):=\sup_{|x|<r}\bigl|\lambda_0^{(m)}(x)\bigr| .
\]
Then
\[
|c_k^{(m)}|  \le  \frac{\Lambda_m(r)}{r^{2m+2k}}\qquad (k\ge0).
\]
\end{lemma}
\begin{proof}
By Theorem~\ref{thm:series} the first $m$ Taylor coefficients of
$\lambda_0^{(m)}$ vanish, so
$h(x):=\lambda_0^{(m)}(x)/x^{2m}$ extends analytically to
$\{|x|<r\}$ with $h(x)=\sum_{k\ge0}c_k^{(m)}x^{2k}$. For $r'<r$ the
maximum principle gives
$\sup_{|x|\le r'}|h|=\sup_{|x|=r'}|h|\le \Lambda_m(r)/r'^{\,2m}$, and
Cauchy's inequality applied to $h$ on $|x|=r'$ gives
$|c_k^{(m)}|\le \Lambda_m(r)\,r'^{\,-2m-2k}$. Letting $r'\to r$ gives the claim.
\end{proof}

Everything now rests on one estimate: how large the eigenvalue can be on the complex
disc $|x|\le r$. We obtain it from the eigenvalue equation
$\lambda(\lambda+\Delta)=x^{2}R_m(\lambda,x)$ by one application of Rouch\'e's
theorem~\cite[Ch.~5]{ahlfors1979complex}, which says that if $f$ and $g$ are
analytic inside a contour $\partial K$ and $|g|<|f|$ on $\partial K$, then $f$
and $f+g$ have the same number of zeros inside $K$. All that is needed first is
to know how $x^{2}R_m$ depends on $\lambda$.

Lemma~\ref{lem:poles} already tells us everything we need about $R_m$ on the
complex disc: it is analytic on $\mathcal D=\{|\lambda|<1-|x|\}$, its value at
$\lambda=0$ is $M(x)/x^{2}$, and its $\lambda$-dependence is a sum of simple
poles at the explicitly known $w_j$. Subtracting the partial-fraction expansion~\eqref{eq:partialfrac} from its value
at $\lambda=0$ therefore bounds $x^{2}R_m-M(x)$ term by term. The only danger is
a pole $w_j$ close to the domain, which by~\eqref{eq:polebound2} can happen only
when $\sin\vartheta_j$ is of order $1-r$ or smaller, and then the residue
$\kappa_j=\sin^{2}\vartheta_j/m$ is small too. This compensation keeps the
prefactor in~\eqref{eq:lipschitz} of order one as $r\to1$. 
\begin{lemma}[$x^{2}R_m$ equals $M(x)$ up to $O(\lambda^{2})$]\label{lem:Rflat}
Let $0<|x|\le r<1$ and $|\lambda|<1-r$, and let
$M(x)=\dfrac{x^{2m}(1-x^{2})}{1-x^{2m}}$. Then
\begin{equation}
\bigl|\,x^{2}R_m(\lambda,x)-M(x)\,\bigr|
 \le \frac{2\,r^{2}\,|\lambda|^{2}}{(1-r)^{2}-|\lambda|^{2}}
\label{eq:lipschitz}
\end{equation}
for every $m\ge1$. Consequently, with
\begin{equation}
a_m:=\frac{(1+r^{2})\,r^{2m}}{1-r^{2m}} \ge \sup_{|x|\le r}|M(x)| ,
\label{eq:ab}
\end{equation}
we have $|x^{2}R_m|\le a_m+\frac{2r^{2}|\lambda|^{2}}{(1-r)^{2}-|\lambda|^{2}}$ there.
\end{lemma}
\deferproof{app:Rflat}

Lemma~\ref{lem:Rflat} is the complex counterpart of~\eqref{eq:monotone}: the
identity behind both holds off the real axis too, where it loses its sign but
keeps a bound on its modulus. It says that on a small disc $|\lambda|\le\rho$ the
right-hand side of $\lambda(\lambda+\Delta)=x^{2}R_m$ is nearly constant, of
modulus at most $a_m+b(\rho)\rho^{2}$, with $b(\rho)$ as in~\eqref{eq:onestep_hyp} below. For large $\Delta$ the linear term
$\Delta\lambda$ then dominates the rest of the equation on the circle
$|\lambda|=\rho$, and Rouch\'e's theorem says that the equation has as many roots
inside as $\Delta\lambda$, namely one. On the real axis the bracket of
Section~\ref{sec:bracket} gives $\lambda_0^{(m)}\approx M/\Delta$, and
$|M|\le a_m$ on the whole disc, so the root lives at scale $a_m/\Delta$ and we
take a circle twice as large, $\rho=2a_m/\Delta$. Domination then reduces to the
two conditions~\eqref{eq:onestep_hyp}, and Lemma~\ref{thm:cauchy} turns the
resulting bound $|\lambda_0^{(m)}|\le\rho$ into a bound on every coefficient.
\begin{theorem}[Coefficient bound, application of Rouch\'e]\label{thm:onestep} 
Let $\Delta>0$, $0<r<1$ and $a_m$ as in~\eqref{eq:ab}. Set $\rho:=2a_m/\Delta$ and assume
\begin{equation}
\rho < 1-r
\qquad\text{and}\qquad
\rho^{2}\Bigl(1+b(\rho)\Bigr) \le a_m,
\qquad
b(\rho)=\frac{2r^{2}}{(1-r)^{2}-\rho^{2}} .
\label{eq:onestep_hyp}
\end{equation}
Then for every $|x|<r$ the characteristic polynomial of $E_{2m}(x)$ has exactly
one root in $\{|\lambda|<\rho\}$, and it is simple. That root is
$\lambda_0^{(m)}(x)$, which is therefore analytic on $\{|x|<r\}$, and
\begin{equation}
\sup_{|x|<r}\bigl|\lambda_0^{(m)}(x)\bigr| \le \rho = \frac{2a_m}{\Delta},
\qquad\text{and}\qquad
\bigl|c_k^{(m)}\bigr| \le \frac{C_{r,m}}{\Delta\,r^{2k}},
\qquad
C_{r,m}:=\frac{2(1+r^{2})}{1-r^{2m}} .
\label{eq:onestep_ck}
\end{equation}
\end{theorem}
\deferproof{app:onestep}
Theorem~\ref{thm:onestep} is weakest at $m=1$, where $a_1$ is largest. That case, however, can be treated separately: the chain has a single $x$-bond, $S_1=1$ and $S_0=0$, so $R_1=1$ and the eigenvalue equation is a quadratic.

\begin{proposition}[The case $m=1$ is exact]\label{prop:m1}
For every $x$,
\begin{equation}
\lambda_0^{(1)}(x)=\frac{-\Delta+\sqrt{\Delta^{2}+4x^{2}}}{2}
=\frac{2x^{2}}{\Delta+\sqrt{\Delta^{2}+4x^{2}}} ,
\label{eq:m1}
\end{equation}
whose only branch points are $x=\pm i\Delta/2$. Hence $\lambda_0^{(1)}$
is analytic on $|x|<\Delta/2$, and for every $r\le\Delta/2$
\begin{equation}
\sup_{|x|<r}\bigl|\lambda_0^{(1)}(x)\bigr|\le\frac{2r^{2}}{\Delta},
\qquad\text{so}\qquad
\bigl|c_k^{(1)}\bigr|\le\frac{2}{\Delta\,r^{2k}} .
\label{eq:m1_ck}
\end{equation}
\end{proposition}
\deferproof{app:m1}

\subsection{The working detuning \texorpdfstring{$\Delta_\star(r)$}{Delta*(r)}}\label{sec:summary}
For $m\ge2$ we return to Theorem~\ref{thm:onestep}. Its
hypothesis~\eqref{eq:onestep_hyp} constrains $\rho$, and because
$b(\rho)$ is a single simple pole in $\rho^{2}$ that constraint turns out to be 
one quadratic inequality that can be solved.
\begin{lemma}[The admissible contour radii]\label{lem:rho_interval}
Fix $0<r<1$ and $m\ge1$, let $a_m$ be as in~\eqref{eq:ab}, and set
\begin{equation}
\mathsf p:=(1-r)^{2}+2r^{2}+a_m,
\qquad
\mathsf q:=a_m(1-r)^{2},
\qquad
s_-:=\frac{\mathsf p-\sqrt{\mathsf p^{2}-4\mathsf q}}{2}=\frac{2\mathsf q}{\mathsf p+\sqrt{\mathsf p^{2}-4\mathsf q}} .
\label{eq:sminus}
\end{equation}
Then $\mathsf p^{2}>4\mathsf q$, $0<s_-<(1-r)^{2}$, and for $\rho\in(0,1-r)$
\begin{equation}
\rho^{2}\bigl(1+b(\rho)\bigr)\le a_m
\qquad\Longleftrightarrow\qquad
\rho\le\sqrt{s_-} .
\label{eq:rho_interval}
\end{equation}
\end{lemma}
\deferproof{app:rho}
This constraint on $\rho$ becomes a constraint on the detuning $\Delta$: for any fixed
target radius $r$ there is a detuning beyond which Theorem~\ref{thm:onestep} applies.
\begin{corollary}[Minimum $\Delta$ for a target radius]\label{cor:Dmin} Fix $m\ge1$ and a target $r\in(0,1)$. Theorem~\ref{thm:onestep} applies at radius $r$ as soon as
\begin{equation}
\Delta\ge\Delta_{\min}(m,r):=\frac{2a_m}{\sqrt{s_-}}
=\frac{\sqrt{2\,a_m\Bigl(\mathsf p+\sqrt{\mathsf p^{2}-4\mathsf q}\Bigr)}}{1-r} ,
\label{eq:Dmin}
\end{equation}
and this is exact for the choice $\rho=2a_m/\Delta$.\footnote{That choice is
the one that minimises $\Delta$. Writing $\rho=\mathsf t\,a_m/\Delta$ turns the
Rouch\'e condition into
$\Delta^{2}>\frac{\mathsf t^{2}}{\mathsf t-1}a_m\bigl(1+b(\rho)\bigr)$, and
$\mathsf t^{2}/(\mathsf t-1)$ is smallest at $\mathsf t=2$. Once $\Delta$ exceeds
$\Delta_{\min}$ one may take $\mathsf t$ nearer $1$, improving the constant
$C_{r,m}$ of~\eqref{eq:onestep_ck}, which enters the final error estimate only as
a prefactor, so we do not pursue it.} Since $\sqrt{\mathsf p^{2}-4\mathsf q}\le\mathsf p$, the simpler condition
\begin{equation}
\Delta\ge\frac{2}{1-r}\sqrt{a_m\bigl((1-r)^{2}+2r^{2}+a_m\bigr)}
\label{eq:Dmin_simple}
\end{equation}
suffices. Moreover $a_m$ is strictly decreasing in $m$, and $\Delta_{\min}(\cdot,r)$ with it, so $\Delta\ge\Delta_{\min}(m_0,r)$ certifies every $m\ge m_0$ at once: the limiting case is always the shortest chain in the basis. 
\end{corollary}
\deferproof{app:dmin}

This subsection delivers exactly one requirement. Fix a target input radius $r\in(0,1)$ and set
\begin{equation}
\Delta_\star(r) := \max\bigl\{\,2r,\ \Delta_{\min}(2,r)\,\bigr\},
\qquad
C_r := \frac{2}{1-r^{2}} .
\label{eq:summary}
\end{equation}
Then $\Delta\ge\Delta_\star(r)$ implies
\begin{equation}
\bigl|c_k^{(m)}\bigr| \le \frac{C_r}{\Delta\,r^{2k}}
\qquad\text{for every } m\ge1,\ k\ge0 .
\label{eq:summary_ck}
\end{equation}

Two cases. For $m=1$, Proposition~\ref{prop:m1} needs only $\Delta>2r$ and gives the constant $2\le C_r$. For $m\ge2$, Corollary~\ref{cor:Dmin} applies at every such $m$ once $\Delta\ge\Delta_{\min}(2,r)$, and the constant of~\eqref{eq:onestep_ck} obeys
\[
\frac{2(1+r^{2})}{1-r^{2m}} \le \frac{2(1+r^{2})}{1-r^{4}} = \frac{2}{1-r^{2}} = C_r ,
\]
the extreme case being again $m=2$. So the shortest chains fix both the detuning and the constant, $m=2$ is the
single worst case, and longer gadgets are strictly easier. (If the target polynomial has no term below degree $2q$ with $q\ge3$, the branch $\lambda_0^{(2)}$ is never used and $\Delta_{\min}(2,r)$ may be replaced throughout by the smaller
$\Delta_{\min}(q,r)$.)

We assume throughout that $\Delta>1$, which the separation of the branch in
Section~\ref{sec:bracket} uses. For
$r>\tfrac12$, the only regime used below, this costs nothing, since
$\Delta\ge\Delta_\star(r)$ then already implies $\Delta>1$.

Each of the two constraints in~\eqref{eq:summary} binds on part of the range of
$r$. The shortest chain needs $\Delta>2r$, because its branch
$\lambda_0^{(1)}$ has branch points at $x=\pm i\Delta/2$
(Proposition~\ref{prop:m1}), and this requirement is linear in $r$. The chains
with $m\ge2$ need $\Delta\ge\Delta_{\min}(2,r)$, which is quadratic in $r$ for
small radii: as $r\to0$ one has $a_2\sim r^{4}$ and
$\Delta_{\min}(2,r)\sim2r^{2}$. The two curves cross at $r_c\simeq0.4935$.
Below $r_c$ the branch point of the $m=1$ chain sets $\Delta_\star=2r$, and above
it the $m=2$ chain sets $\Delta_\star=\Delta_{\min}(2,r)$.

In the second regime the requirement grows like $(1-r)^{-2}$ as $r\to1$. Indeed
$a_2=r^{4}/(1-r^{2})\sim\frac1{2(1-r)}$, so in~\eqref{eq:sminus}
$\mathsf p\sim\frac1{2(1-r)}$ and $\mathsf q\sim\frac{1-r}{2}$, hence
$s_-\sim\mathsf q/\mathsf p\sim(1-r)^{2}$ and
$\Delta_{\min}(2,r)=2a_2/\sqrt{s_-}\sim(1-r)^{-2}$, that is,
\begin{equation}
\Delta_\star(r)\sim\frac1{(1-r)^{2}}\quad(r\to1),
\qquad\text{equivalently}\qquad
1-r\sim\Delta^{-1/2} .
\label{eq:sqrtlaw}
\end{equation}
Asymptotically, halving the distance from the target radius to $1$ costs a
factor four in detuning.

\section{Synthesis of an even polynomial}\label{sec:synthesis}

Section~\ref{sec:branch} controlled each chain separately: the eigenvalue
$\lambda_0^{(m)}$ starts at $x^{2m}$, and its coefficients obey a bound uniform in
$m$. This section combines chains of different lengths and completes the argument
sketched in Section~\ref{sec:intro:gadget} in two steps. First, since the leading
orders are distinct, the eigenvalues form a triangular family, and the weights
reproducing a target even polynomial follow by back-substitution
(Section~\ref{sec:triangular}). The work is to show that these weights do not grow
faster than the truncation tail shrinks, so that the synthesis converges for every
input norm below $r$ (Section~\ref{sec:convergence}). Finally,
Theorem~\ref{thm:realisation} turns the scalar combination into a single
Hamiltonian acting on the input.

\subsection{The triangular system}\label{sec:triangular}
Let a target even polynomial of degree $2\ell$ be expressed as
\begin{equation}
P(x) = \sum_{k=q}^\ell p_{2k} x^{2k}, \qquad p_{2q} \neq 0.
\end{equation}
The eigenvalues vanish at $x=0$, so a constant term is not synthesised: it is the
energy shift $p_0I$, added directly to the Hamiltonian, and we take $q\ge1$
throughout.
We seek synthesis coefficients $\alpha_q, \dots, \alpha_L$ (with $L \ge \ell$) such that the linear combination
\begin{equation}
P_L^e(x) := \sum_{j=q}^L \alpha_j \lambda_0^{(j)}(x)
\end{equation}
matches $P(x)$ up to degree $2L$. Expanding $P_L^e(x)$ into power series gives
\begin{equation}
P_L^e(x) = \sum_{j=q}^L \alpha_j \sum_{n=0}^\infty c_n^{(j)} x^{2j+2n} = \sum_{k=q}^\infty d_k x^{2k},
\end{equation}
where for each index $k \ge q$, the composite coefficient $d_k$ is given by
\begin{equation}\label{eq:dk_def}
d_k = \sum_{j=q}^{\min(k, L)} \alpha_j c_{k-j}^{(j)}.
\end{equation}
To match $P(x)$ up to degree $2\ell$ and eliminate the higher orders up to degree
$2L$, we impose
\begin{equation}
d_k = p_{2k}\quad (q\le k\le\ell),
\qquad
d_k = 0\quad (\ell<k\le L).
\label{eq:matching}
\end{equation}
The coefficients $d_k$ with $k>L$ are left free. They form the truncation tail,
which Section~\ref{sec:convergence} bounds.

We show that this linear system for the weights is triangular and hence uniquely solvable by back-substitution.
\begin{proposition}[Triangular system and back-substitution]\label{prop:triangular_solvability}
Let $\Delta>0$. The conditions~\eqref{eq:matching} have a unique solution
$\alpha_q,\dots,\alpha_L$, given by the recurrence
\begin{equation}\label{eq:alpha_rec}
\alpha_k = \Delta \Bigl( p_{2k} - \sum_{j=q}^{k-1} \alpha_j c_{k-j}^{(j)} \Bigr),
\qquad k = q, \dots, L,
\end{equation}
where $p_{2k} := 0$ for $k > \ell$.
\end{proposition}
\begin{proof}
For $q\le k\le L$, \eqref{eq:dk_def} reads $d_k=\sum_{j=q}^{k}\alpha_jc^{(j)}_{k-j}$.
The $k$-th condition in~\eqref{eq:matching} therefore involves only
$\alpha_q,\dots,\alpha_k$, and involves $\alpha_k$ only through
$\alpha_kc^{(k)}_0=\alpha_k/\Delta$ by Theorem~\ref{thm:series}. The conditions
form a lower-triangular linear system with diagonal entries $1/\Delta\neq0$, so
they have a unique solution, and solving the $k$-th one for $\alpha_k$
gives~\eqref{eq:alpha_rec}.
\end{proof}

Theorem~\ref{thm:series} gives more than the diagonal. For $j\ge2$ it fixes
$c^{(j)}_1=-1/\Delta$ and $c^{(j)}_n=0$ for $1<n<j$, so most terms
of~\eqref{eq:alpha_rec} vanish and one is explicit. This sparsity is what keeps
the weights under control in Section~\ref{sec:convergence}.

\begin{proposition}[Sparsity of the back-substitution]\label{prop:sparsity}
In the recurrence~\eqref{eq:alpha_rec} the coefficient $c_{k-j}^{(j)}$ vanishes for every $j\ge2$ with $\tfrac k2<j<k-1$. Consequently, for $k\ge\max\{3,q+1\}$,
\begin{equation}
\alpha_k = \Delta\,p_{2k}+\alpha_{k-1} -\Delta\!\!\sum_{q\le j\le\lfloor k/2\rfloor}\!\!\alpha_j\,c_{k-j}^{(j)} ,
\label{eq:alpha_sparse}
\end{equation}
while for $k=2$ (which occurs only when $q=1$) the unreduced recurrence gives
\begin{equation}
\alpha_2 = \Delta\bigl(p_4-\alpha_1c_1^{(1)}\bigr) =\Delta p_4+\frac{\alpha_1}{\Delta^{2}} .
\label{eq:alpha_k2}
\end{equation}
\end{proposition}
\deferproof{app:sparsity}

\begin{remark}[Computing the coefficients]\label{rem:compute}
The numbers $c^{(j)}_n$ are not available in closed form, but they are the Taylor
coefficients of an algebraic function. In the variable $x^{2}$ the eigenvalue is
the unique solution in $x^2\mathbb R[[x^2]]$ of
$S_m\bigl(\hat\lambda^{2}+\Delta\hat\lambda-x^2\bigr)-x^2S_{m-1}=0$, and Newton's
method in a ring of truncated power series~\cite{gathen2013modern}
doubles the number of correct coefficients at each step. Eliminating the chain
site by site from the far end writes that equation as a finite continued
fraction whose convergents are the $S_j$, which is the classical expansion of
the Green's function of a Jacobi matrix~\cite{teschl2000jacobi}. Evaluating $S_m$ through
the Binet form $S_m=(\mu_+^{m}-\mu_-^{m})/(\mu_+-\mu_-)$ of Lemma~\ref{lem:ch}
costs $O(1)$ series operations independently of $m$, so all the data the
back-substitution~\eqref{eq:alpha_rec} needs up to order $L$ cost
$O\bigl(L\,\mathsf M(L)\bigr)$ arithmetic operations, where $\mathsf M(p)$ is the
cost of multiplying two series truncated at order
$p$~\cite{brent2010arithmetic}. That is $O(L^{2}\log L)$ with fast
multiplication, and optimal up to the logarithm, the output having
$\Theta(L^{2})$ coefficients. The arithmetic must be exact or of high precision.
\end{remark}

\subsection{Convergence of the synthesis}\label{sec:convergence}
The back-substitution always has a solution, but it is useful only if the tail it
leaves behind, $\mathcal E_L=\sum_{k>L}d_kx^{2k}$, is small. Two effects compete
there. By~\eqref{eq:dk_def} each $d_k$ combines eigenvalue coefficients
$c^{(j)}_{k-j}$, which are small, weighted by the synthesis coefficients
$\alpha_j$, which the recurrence~\eqref{eq:alpha_rec} may amplify. We settle this
competition along the chain
\[
\text{bound } |c_k^{(j)}|
\longrightarrow
\text{bound } |\alpha_j|
\longrightarrow
\text{bound } \|\mathcal E_L(H)\|=\|P^e_L(H)-P(H)\| ,
\]
which answers the two questions the construction is for: how many chains are
needed, and for which input norms $\|H\|$ the error vanishes at all. The answer to
the second is every $\|H\|<r$, with no loss beyond the analyticity radius of
Section~\ref{sec:branch}. The only estimate on the gadget that enters is the
coefficient bound of Section~\ref{sec:branch}, restated for the chains used here,
\begin{equation}\label{eq:cauchy_bound_j}
|c_k^{(j)}| \le \frac{C_r}{\Delta\, r^{2k}} \qquad \forall j \ge q,\ k \ge 0 ,
\end{equation}
which holds as soon as $\Delta\ge\Delta_\star(r)$ by~\eqref{eq:summary_ck}, with
$C_r=2/(1-r^{2})$.

\paragraph{Growth of the weights.} The tail does not consume the weights one by one but
through the rescaled partial sums $\sum_{j\le L}r^{2j}|\alpha_j|$, since for $k>L$
\eqref{eq:dk_def} and~\eqref{eq:cauchy_bound_j} give
\[
|d_k|\le\frac{C_r}{\Delta\,r^{2k}}\sum_{j=q}^{L}r^{2j}|\alpha_j| .
\]
Bounding these sums through the full recurrence~\eqref{eq:alpha_rec} is too crude. It gives only
$r^{2k}|\alpha_k|\le r^{2k}\Delta|p_{2k}|+C_r\sum_{j<k}r^{2j}|\alpha_j|$, hence
growth like $(1+C_r)^{k}$ and convergence only for $\|H\|<r/\sqrt{1+C_r}$, a loss
of at least a factor $\sqrt3$ since $C_r\ge2$. The sparse form~\eqref{eq:alpha_sparse} avoids this loss: each $\alpha_k$ now
depends on its predecessor with coefficient exactly one, and otherwise only on
the first half of the weights, $j\le\lfloor k/2\rfloor$. As a result the partial
sums grow only quasi-polynomially in $L$, slowly enough to be beaten by the
geometric decay of the tail.

\begin{proposition}[Quasi-polynomial growth of the partial sums]\label{prop:alpha_growth}
Put $M_P:=\max_{k\in\{q,\dots,\ell\}}|p_{2k}|$ and, for $k\ge q$,
\[
A_k:=r^{2k}|\alpha_k| ,
\qquad
\mathcal A_k:=\sum_{j=q}^{k}A_j ,
\]
with the convention $\mathcal A_k:=0$ for $k<q$. Define
\begin{equation}
\Phi_k:=1\quad(k\le q),
\qquad
\Phi_k:=\frac{1+C_r\,(k-q)}{1-r^{2}}\Phi_{\lfloor k/2\rfloor}\quad(k>q).
\label{eq:Phi}
\end{equation}
Directly from~\eqref{eq:Phi}, $\Phi$ is non-decreasing with $\Phi_k\ge1$ for every $k$. Then
\begin{equation}
\mathcal A_L \le \frac{r^{2q}\,\Delta M_P}{1-r^{2}} \Phi_L
\qquad (q\le L).
\label{eq:Asum_bound}
\end{equation}
Moreover,
\begin{equation}
\ln\Phi_L \le \frac{\ln^{2}\bigl(\frac{\sqrt2\,C_r}{1-r^{2}}\,L\bigr)}{2\ln2} ,
\label{eq:Phi_growth}
\end{equation}
so the growth is quasi-polynomial and $\Phi_L^{1/L}\to1$.
\end{proposition}
\deferproof{app:alpha}

\paragraph{The tail.} Inserting~\eqref{eq:cauchy_bound_j} and~\eqref{eq:Asum_bound}
into the tail gives the synthesis error on a scalar input.

\begin{proposition}[Tail estimate]\label{prop:tail_estimate}
Let $x_0<r$ and $\gamma:=x_0^{2}/r^{2}<1$. The truncation remainder $\mathcal E_L(x):=P_L^e(x)-P(x)=\sum_{k>L}d_kx^{2k}$ satisfies, for every $|x|\le x_0$,
\begin{equation}
|\mathcal E_L(x)|  \le  \frac{C_r\,M_P\,r^{2q}\,\Phi_L}{(1-r^{2})(1-\gamma)} \gamma^{\,L+1} .
\label{eq:tail_final}
\end{equation}
\end{proposition}
\deferproof{app:tail}

Because $\gamma<1$ is fixed while $\Phi_L$ grows only quasi-polynomially, the geometric factor wins: the remainder falls off geometrically in the number of gadgets, at a rate set by nothing but the margin between the input norm and the analyticity radius.

\begin{corollary}[Geometric convergence]\label{cor:L_needed}
Let $H$ be Hermitian with $0<\|H\|<r$ and put $\gamma:=\|H\|^{2}/r^{2}<1$. Then, for every $L\ge q$,
\begin{equation}
\bigl\|P(H)-P^{e}_L(H)\bigr\|
 \le 
\frac{C_r\,M_P\,r^{2q}}{(1-r^{2})(1-\gamma)} \Phi_L \gamma^{\,L+1} ,
\label{eq:tail_op}
\end{equation}
and consequently
\begin{equation}
\limsup_{L\to\infty} \bigl\|P(H)-P^{e}_L(H)\bigr\|^{1/L}
 \le \gamma = \frac{\|H\|^{2}}{r^{2}} .
\label{eq:geom_rate}
\end{equation}
Reaching accuracy $\varepsilon$ therefore costs about
$\ln(M_P/\varepsilon)\big/\ln(1/\gamma)$ gadgets, and the longest chain used has
$2L$ sites. Since $L$ must also be at least $\ell$ for the coefficients of $P$ to
be matched at all, a target of degree $2\ell$ approximated to accuracy
$\varepsilon$ is reached with $L=O\bigl(\max\{\ell,\ln(M_P/\varepsilon)\}\bigr)$
chains, the constant depending only on $r$ and $\gamma$, hence
$O\bigl(\ell^{2}+\ln^{2}(M_P/\varepsilon)\bigr)$ ancilla qubits.
\end{corollary}

\begin{proof}
$\mathcal E_L=P^{e}_L-P$ is analytic on $|x|<r$ and $\spec H\subset[-\|H\|,\|H\|]$, so the spectral theorem gives $\|\mathcal E_L(H)\|=\max_{x\in\spec H}|\mathcal E_L(x)|$. Then~\eqref{eq:tail_op} is Proposition~\ref{prop:tail_estimate} applied with $x_0=\|H\|$. Taking $L$-th roots there, the constant prefactor tends to $1$, so does $\Phi_L^{1/L}$ by~\eqref{eq:Phi_growth}, and $\gamma^{(L+1)/L}\to\gamma$. Finally, using $r^{2q}\le1$, the right-hand side of~\eqref{eq:tail_op} is at most
$\varepsilon$ as soon as
$(L+1)\ln(1/\gamma)\ge\ln\frac{M_P}{\varepsilon}+\ln\frac{C_r}{(1-r^{2})(1-\gamma)}+\ln\Phi_L$,
and $\ln\Phi_L=o(L)$ by~\eqref{eq:Phi_growth}, which gives the stated $L$, with
$\ln(1/\gamma)=2\ln(r/\|H\|)$.
\end{proof}

\subsection{The realised Hamiltonian}\label{sec:realisation}

We now attach all chains to the same input $H$, each on its own set of ancillas
and with weight $\alpha_j$ (Figure~\ref{fig:realisation}). Because the chains are
disjoint they commute, so their branch eigenvalues add, and the sum carries
$P^{e}_L$ exactly.

\begin{figure}[ht]
\centering
\begin{tikzpicture}[
  every node/.style={font=\small},
  site/.style={circle,draw,thick,fill=white,inner sep=0pt,minimum size=4.5mm},
  sink/.style={circle,draw,thick,fill=black!15,inner sep=0pt,minimum size=4.5mm},
  xb/.style={very thick},
  ub/.style={thin}]
  \def\s{0.72}
  \node[draw,thick,rounded corners,minimum width=20mm,minimum height=24mm,fill=black!4]
       (sys) at (0,0) {};
  \node at (0,0.72) {$\mathcal H$};
  \node at (0,0.12) {$\ket{v_x}$};
  \node[font=\scriptsize] at (0,-0.48) {$H\ket{v_x}=x\ket{v_x}$};
  \foreach \i in {1,2,3}{\node[site] (p\i) at (2.3+\s*\i,1.1) {};}
  \node[sink] (p4) at (2.3+\s*4,1.1) {};
  \draw[xb] (p1)--(p2); \draw[ub] (p2)--(p3); \draw[xb] (p3)--(p4);
  \node[anchor=south,font=\small] at (2.3+\s*1.5,1.28) {$\alpha_q$};
  \node[anchor=west,font=\scriptsize] at (2.3+\s*4+0.32,1.1)
       {$\ket{\xi_q(x)}$,\ \ $\lambda_0^{(q)}(x)$};
  \node at (2.3+\s*2.5,0.22) {$\vdots$};
  \foreach \i in {1,...,5}{\node[site] (r\i) at (2.3+\s*\i,-0.9) {};}
  \node[sink] (r6) at (2.3+\s*6,-0.9) {};
  \draw[xb] (r1)--(r2); \draw[ub] (r2)--(r3); \draw[xb] (r3)--(r4);
  \draw[ub] (r4)--(r5); \draw[xb] (r5)--(r6);
  \node[anchor=south,font=\small] at (2.3+\s*1.5,-0.72) {$\alpha_L$};
  \node[anchor=west,font=\scriptsize] at (2.3+\s*6+0.32,-0.9)
       {$\ket{\xi_L(x)}$,\ \ $\lambda_0^{(L)}(x)$};
  \draw[->,thick] (sys.east) to[out=30,in=180] (p1.west);
  \draw[->,thick] (sys.east) to[out=-30,in=180] (r1.west);
  \node[font=\scriptsize,anchor=south,align=center] at (1.55,1.35)
       {thick bonds\\carry $H$};
  \draw[decorate,decoration={brace,mirror,amplitude=5pt}]
       (2.15,-1.55) -- (2.3+\s*6+0.35,-1.55)
       node[midway,below=6pt,font=\small]
       {$L-q+1$ disjoint chains, one excitation each, $O(L^{2})$ ancilla qubits};
\end{tikzpicture}

\vspace{-2mm}
\[
W:\ \ket{v_x}\ \longmapsto\ \ket{v_x}\otimes\bigotimes_{j=q}^{L}\ket{\xi_j(x)} ,
\qquad\qquad
\Sigma\,W \;=\; W\,P^{e}_L(H) .
\]
\vspace{-4mm}
\caption{The construction of Theorem~\ref{thm:realisation}. Every chain is attached to the same system register through its
$x$-bonds, so on an eigenvector $\ket{v_x}$ it reduces to the matrix $E_{2j}(x)$
of Figure~\ref{fig:chain} and contributes $\lambda_0^{(j)}(x)$. The chains
occupy disjoint sets of ancillas, so the summands of
$\Sigma=\sum_j\alpha_jE_{2j}(H)$ commute, their branch eigenvalues add, and $W$
intertwines $\Sigma$ with $P^{e}_L(H)$ exactly, which is~\eqref{eq:intertwine}.
The image of $W$ is invariant under $\Sigma$.}
\label{fig:realisation}
\end{figure}

\begin{theorem}[Realisation]\label{thm:realisation}
Fix $r\in(0,1)$ and $\Delta\ge\Delta_\star(r)$, let $H$ be Hermitian with
eigenpairs $(x,\ket{v_x})$ and $\|H\|<r$, let the $\alpha_j$ solve the
triangular system of Proposition~\ref{prop:triangular_solvability}, and put
\begin{equation}
\Sigma\;:=\;\sum_{j=q}^{L}\alpha_j\,E_{2j}(H)
\qquad\text{on}\qquad
\mathcal H\otimes\bigotimes_{j=q}^{L}\mathcal C_j ,
\label{eq:Sigma}
\end{equation}
each summand padded by the identity on the other chains, where
$\mathcal C_j\cong\mathbb C^{2j}$ is the one-excitation sector of chain $j$.
Then:
\begin{enumerate}[label=(\roman*)]
\item the summands commute pairwise, $[E_{2j}(H),E_{2j'}(H)]=0$, and $\Sigma$ is
      $(\varkappa+2)$-local whenever $H$ is $\varkappa$-local,
\item writing $\ket{\xi_j(x)}\in\mathcal C_j$ for the unit eigenvector of
      $E_{2j}(x)$ belonging to $\lambda_0^{(j)}(x)$, well defined up to a phase
      because that eigenvalue is simple, the map
      \begin{equation}
      W:\ \ket{v_x}\ \longmapsto\ \ket{v_x}\otimes\bigotimes_{j=q}^{L}\ket{\xi_j(x)}
      \label{eq:isometry}
      \end{equation}
      extends to an isometry
      $W:\mathcal H\to\mathcal H\otimes\bigotimes_j\mathcal C_j$ satisfying the
      exact intertwining relation
      \begin{equation}
      \Sigma\,W \;=\; W\,P_L^{e}(H) .
      \label{eq:intertwine}
      \end{equation}
\end{enumerate}
Consequently $e^{-i\Sigma t}W=W e^{-iP_L^{e}(H)t}$ for every $t\in\mathbb R$, and
hence
$\bigl\|e^{-i\Sigma t}W-We^{-iP(H)t}\bigr\|\le |t|\,\|P(H)-P_L^{e}(H)\|$, a
quantity controlled by Corollary~\ref{cor:L_needed}.
\end{theorem}

\begin{proof}
(i) By~\eqref{eq:local-form}, $E_{2j}(H)=H\otimes \mathsf A_j+I\otimes \mathsf B_j$
with $\mathsf A_j$ the sum of the $x$-bond hopping terms and $\mathsf B_j$ the
unit bonds together with the detuned site, both of them supported on chain $j$
alone. For $j\ne j'$ the four commutators
$[\mathsf A_j,\mathsf A_{j'}]$, $[\mathsf A_j,\mathsf B_{j'}]$,
$[\mathsf B_j,\mathsf A_{j'}]$, $[\mathsf B_j,\mathsf B_{j'}]$ vanish because
the operators act on disjoint tensor factors, and expanding
\[
[H\otimes \mathsf A_j+I\otimes \mathsf B_j,\ H\otimes \mathsf A_{j'}+I\otimes \mathsf B_{j'}]
= H^{2}\otimes[\mathsf A_j,\mathsf A_{j'}]
+H\otimes\bigl([\mathsf A_j,\mathsf B_{j'}]+[\mathsf B_j,\mathsf A_{j'}]\bigr)
+I\otimes[\mathsf B_j,\mathsf B_{j'}]
\]
gives $0$. Every term of~\eqref{eq:local-form} commutes with the excitation
number of its chain, so each $\mathcal C_j$ is invariant, and locality is
Lemma~\ref{lem:local}.

\emph{Simplicity of the branch.} Fix $j$ and $|x|<r$. On $\mathcal D$ one has
$\chi_{2j}=S_j\,F$ with $F(\lambda)=\lambda^{2}+\Delta\lambda-x^{2}R_j$ and
$S_j\neq0$ by Lemma~\ref{lem:poles}(ii), so the eigenvalues of $E_{2j}(x)$ lying
in $\mathcal D$ are exactly the zeros of $F$ there, with the same multiplicities.
For $j\ge2$, Theorem~\ref{thm:onestep} produces \emph{precisely one} zero of $F$
in $|\lambda|<\rho\subset\mathcal D$, so $\lambda_0^{(j)}(x)$ is a simple
eigenvalue of $E_{2j}(x)$. For $j=1$ the eigenvalue equation is the quadratic of
Proposition~\ref{prop:m1}, whose two roots differ by $\sqrt{\Delta^{2}+4x^{2}}>0$.
In both cases $\ket{\xi_j(x)}$ is determined up to a phase.

(ii) The vectors $\ket{v_x}$ are orthonormal and each $\ket{\xi_j(x)}$ is a unit
vector, so $W$ is an isometry. On
$\ket{v_x}\otimes\bigotimes_{j'}\ket{\xi_{j'}(x)}$ the operator
$H\otimes \mathsf A_j+I\otimes \mathsf B_j$ acts as
$x\mathsf A_j+\mathsf B_j=E_{2j}(x)$ on the factor $\mathcal C_j$ and as the
identity on the others, so that vector is an eigenvector of $E_{2j}(H)$ with
eigenvalue $\lambda_0^{(j)}(x)$, and of $\Sigma$ with eigenvalue
$\sum_j\alpha_j\lambda_0^{(j)}(x)=P_L^{e}(x)$. Since
$P_L^{e}(H)\ket{v_x}=P_L^{e}(x)\ket{v_x}$, both sides of~\eqref{eq:intertwine}
agree on every $\ket{v_x}$, hence everywhere. The evolution statement follows by
functional calculus, and the error estimate from
$\|e^{-iAt}-e^{-iBt}\|\le|t|\,\|A-B\|$ for commuting Hermitian $A,B$, here
$A=P_L^{e}(H)$ and $B=P(H)$, both functions of $H$.
\end{proof}

The construction has some costs that the per-chain estimates of
Section~\ref{sec:branch} do not see: the number of ancillas and the energy scale
$\|\Sigma\|$.
\begin{itemize}
  \item \emph{Ancillas.} The construction uses $\sum_{j=q}^{L}2j=O(L^{2})$ ancilla
  qubits carrying $L-q+1$ excitations, one per chain. With $L$ as in
  Corollary~\ref{cor:L_needed} this is
  $O\bigl(\ell^{2}+\ln^{2}(M_P/\varepsilon)\bigr)$ for a target of degree $2\ell$
  reached to accuracy $\varepsilon$.
  \item \emph{Energy scale.} Each chain has $\|E_{2j}(H)\|\le\Delta+2$, so
  $\|\Sigma\|\le(\Delta+2)\Lambda_L$, with $\Lambda_L:=\sum_{j=q}^{L}|\alpha_j|$
  the total weight of the synthesis.
\end{itemize}

Together, Theorem~\ref{thm:realisation} and Corollary~\ref{cor:L_needed} prove
Theorem~B, the odd part being treated in Appendix~\ref{app:odd}: every even
polynomial of an input with $\|H\|<r$ is approximated, with an error decaying
geometrically in the number of chains, by the encoded action of a single
time-independent $(\varkappa+2)$-local Hamiltonian. The price grows with the
target, through its degree and $\ln(M_P/\varepsilon)$ in the number of chains, and
through the weights $\Lambda_L$ in the energy scale. The next section turns to an
application, analog search, whose target is a filter isolating a marked state.

\section{Application: a filter for a marked state}\label{sec:filter}

Analog search~\cite{farhi1998analog,farhi2000adiabatic,RolandLocalAQC}
interpolates $H(s)=(1-s)H_0-s\ket m\!\bra m$ from a transverse field $H_0$. Its minimum
gap closes as $2^{-n/2}$~\cite[Sec.~4.2]{farhi2000adiabatic}, and with a rank-one driver the gap and
the optimal schedule are known in closed form~\cite{RolandLocalAQC,Braida2025unstructured},
but $\ket m\!\bra m$ is a rank-one projector on $2^{n}$ dimensions and no device
implements it. A device does implement a local $H_1$ whose ground state is
$\ket m$, and more generally one in which $\ket m$ is an isolated eigenvector.
Fix $\sigma>0$ and take
\begin{equation}
\spec(H_1)\subset[-\sigma,-\delta]\cup\{0\}\cup[\delta,\sigma] ,\qquad 0<\delta\le\sigma ,
\label{eq:H1spec}
\end{equation}
with $\ket m$ the unique $0$-eigenvector. For example, an Ising Hamiltonian
shifted by its ground energy and rescaled has a nonnegative spectrum of this form,
with $\ket m$ its ground state. Nothing below requires the spectrum to be
one-sided. The projector is then a
function of $H_1$: if $R(0)=1$ and $|R(x)|\le\eta$ for $\delta\le|x|\le\sigma$,
the spectral theorem gives
$\|R(H_1)-\ket m\!\bra m\|\le\eta$. This section builds such an $R$ as the encoded
action of a local Hamiltonian, and then puts it into the search.

\subsection{Why not the polynomial directly}\label{sec:filter:target}

The natural target for filtering is the Chebyshev filter, flat and equal to $1$ at the marked energy
and exponentially small on the rest of the spectrum. Among polynomials of a given degree equal to $1$ at the marked energy, it has
the smallest maximum on the rest of the spectrum~\cite{lin2020filtering}. 
This subsection fixes it and
its two parameters, the passband edge $d$ and the half-degree $\ell$.

\begin{definition}[Chebyshev filter]\label{def:filter}
For $0<d<\sigma$ and $\ell\ge1$ set
\begin{equation}
\phi(x):=-1+\frac{2\,(x^{2}-d^{2})}{\sigma^{2}-d^{2}} ,
\qquad
R_\ell(x):=\frac{T_\ell\bigl(\phi(x)\bigr)}{T_\ell\bigl(\phi(0)\bigr)} ,
\label{eq:filter}
\end{equation}
with $T_\ell$ the Chebyshev polynomial of the first
kind~\cite{mason2003chebyshev,trefethen2013atap}. Thus $\phi(0)=-\frac{\sigma^{2}+d^{2}}{\sigma^{2}-d^{2}}<-1$, $R_\ell$ is even of degree $2\ell$, and $R_\ell(0)=1$.
\end{definition}

\begin{lemma}\label{lem:filter}
Put $u:=d/\sigma\in(0,1)$. Then $|R_\ell(x)|\le1$ for $|x|\le\sigma$, and
\begin{equation}
|R_\ell(x)|\le\frac{1}{\cosh\bigl(2\ell \arctanh u\bigr)}\le2\mathrm e^{-2\ell\tfrac{d}{\sigma}} \qquad (d\le|x|\le\sigma).
\label{eq:stopband}
\end{equation}
Consequently, if $d\le\delta$ then $\bigl\|R_\ell(H_1)-\ket m\!\bra m\bigr\|\le2\mathrm e^{-2\ell u}$.
\end{lemma}

\begin{proof}
$\phi$ is increasing in $x^{2}$ and maps $|x|\in[d,\sigma]$ onto $[-1,1]$, so $\phi(x)\in[\phi(0),1]$ for $|x|\le\sigma$. On $[-1,1]$ one has $|T_\ell|\le1$, while on $(-\infty,-1]$, $|T_\ell(y)|=\cosh(\ell\arccosh|y|)$ is increasing in $|y|$. Since $|\phi(x)|\le|\phi(0)|$ there, in both regimes $|T_\ell(\phi(x))|\le|T_\ell(\phi(0))|$, which is $|R_\ell|\le1$. For $d\le|x|\le\sigma$ the numerator is at most $1$, giving the first bound in~\eqref{eq:stopband} once we identify
\[
\arccosh|\phi(0)|=\arccosh\frac{1+u^{2}}{1-u^{2}}=2\arctanh u ,
\]
which is the identity $\cosh(2\arctanh u)=\frac{1+u^{2}}{1-u^{2}}$. The second bound follows from $\cosh y\ge\tfrac12e^{y}$ and $\arctanh u\ge u$. Finally $R_\ell(H_1)-\ket m\!\bra m=\sum_{\delta\le|x|\le\sigma}R_\ell(x)\ket x\!\bra x$ and $R_\ell(0)=1$, so its norm is $\max_{\delta\le|x|\le\sigma}|R_\ell(x)|$.
\end{proof}

The branches vanish at $x=0$, so what the gadgets must synthesise is $R_\ell-1$, the constant being a free shift of energy. Section~\ref{sec:synthesis} can do this in one combination of gadgets, since the triangular system of Proposition~\ref{prop:triangular_solvability} is solvable for any target, but at a cost that is intrinsic to the polynomial.

\begin{lemma}[The coefficients of $R_\ell$ are exponentially large]\label{lem:MPbig}
Write $R_\ell(x)=\sum_{j=0}^{\ell}p_{2j}x^{2j}$. Then
\begin{equation}
M_P\ge|p_{2\ell}|=\frac{2^{\ell-1}}{|T_\ell(\phi(0))|} \left(\frac{2}{\sigma^{2}-d^{2}}\right)^{\!\ell} \ge\frac12\left(\frac{4\,e^{-2\arctanh u}}{\sigma^{2}-d^{2}}\right)^{\!\ell} .
\label{eq:MPbig}
\end{equation}
In particular $M_P$ grows exponentially in $\ell$ whenever
$\sigma^{2}-d^{2}<4\mathrm e^{-2\arctanh u}$. Since
$\mathrm e^{-2\arctanh u}=\frac{1-u}{1+u}=\frac{\sigma-d}{\sigma+d}$, that condition
is exactly $\sigma+d<2$, so it holds whenever $\sigma\le1$.
\end{lemma}

\begin{proof}
$T_\ell$ has leading coefficient $2^{\ell-1}$ and $\phi$ has leading coefficient $2/(\sigma^{2}-d^{2})$ in $x^{2}$, which gives $p_{2\ell}$ exactly. For the second inequality, $|T_\ell(\phi(0))|=\cosh(2\ell\arctanh u)\le \mathrm e^{2\ell\arctanh u}$.
\end{proof}

A one-pass synthesis of $R_\ell$ therefore needs $L=\Theta(\ell)$ chains,
$\Theta(\ell^{2})$ ancillas and weights exponential in $\ell$. The doubling identity
$T_{2\ell}=2T_\ell^{2}-1$ suggests composing instead: synthesise $y\mapsto y^{2}$ once,
apply it $\log_2\ell$ times, and rescale between stages so that every input stays in
the certified disc. Appendix~\ref{app:recursion} carries this out for $\sigma<r$,
and it works. The filter reaches any accuracy $\eta$ with
$L=O(\log\frac \ell\eta)$ chains per stage and $O(\log^{3}\ell)$ ancillas in all
(Proposition~\ref{prop:filter}). Its cost is the energy scale. Every stage
multiplies the norm by the total weight $\Lambda_L\ge\Delta$ of the synthesis, so
$\|\Sigma\|$ grows as $\ell^{\log_2(\Lambda_L/\sigma)}$ (Proposition~\ref{prop:energy}),
quasi-polynomially in the degree. The next subsection shows that the synthesis,
and with it this cost, can be dropped.

\subsection{One exact block, iterated}\label{sec:single}

The recursive doubling procedure does not need the square. It needs a map that folds an interval onto
itself, as $y\mapsto2y^{2}/\sigma-\sigma$ folds $[-\sigma,\sigma]$, and a point
outside the interval that escapes under iteration. The folding makes the stopband (the range $d\le|x|\le\sigma$)
oscillate and the escape makes the marked state stand out. The shortest chain
provides such a map exactly. By Proposition~\ref{prop:m1} its branch is
\begin{equation}
\lambda(y):=\lambda_0^{(1)}(y)=\tfrac12\Bigl(\sqrt{\Delta^{2}+4y^{2}}-\Delta\Bigr)
\label{eq:lam1}
\end{equation}
for every real $y$, not only on a disc, and for $m=1$ the ancilla register
of~\eqref{eq:even-gadget-Hamiltonian-matrix} is a single qubit,
\begin{equation}
E_2(A)=A\otimes X-\Delta\,I\otimes\ket1\!\bra1 .
\label{eq:E2}
\end{equation}
On an eigenvector of $A$ with eigenvalue $a$ it acts as the $2\times2$ matrix
$E_2(a)$, whose eigenvalues are $\lambda(a)$, with unit eigenvector
$\ket{\xi(a)}$, and $-\Delta-\lambda(a)$.

\begin{definition}[Single-block stages]\label{def:single}
For $\Delta>0$ and $0<d<\sigma$ put
\begin{equation}
\theta:=\frac{2\sigma}{\lambda(\sigma)} ,
\qquad
\theta_{\rm in}:=\frac{2\sigma}{\lambda(\sigma)-\lambda(d)} ,
\qquad
Q(y):=\sigma-\theta\,\lambda(y) ,
\label{eq:Qmap}
\end{equation}
$f_0(x):=\sigma-\theta_{\rm in}\bigl(\lambda(x)-\lambda(d)\bigr)$ and
$f_{i+1}:=Q\circ f_i$, and, with one fresh ancilla qubit per stage,
\[
\Sigma_0:=\bigl(\sigma+\theta_{\rm in}\lambda(d)\bigr)I-\theta_{\rm in}\,E_2(H_1) ,
\qquad
\Sigma_{i+1}:=\sigma I-\theta\,E_2(\Sigma_i) .
\]
\end{definition}

The constants are chosen so that the stopband is trapped and the marked point is
not. First, $Q$ is even and decreasing in $|y|$, with $Q(0)=\sigma$, and $\theta$ is
fixed by $Q(\pm\sigma)=-\sigma$. So $Q$ maps each half of $[-\sigma,\sigma]$ onto the
whole interval: it folds the interval onto itself, as the Chebyshev map does, and a
value that lies in $[-\sigma,\sigma]$ stays there at every later stage. Second,
$\theta_{\rm in}$ is fixed so that the entry map sends the passband edge $d$ to
$\sigma$ and the band edge $\sigma$ to $-\sigma$. The whole stopband
$d\le|x|\le\sigma$ therefore enters the interval, while the marked point $x=0$
lands just outside, at
\[
f_0(0)=\sigma+u_0 ,\qquad u_0:=\theta_{\rm in}\lambda(d)>0 .
\]
Third, outside the interval $Q$ no longer folds, it expands. A point at distance
$u$ beyond $\pm\sigma$ is sent to distance
$\theta\bigl(\lambda(\sigma+u)-\lambda(\sigma)\bigr)$ beyond $-\sigma$. This is at
least $\theta\lambda'(\sigma)\,u$ because $\lambda$ is convex, and less than
$\theta\,u$ because $\lambda'<1$. At each stage the marked point therefore moves
away from the interval by a factor between
\begin{equation}
\varrho:=\theta\lambda'(\sigma)=2\Bigl(1+\frac{\Delta}{\sqrt{\Delta^{2}+4\sigma^{2}}}\Bigr)\in(2,4)
\qquad\text{and}\qquad
\theta=\frac{\Delta+\sqrt{\Delta^{2}+4\sigma^{2}}}{\sigma}\;>\;\varrho .
\label{eq:rho}
\end{equation}
After $k$ stages the stopband is still within $\sigma$ while the marked value is at
least $\sigma+\varrho^{k}u_0$, which is what the following proposition makes
precise.

\begin{proposition}[A filter from one exact block]\label{prop:single}
Let $H_1$ be $\varkappa$-local and satisfy~\eqref{eq:H1spec}, let $0<d\le\delta$
with $d<\sigma$, let $\Delta>0$ and $k\ge0$.
\begin{enumerate}[label=(\roman*)]
\item $\Sigma_k$ is $(\varkappa+k+1)$-local on $k+1$ ancilla qubits, and for every
eigenpair $H_1\ket{v_x}=x\ket{v_x}$
\[
\Sigma_k\bigl(\ket{v_x}\otimes\ket{\Xi_k(x)}\bigr)=f_k(x)\,\ket{v_x}\otimes\ket{\Xi_k(x)} ,
\qquad
\ket{\Xi_k(x)}:=\bigotimes_{i=0}^{k}\bigl|\xi\bigl(f_{i-1}(x)\bigr)\bigr\rangle ,
\quad f_{-1}(x):=x .
\]
\item $|f_k(0)|\ge\sigma+\varrho^{k}u_0$ with $u_0>2d^{2}/\sigma$, and $f_k(0)<0$ for
$k\ge1$.
\item $\widehat R_k:=f_k/f_k(0)$, the encoded action of $\Sigma_k/f_k(0)$, satisfies
\begin{equation}
\bigl\|\widehat R_k(H_1)-\ket m\!\bra m\bigr\|\;\le\;\frac{\sigma}{|f_k(0)|}\;\le\;\frac{\sigma}{u_0}\,\varrho^{-k} ,
\label{eq:single_stop}
\end{equation}
with equality in the first inequality when $\sigma=\|H_1\|$. For every $\eta>0$ the
error is therefore at most $\eta$ once $k\ge\log_\varrho\frac{\sigma^{2}}{2d^{2}\eta}$.
\item There is a constant $C$, depending only on $\Delta/\sigma$, such that
$\|\Sigma_k\|\le C\,(\sigma/d)^{2\log_\varrho\theta}\,|f_k(0)|$ for every $k\ge0$ and
every $d\le\sigma/2$.
\end{enumerate}
\end{proposition}
\deferproof{app:single}

Proposition~\ref{prop:single} suggests a practical rule: use the transformation
that the gadget performs natively, instead of synthesising a known one. Compared
with the synthesised square of Appendix~\ref{app:recursion} to construct the Chebyshev filter, 
the single block
improves every resource (Table~\ref{tab:tradeoff}). Each stage costs one ancilla
qubit and one unit of locality, where a stage of the recursion is a synthesis on
$O\bigl((k+\log\frac1\eta)^{2}\bigr)$ qubits (Proposition~\ref{prop:filter}(iii)).
There is no synthesis error to control and no radius condition on the input,
since~\eqref{eq:lam1} holds for every real input. Rescaling $H_1$, $d$ and
$\Delta$ by a common factor rescales $\Sigma_k$ and leaves $\widehat R_k$
unchanged, so the construction depends only on $d/\sigma$ and $\Delta/\sigma$.
The scale $\sigma$ is then a free unit, required only to bound the spectrum of
$H_1$ as in~\eqref{eq:H1spec}, and since every cost grows with $\sigma/d$ the best
choice is $\sigma=\|H_1\|$. Finally the normalised norm does not grow with the
number of stages. By part~(iv) it depends on the passband alone, as
$O\bigl((\sigma/d)^{2\log_\varrho\theta}\bigr)$, where the recursion for Chebyshev synthesis 
pays a factor $\Lambda_L/\sigma$ per stage.

The price is the stopband amplitude, which decays geometrically in $k$ rather than as
$\mathrm e^{-2^{k+1}d/\sigma}$, because $\lambda$ is quadratic only for inputs
below $\Delta$ and linear above. By part~(iii), accuracy $\eta$ needs
\[
k=\frac{2\log_2(\sigma/d)+\log_2(1/\eta)}{\log_2\varrho}+O(1)
\]
stages, against $k=\log_2\bigl(\frac\sigma d\log\frac1\eta\bigr)+O(1)$ for the
Chebyshev filter by~\eqref{eq:Lneed}. The detuning tunes the block between two
classical maps. As $\Delta/\sigma\to0$, $Q$ tends to the tent map
$y\mapsto\sigma-2|y|$ and $\theta,\varrho\to2$. As $\Delta/\sigma\to\infty$ it tends
to the Chebyshev map and $\varrho\to4$, so the stage count in $\sigma/d$, the
parameter that makes filtering hard, becomes that of the Chebyshev filter: the
coefficient $2/\log_2\varrho$ is already $1.007$ at $\Delta=10\sigma$.
Figure~\ref{fig:single}(a) shows these regimes at the same depth. A larger
detuning costs energy, however: the exponent $2\log_\varrho\theta$ of part~(iv)
grows from $4.37$ at $\Delta=10\sigma$ to $7.64$ at $\Delta=100\sigma$. For example,
at $d=\sigma/10$, $\Delta=4\sigma$ and $\eta=10^{-2}$, six stages suffice, that is
$7$ ancilla qubits, at a normalised norm of $3.4\cdot10^{4}$, where the synthesised
square at similar working points ($\Delta=10, \, \sigma=0.5$) needs a normalised norm above
$10^{11}$.


A last option keeps the locality constant, by replacing the iteration with one
long chain of Section~\ref{sec:gadget}. Feed the entry stage alone into that
chain: take $E_{2m}(c\,\Sigma_0)$, with a detuning $\Delta'\ge\Delta_{\min}(m,r)$ as
in Corollary~\ref{cor:Dmin} and $c:=r'/(\sigma+u_0)$ for some $r'<r$, so that its
encoded inputs $cf_0(x)$ lie in $[-r',r']$, the marked one at $r'$. It is
$(\varkappa+3)$-local whatever the accuracy, and it uses $2m+1$ ancilla qubits.
Since $\lambda_0^{(m)}(y)$ starts as $y^{2m}/\Delta'$, the bracket of
Proposition~\ref{prop:bracket} puts the stopband below the marked value by a
factor $(1+u_0/\sigma)^{-2m}$, up to a factor polynomial in $m$, so accuracy
$\eta$ needs $m$ of order $\frac{\sigma}{u_0}\ln\frac1\eta$. Everything then
depends on the separation $u_0$ that the entry stage produces. It is of order $d$
when the entry detuning is at most $d$, where $\lambda$ is close to $|x|$, and of
order $d^{2}/\sigma$ once that detuning exceeds $\sigma$, where $\lambda$ is
quadratic and flattens the passband. The number of ancillas is thus linear in
$\sigma/d$ in the first case and quadratic in the second, where the iterated
block needs $O\bigl(\log\frac{\sigma}{d\eta}\bigr)$. This is where a power is a
weaker filter than an iteration: at $2m=2^{k}$ the iterated block has moved the
marked point by at least $\varrho^{k}u_0=(2m)^{\log_2\varrho}u_0$, close to
$(2m)^{2}u_0$, where the chain gains only about $2m\,u_0$. The marked value is at
most $r'^{2m}/\Delta'$, so the normalised norm grows as $r'^{-2m}$ unless $1-r'$
is of order $1/m$, and Corollary~\ref{cor:Dmin} then asks for a detuning
$\Delta'$ of order $m$. Figure~\ref{fig:single}(b) shows this option at
$2m=2^{k}=16$ for two entry detunings.

\begin{figure}[ht]
\centering
\includegraphics[width=\textwidth]{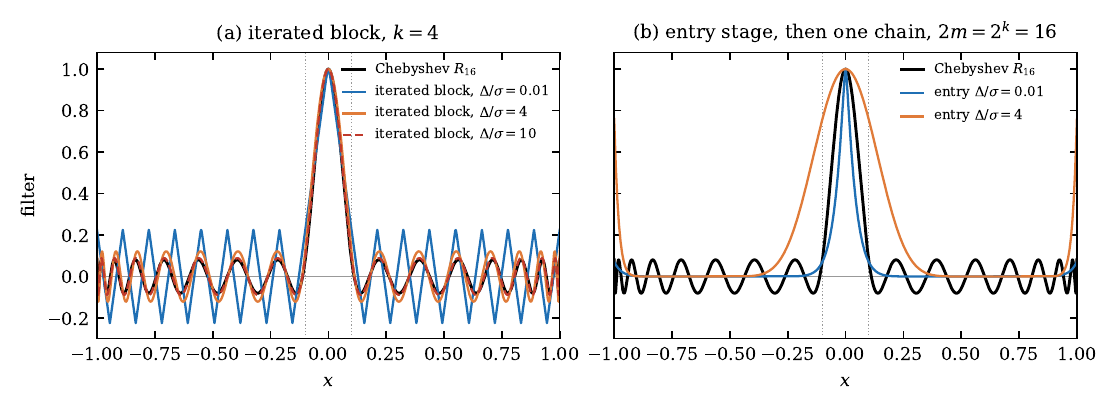}
\caption{Filters at $\sigma=1$ and $d=0.1$ (dotted verticals), each normalised to
$1$ at the marked energy $x=0$. In both panels the black curve is the Chebyshev
filter $R_{16}$ of Definition~\ref{def:filter}, which the synthesised square of
Appendix~\ref{app:recursion} approximates. (a) The iterated block
$\widehat R_4$ of Proposition~\ref{prop:single}, at depth $k=4$, for the three
detunings $\Delta/\sigma$ given in the legend. The smallest one is close to the
tent map, and the block approaches $R_{16}$ as $\Delta/\sigma$ grows. (b) The
last option of this section, the entry stage $\Sigma_0$ followed by one chain
$E_{2m}$ with $2m=2^{k}=16$ sites. The legend gives the detuning $\Delta/\sigma$
of the entry stage. The chain has detuning $\Delta'=10$, and its input is rescaled
so that the marked point sits at $r'=0.8$. The option filters with the small
entry detuning, where the entry map is close to $|x|$, and not with the large one,
where it is quadratic.}
\label{fig:single}
\end{figure}

\subsection{Adiabatic search}\label{sec:search}

We now put the filter in place of the projector in analog search. Take for $H_1$
the Hamming distance to $m$ multiplied by $\sigma/n$, which is $1$-local with
$\delta=\sigma/n$, set $d=\delta$ and $H_0=-\frac1n\sum_iX_i$, and consider
\begin{equation}
H(s)=(1-s)\,H_0\otimes I+s\,\frac{\Sigma_k}{|f_k(0)|} ,\qquad s\in[0,1] ,
\label{eq:Hsearch}
\end{equation}
on $n$ system qubits and $k+1$ ancilla qubits. Here $\sigma/d=n$, so at constant
accuracy $\eta$ Proposition~\ref{prop:single}(iii) gives $k=O(\log n)$ stages: the
Hamiltonian uses $O(\log n)$ ancilla qubits and is $O(\log n)$-local. Since
$f_k(0)<0$ for $k\ge1$, its second term acts on the encoded subspace as
$-\widehat R_k(H_1)$, the filter in place of $-\ket m\!\bra m$. Conjugating by Hadamard gates exchanges the two roles: $H_1$ becomes a transverse field whose unique zero-energy state is the uniform superposition, and $H_0$ a diagonal cost function. Read this way, \eqref{eq:Hsearch} is the unstructured adiabatic optimisation of~\cite{Braida2025unstructured} with its rank-one driver replaced by the gadget, for the cost function given by the Hamming weight. For another cost function the filter is unchanged, since it depends on the driver alone, but the spectrum is no longer the one simulated here.

Proposition~\ref{prop:single} does not by itself say that the search works. The
driver $H_0\otimes I$ does not leave the encoded subspace invariant, because the
ancilla factor $\ket{\Xi_k(x)}$ depends on the eigenvalue $x$ of $H_1$ while $H_0$
mixes eigenvalues. Compressed to the encoded subspace, \eqref{eq:Hsearch} is the
operator $(1-s)\,H_0\circ G-s\,\widehat R_k(H_1)$ on the system register, written
in the eigenbasis of $H_1$, where $\circ$ is the entrywise product and
$G_{yx}=\langle\Xi_k(y)|\Xi_k(x)\rangle$ damps each entry of the driver. The rest
of the driver couples the subspace to its complement. We therefore
diagonalise~\eqref{eq:Hsearch} numerically as it stands, with every ancilla, in
the sector invariant under permutations of the system qubits.

Figure~\ref{fig:search} shows the result. The compressed operator reproduces the
spectrum of analog search (a). In the full Hamiltonian, the gap of the encoded 
branch at the anticrossing follows that of analog search, which closes as
$2^{-n/2}$, to within $15\%$ for $n\le20$ (b). Near the anticrossing the
non-encoded levels stay at a distance of order $s\theta\Delta/|f_k(0)|$ from that
branch, which must exceed the driver, and this is why the detuning is taken
large.

The price is the energy scale. A device with a fixed energy budget runs
$\Sigma_k/\|\Sigma_k\|$ in place of $\Sigma_k/|f_k(0)|$, which stretches every
evolution time by $\|\Sigma_k\|/|f_k(0)|$. By Proposition~\ref{prop:single}(iv)
that factor is $O\bigl(n^{2\log_\varrho\theta}\bigr)$, that is $O(n^{4.37})$ at
$\Delta=10\sigma$, and this is proved. The spectrum and the gap are observed, on
one family and for $n\le20$, and we do not run the search itself. If its run time
is set by the gap at the anticrossing, as for analog search, the energy-time
exceeds that of analog search, which is $\Theta(2^{n/2})$, by at most this
polynomial factor. A proof that the interpolation works for every $n$ remains
open, and Section~\ref{sec:discussion} says what it would need.

\begin{figure}[ht]
\centering
\includegraphics[width=1\textwidth]{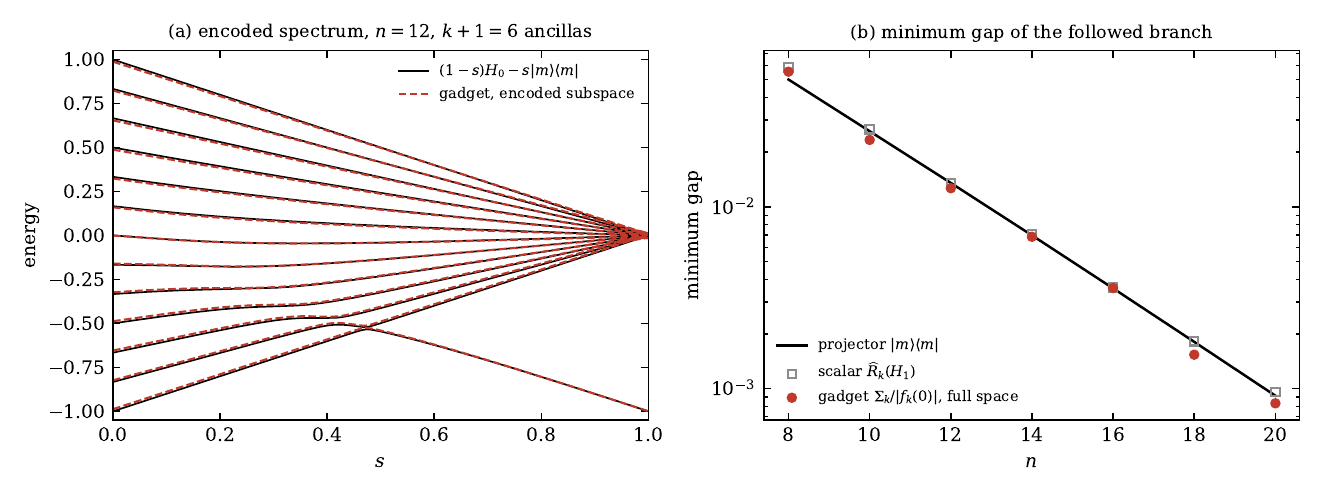}
\caption{Analog search with the gadget Hamiltonian,
$H(s)=(1-s)H_0\otimes I+s\,\Sigma_k/|f_k(0)|$ as in~\eqref{eq:Hsearch}, simulated
with every ancilla qubit, at $\sigma=0.5$, $\Delta=10\sigma$, $d=\delta=\sigma/n$
and stopband at most $5\cdot10^{-2}$. (a) At $n=12$ and $k=5$, the spectrum
of~\eqref{eq:Hsearch} compressed to the encoded subspace, an operator on the
system register (dashed), against that of $(1-s)H_0-s\ket m\!\bra m$ (solid).
(b) Gap at the anticrossing against $n$, for the projector, for the scalar
filter $\widehat R_k(H_1)$ substituted on the system register, and for the
encoded eigenvalue of the gadget Hamiltonian in the full space, whose gap
also closes as $s\to0$, where $H(0)=H_0\otimes I$ is degenerate.}
\label{fig:search}
\end{figure}

\section{Discussion}\label{sec:discussion}

Analog devices implement local, few-body
Hamiltonians~\cite{bloch2012simulations,browaeys2020rydberg,daley2022practical},
whereas the operators that algorithms ask for are often far from local. A
polynomial of degree $2\ell$ of a $\varkappa$-local Hamiltonian couples up to
$2\ell\varkappa$ qubits at once, and the projector of analog search couples all
$n$ of them. The result of this paper is that such operators can be carried by a
local Hamiltonian, at a modest and explicit price. One attaches short chains of
ancilla qubits to the input $H$. Each chain has an isolated eigenvalue that is an
exactly known analytic function of $H$, with a leading order set by its length
(Sections~\ref{sec:gadget} and~\ref{sec:branch}), and a weighted sum of chains
reproduces any prescribed polynomial of $H$ on an invariant subspace
(Section~\ref{sec:synthesis}, Appendix~\ref{app:odd}). The resulting Hamiltonian
is time-independent, so the transformation is switched on once rather than
sequenced in time, and its locality is only two more than that of $H$. The price
is $O(\ell^{2}+\log^{2}(1/\varepsilon))$ ancilla qubits and a larger energy
scale, where the circuit model pays the degree in sequential oracle calls
(Table~\ref{tab:tradeoff}).

The filter of Section~\ref{sec:filter} shows how far this goes. A rank-one
projector on $2^{n}$ dimensions, the least local operator there is, is
approximated to constant accuracy by a Hamiltonian on $O(\log n)$ ancilla qubits,
with $O(\log n)$-body terms and an energy scale polynomial in $n$, and a variant
keeps the locality constant at the cost of more ancillas. The filter also shows
which construction to prefer. Synthesising the Chebyshev filter, a known
polynomial, costs an energy scale quasi-polynomial in its degree. Iterating
instead the exact eigenvalue transformation of the two-site chain needs no
synthesis and no radius condition, costs one ancilla qubit per stage, and has an
energy scale that does not grow with the number of stages. Only the dependence on the
accuracy is exponentially worse: the number of stages grows as $\log(1/\eta)$
instead of $\log\log(1/\eta)$. Using the
transformation that the gadget performs natively is more efficient than imitating
a known one.

One point remains to be established. The encoded subspace is invariant under the
gadget Hamiltonian, so an evolution under that Hamiltonian alone never leaves it.
An adiabatic algorithm adds a driver that does not commute with the input, and it
is then not immediate that the evolution behaves as it would with the ideal
operator. On small instances the spectrum supports this: up to $n=20$ the full
Hamiltonian, ancillas included, reproduces the spectrum of analog search and the
scaling of its gap at the anticrossing (Section~\ref{sec:search}). Validating
this observation formally is left for further work, for which perturbative
methods that describe the effective Hamiltonian in a given
subspace~\cite{feshbach1958,schrieffer1966,bravyi2011schrieffer} are the natural
tool.

Two questions seem to us the most promising. The first is which other Hamiltonian
transformations are useful for applications and can be built in this way.
Our filter example also
suggests looking for transformations using the exact eigenvalue map of the
gadgets, combined by sums and compositions, rather than through the synthesis of a
prescribed polynomial. The second is whether the same idea, an ancilla gadget
whose spectrum is known exactly, reaches other algorithmic building blocks. The
main one is the block encoding: a design in which a prescribed matrix $A$,
rather than a function of a Hermitian input, appears as a block of a local
Hamiltonian. It would give the continuous-time paradigm the primitive on which
the circuit model has built almost
everything~\cite{gilyen2019qsvt,martyn2021grand,lloyd2021hsvt}.

\appendix

\section{The Chebyshev recursion with a synthesised square}\label{app:recursion}

This appendix supports the two claims made at the end of
Section~\ref{sec:filter:target}: composing a synthesised square along the
doubling identity gives the Chebyshev filter to any accuracy, and it does so at a
quasi-polynomial energy scale. Throughout, $\sigma<r$ and a margin
$\sigma<\bar\sigma\le2\sigma$ with $\bar\sigma<r$ is fixed.

\subsection{Construction}

Let $\ell=2^{k}$, $c_i:=T_{2^{i}}(\phi(0))$ and $g_i:=\sigma R_{2^{i}}$ for
$i=0,\dots,k$, with $\phi$ and $R_\ell$ as in Definition~\ref{def:filter}. The
values $c_i$ carry no superscript, unlike the branch coefficients $c^{(m)}_j$.

\begin{lemma}[The stage map]\label{lem:stage}
The $c_i$ satisfy $c_{i+1}=2c_i^{2}-1$ and $c_i>1$ for $i\ge1$, and
\begin{equation}
g_{i+1}=\theta_ig_i^{2}-\frac{\sigma}{c_{i+1}} ,
\qquad
\theta_i:=\frac{2c_i^{2}}{\sigma\,c_{i+1}}=\frac1\sigma\Bigl(1+\frac1{c_{i+1}}\Bigr) .
\label{eq:stagemap}
\end{equation}
Moreover $\|g_i\|_{\infty}\le\sigma$ on $|x|\le\sigma$, and
\begin{equation}
\frac1\sigma<\theta_i<\frac2\sigma .
\label{eq:stagebounds}
\end{equation}
\end{lemma}
\begin{proof}
$c_{i+1}=2T_{2^{i}}(\phi(0))^{2}-1$ is the Chebyshev doubling identity. $|c_0|>1$ gives $c_1>1$ and inductively $c_{i+1}>1$. Dividing $T_{2^{i+1}}(\phi)=2T_{2^{i}}(\phi)^{2}-1$ by $c_{i+1}$ and multiplying by $\sigma$ gives~\eqref{eq:stagemap}, and $2c_i^{2}=c_{i+1}+1$ gives the second form of $\theta_i$, from which~\eqref{eq:stagebounds} follows since $c_{i+1}>1$. The bound $\|g_i\|_\infty\le\sigma$ is Lemma~\ref{lem:filter}.
\end{proof}

Stage $0$ has the same form with input $g_{-1}:=x$, since
$g_0=\sigma\phi/\phi(0)=\theta_{-1}x^{2}+\sigma$ with
$\theta_{-1}:=-2\sigma/(\sigma^{2}+d^{2})$ and $|\theta_{-1}|<2/\sigma$. So every
stage is a prefactor times a square plus a constant, every input has norm at most
$\sigma<r$, and, the triangular system being linear in the target, one synthesis
serves every stage.

\begin{definition}[The synthesised square, and the recursion it drives]\label{def:square}
Let $\Psi=\sum_{j\le L}\alpha_j\lambda_0^{(j)}$ be the synthesis of $y\mapsto y^{2}$
with $L$ chains, the $\alpha_j$ given by~\eqref{eq:alpha_rec} for $p_2=1$. Put
$b_{-1}:=\sigma$ and $b_i:=-\sigma/c_{i+1}$ for $i\ge0$. The realised stages are
$\widetilde g_{-1}:=x$ and $\widetilde g_{i+1}:=\theta_i\Psi(\widetilde g_i)+b_i$,
and the stage Hamiltonians are $\Sigma_{-1}:=H_1$ and
$\Sigma_{i+1}:=\theta_i\sum_{j\le L}\alpha_jE_{2j}(\Sigma_i)+b_i$, each stage on
$L$ fresh chains, for $-1\le i<k$.
\end{definition}

The weights depend on $L$ and $\Delta$ only, so raising $\ell$ adds stages rather
than amplitude. The construction is exact only if $\Psi(y)=y^{2}$, and the
proposition below says how close is close enough.

\begin{proposition}[The recursion works]\label{prop:filter}
Let $0<\eta<1$, put $\varepsilon_L:=\sup_{|y|\le\bar\sigma}|\Psi(y)-y^{2}|$, and
assume
\begin{equation}
2\ell\frac d\sigma\;\ge\;\ln\frac8\eta ,
\qquad
\varepsilon_L\;\le\;\frac{\sigma}{6^{k+1}}\,
\min\Bigl\{\bar\sigma-\sigma,\ \frac{\eta\sigma}4\Bigr\} .
\label{eq:Lneed}
\end{equation}
\begin{enumerate}[label=(\roman*)]
\item Every realised stage stays in the certified disc,
$\|\widetilde g_i\|_\infty\le\bar\sigma$ on $|x|\le\sigma$, and for every eigenpair
$H_1\ket{v_x}=x\ket{v_x}$
\[
\Sigma_k\bigl(\ket{v_x}\otimes\ket{\Xi_k(x)}\bigr)=\widetilde g_k(x)\,\ket{v_x}\otimes\ket{\Xi_k(x)} ,
\qquad
\ket{\Xi_k(x)}:=\bigotimes_{i=0}^{k}\bigotimes_{j\le L}\bigl|\xi_j\bigl(\widetilde g_{i-1}(x)\bigr)\bigr\rangle ,
\]
with $\ket{\xi_j(y)}$ the branch eigenvector of $E_{2j}(y)$.
\item $v_k:=\widetilde g_k(0)\ge\sigma/2$, and $\widetilde R_\ell:=\widetilde g_k/v_k$,
the encoded action of $\Sigma_k/v_k$, satisfies
$\|\widetilde R_\ell(H_1)-\ket m\!\bra m\|\le\eta$.
\item The second condition of~\eqref{eq:Lneed} holds with $L=O(\log\frac \ell\eta)$,
so that $\Sigma_k$ uses $(k+1)L(L+1)=O\bigl(\log \ell\log^{2}\frac \ell\eta\bigr)$
ancilla qubits and is $\bigl(\varkappa+2(k+1)\bigr)$-local.
\end{enumerate}
\end{proposition}
\begin{proof}
Write $\mathcal E_L:=\Psi-y^{2}$, so that the exact and realised stages are
$g_{i+1}=\theta_ig_i^{2}+b_i$ and $\widetilde g_{i+1}=\theta_i\Psi(\widetilde g_i)+b_i$
for $i\ge-1$, with $g_{-1}=\widetilde g_{-1}=x$. By Lemma~\ref{lem:stage} and the
bound on $\theta_{-1}$, $|\theta_i|<2/\sigma$ for every $i\ge-1$, and
$\|g_i\|_\infty\le\sigma$ on $|x|\le\sigma$. The deviation
$D_i:=\widetilde g_i-g_i$ obeys, pointwise in $x$,
\[
D_{i+1}=\theta_i\bigl(2g_iD_i+D_i^{2}+\mathcal E_L(\widetilde g_i)\bigr) ,
\qquad D_{-1}=0 .
\]

(i) Put $D_\star:=6^{k+1}\varepsilon_L/\sigma$, so that the second condition
of~\eqref{eq:Lneed} reads $D_\star\le\min\{\bar\sigma-\sigma,\ \eta\sigma/4\}$, and
$D_\star\le\sigma$ because $\bar\sigma\le2\sigma$. Suppose that
$\|D_j\|_\infty\le D_\star$ for every $j\le i$, where $i<k$. Then
$\|\widetilde g_j\|_\infty\le\sigma+D_\star\le\bar\sigma$, hence
$|\mathcal E_L(\widetilde g_j)|\le\varepsilon_L$, and $D_j^{2}\le\sigma|D_j|$, so
\[
\|D_{j+1}\|_\infty\le\frac2\sigma\bigl(3\sigma\|D_j\|_\infty+\varepsilon_L\bigr)
=6\,\|D_j\|_\infty+\frac{2\varepsilon_L}{\sigma}
\qquad\text{for } j\le i .
\]
Unrolling from $D_{-1}=0$, and using $i+1\le k$,
\[
\|D_{i+1}\|_\infty\le\frac{2\varepsilon_L}{\sigma}\sum_{s=0}^{i+1}6^{s}
\le\frac{2\varepsilon_L}{5\sigma}\,6^{k+1}\le D_\star .
\]
By induction $\|D_i\|_\infty\le D_\star$ and $\|\widetilde g_i\|_\infty\le\bar\sigma<r$
for every $i\le k$. For the encoded states, argue by induction on the stage as in
the proof of Proposition~\ref{prop:single}(i). If
$\Sigma_i\ket{\psi_i}=\widetilde g_i(x)\ket{\psi_i}$ with
$\ket{\psi_i}=\ket{v_x}\otimes\ket{\Xi_i(x)}$, then each $E_{2j}(\Sigma_i)$ acts on
$\ket{\psi_i}\otimes\mathcal C_j$ as the matrix $E_{2j}(y)$, $y=\widetilde g_i(x)$,
whose branch eigenvector $\ket{\xi_j(y)}$ exists and is unique up to a phase
because $|y|\le\bar\sigma<r$, by the simplicity argument in the proof of
Theorem~\ref{thm:realisation}. The $L$ chains of the stage are disjoint, so
$\Sigma_{i+1}$ has eigenvalue $\theta_i\Psi(y)+b_i=\widetilde g_{i+1}(x)$ on
$\ket{\psi_i}\otimes\bigotimes_j\ket{\xi_j(y)}$, and the base case $i=-1$ is the
same statement with $y=x$.

(ii) $g_k(0)=\sigma$, so $v_k\ge\sigma-D_\star\ge\sigma/2$, since
$D_\star\le\eta\sigma/4\le\sigma/4$. On the encoded subspace $\Sigma_k/v_k$ acts as
$\widetilde R_\ell(H_1)$, and as in the proof of Proposition~\ref{prop:single}(iii)
the error is the largest $|\widetilde R_\ell(x)|$ over the nonzero eigenvalues $x$
of $H_1$. For $\delta\le|x|\le\sigma$, hence $d\le|x|\le\sigma$, one has
$|\phi|\le1$ and so $|g_k|\le\sigma/c_k$. Lemma~\ref{lem:filter} and the first condition
of~\eqref{eq:Lneed} give
\[
\frac1{c_k}\;\le\;2\mathrm e^{-2\ell d/\sigma}\;\le\;\frac\eta4 ,
\]
and therefore
\[
\max_{\delta\le|x|\le\sigma}\bigl|\widetilde R_\ell(x)\bigr|
\le\frac{\sigma/c_k+D_\star}{\sigma-D_\star}
\le\frac{\eta/4+\eta/4}{1-\eta/4}\le\eta .
\]

(iii) Proposition~\ref{prop:tail_estimate} applied to $y\mapsto y^{2}$, for which
$M_P=1$ and $q=1$, gives
\begin{equation}
\varepsilon_L\;\le\;C_\star\,\Phi_L\,\gamma^{\,L+1} ,
\qquad
C_\star:=\frac{C_r\,r^{2}}{(1-r^{2})(1-\gamma)} ,
\qquad \gamma=\frac{\bar\sigma^{2}}{r^{2}}<1 .
\label{eq:epsL}
\end{equation}
The right-hand side of the second condition of~\eqref{eq:Lneed} is at least
$c\,\eta\,6^{-k}$ with $c:=\frac\sigma6\min\{\bar\sigma-\sigma,\sigma/4\}$, since
$\eta<1$, and $\ln\Phi_L=o(L)$ by~\eqref{eq:Phi_growth}. That condition therefore
holds as soon as $(L+1)\ln(1/\gamma)$ exceeds
$\ln\bigl(C_\star\Phi_L/(c\eta)\bigr)+k\ln6$, that is for some
$L=O(k+\log\frac1\eta)=O(\log\frac \ell\eta)$. Each stage uses $\sum_{j\le L}2j=L(L+1)$ ancilla qubits,
and by Lemma~\ref{lem:local} it adds two to the locality.
\end{proof}

\subsection{Energy}

Proposition~\ref{prop:filter} settles the accuracy and the number of ancillas of
the recursion. What remains is its energy scale, which in an analog setting is a
cost in time: a device with a fixed energy budget runs $\Sigma_k/\|\Sigma_k\|$ in
place of $\Sigma_k/v_k$, the Hamiltonian that carries the filter, and so stretches
every evolution time by $\|\Sigma_k\|/v_k$. This is where these gadgets for the Chebyshev filter are least
favourable. Each stage feeds the previous one into $L$ chains with weights
$\alpha_j$, so it multiplies the norm by a factor of order $\Lambda_L/\sigma$, and
the norm grows geometrically in the number of stages, that is polynomially in
$\ell=2^{k}$ with exponent $\log_2(\Lambda_L/\sigma)$.

\begin{proposition}[Energy scale of the recursion]\label{prop:energy}
Assume the hypotheses of Proposition~\ref{prop:filter} and $\Delta\ge3\sigma$. The
total weight of the synthesised square satisfies
\begin{equation}
\Lambda_L\;=\;\sum_{j=1}^{L}\bigl|\alpha_j\bigr|\;\ge\;|\alpha_1|\;=\;\Delta ,
\label{eq:LambdaL}
\end{equation}
and, with $\ell=2^{k}$,
\begin{equation}
\|\Sigma_k\|\;\le\;\frac{2\,\Lambda_L\,(\sigma+2\Delta+2)}{d}\,
\Bigl(\frac{\Lambda_L}{\sigma}\Bigr)^{k}
\;=\;\frac{2\,\Lambda_L\,(\sigma+2\Delta+2)}{d}\;\ell^{\log_2(\Lambda_L/\sigma)} .
\label{eq:energyN}
\end{equation}
The normalised norm $\|\Sigma_k\|/v_k$ is at most $2/\sigma$ times this bound, by
Proposition~\ref{prop:filter}(ii).
\end{proposition}
\begin{proof}
\emph{One stage.} For every Hermitian $A$ and every $j$ one has
$\|E_{2j}(A)\|\le\|A\|+\Delta+1$: in the form~\eqref{eq:local-form} the $x$-bonds
are disjoint and contribute $\|A\|$, the unit bonds are disjoint and contribute
$1$, and the detuned site contributes $\Delta$. Since
$\Sigma_{i+1}=\theta_i\sum_{j\le L}\alpha_jE_{2j}(\Sigma_i)+b_i$ with
$|b_i|\le\sigma$, and $\|\Sigma_{-1}\|=\|H_1\|\le\sigma$,
\begin{equation}
\|\Sigma_{i+1}\|\le|\theta_i|\,\Lambda_L\bigl(\|\Sigma_i\|+\Delta+1\bigr)+\sigma
\qquad(-1\le i<k) .
\label{eq:energy_rec}
\end{equation}

\emph{Homogenising.} By~\eqref{eq:stagebounds} and
$|\theta_{-1}|=2\sigma/(\sigma^{2}+d^{2})>1/\sigma$, every prefactor satisfies
$|\theta_i|\Lambda_L\ge\Delta/\sigma\ge3$. Hence
$|\theta_i|\Lambda_L(\Delta+1)\ge3(\Delta+1)\ge2(\Delta+1)+\sigma$, and adding
$2(\Delta+1)$ to~\eqref{eq:energy_rec} gives
$\|\Sigma_{i+1}\|+2(\Delta+1)\le|\theta_i|\Lambda_L\bigl(\|\Sigma_i\|+2(\Delta+1)\bigr)$.
Unrolling from $\|\Sigma_{-1}\|\le\sigma$,
\[
\|\Sigma_k\|\;\le\;\Lambda_L^{\,k+1}\,(\sigma+2\Delta+2)\prod_{i=-1}^{k-1}|\theta_i| .
\]

\emph{The prefactors telescope.} By~\eqref{eq:stagemap},
$\sigma\theta_i=1+1/c_{i+1}$ for $i\ge0$. By Lemma~\ref{lem:filter},
$|\phi(0)|=\cosh\beta$ with $\beta:=2\arctanh(d/\sigma)$, and Chebyshev doubling is
angle doubling, so $c_j=\cosh(2^{j}\beta)$ for $j\ge1$. Since
$1+\frac1{\cosh2y}=\frac{2\cosh^{2}y}{\cosh2y}=\frac{\tanh2y}{\tanh y}$,
\[
\prod_{i=0}^{k-1}\sigma\theta_i
=\prod_{j=1}^{k}\frac{\tanh(2^{j}\beta)}{\tanh(2^{j-1}\beta)}
=\frac{\tanh(2^{k}\beta)}{\tanh\beta}\;<\;\coth\beta
=\frac{\sigma^{2}+d^{2}}{2\sigma d}\;<\;\frac\sigma d .
\]
With $|\theta_{-1}|<2/\sigma$ this gives~\eqref{eq:energyN}. The last sentence is
$v_k\ge\sigma/2$.
\end{proof}

\begin{corollary}[The energy scale is quasi-polynomial]\label{cor:energy}
At fixed $\Delta$, and with $L=O\bigl(\log\tfrac \ell\eta\bigr)$ as in
Proposition~\ref{prop:filter}(iii), one has
$\log_2\Lambda_L=O\bigl(\log\tfrac \ell\eta\bigr)$ and
\[
\|\Sigma_k\|=\exp O\bigl(\log \ell\cdot\log\tfrac \ell\eta\bigr) ,
\]
which is quasi-polynomial in $\ell$.
\end{corollary}
\begin{proof}
$\Lambda_L=\sum_{j\le L}r^{-2j}\bigl(r^{2j}|\alpha_j|\bigr)\le r^{-2L}\mathcal A_L$,
and~\eqref{eq:Asum_bound} with $q=1$ and $M_P=1$, the data of
Definition~\ref{def:square}, gives $\mathcal A_L\le r^{2}\Delta\Phi_L/(1-r^{2})$.
At fixed $\Delta$ the radius $r$ is fixed, and $\ln\Phi_L=O(\log^{2}L)$
by~\eqref{eq:Phi_growth}, so
$\log_2\Lambda_L=2L\log_2(1/r)+O(\log^{2}L)=O(\log\tfrac \ell\eta)$. The first
condition of~\eqref{eq:Lneed} gives $1/d<\ell/\sigma$, so~\eqref{eq:energyN} yields
$\log_2\|\Sigma_k\|\le k\log_2(\Lambda_L/\sigma)+O(\log\tfrac \ell\eta)
=O\bigl(\log \ell\cdot\log\tfrac \ell\eta\bigr)$.
\end{proof}

The exponent in~\eqref{eq:energyN} is at least $\log_2(\Delta/\sigma)$
by~\eqref{eq:LambdaL}, and the bound on $\Lambda_L$ used above grows geometrically
with $L$, which itself grows with $\ell$. We have not been able to prove a better
bound on $\Lambda_L$ than the $r^{-2L}$ of Proposition~\ref{prop:alpha_growth},
and the measurements of Appendix~\ref{app:beyond} indicate that none of the
obvious improvements is available.

\subsection{Numbers}\label{app:beyond}

At the working point
\begin{equation}
\Delta=10,\qquad
r=0.7562,\qquad
\sigma=0.5,\qquad
\bar\sigma=0.6,\qquad
d=\sigma/10 ,
\label{eq:workingpoint}
\end{equation}
where $r$ is the largest four-digit radius certified at $\Delta=10$ by
Corollary~\ref{cor:Dmin}, the synthesis error is $\varepsilon_L=9.1\cdot10^{-6}$
at $L=5$ and $7.1\cdot10^{-8}$ at $L=8$. The weight $\Lambda_L$ is nearly flat
over small $L$, taking the values $10.2$, $10.5$, $11.5$ and $13.8$ at
$L=4,8,16,20$, and the largest weight is $\alpha_1=\Delta$ for every $L\le29$. That
flatness is not a law. The $\alpha_j$ do not depend on $L$, so $\Lambda_L$ is a
partial sum of one fixed series, and it reaches $722$ at $L=40$ and
$1.13\cdot10^{8}$ at $L=80$, the measured ratio $|\alpha_{j+1}/\alpha_j|$ over
$70\le j\le110$ being $1.97$, $1.36$ and $1.27$ at $\Delta=2,10,16$, against the
ceilings $r^{-2}=2.94$, $1.75$ and $1.58$ that Proposition~\ref{prop:alpha_growth}
allows. The growth is geometric at every detuning we have tested, at a rate that
falls with $\Delta$ and never reaches $1$, so no bound of the form
$\Lambda_L=O(L)$ is available.

For $\eta=10^{-2}$ the first condition of~\eqref{eq:Lneed} needs $k\ge6$, and the
second asks for $L=14$, $14$, $24$, $31$ and $38$ chains per stage at $k=6$, $8$,
$10$, $13$ and $16$. Computing $\|\Sigma_k\|$ exactly, which the commutativity of
the summands permits, gives $\log_\ell\|\Sigma_k\|=6.35$, $5.89$, $6.43$, $8.15$ and
$10.71$, within a factor $4.6$ to $7.3$ of~\eqref{eq:energyN}. The exponent first
falls while $\Lambda_L$ is flat and then rises as $L$ enters the geometric range,
which is Corollary~\ref{cor:energy} made visible. Condition~\eqref{eq:Lneed} is
only sufficient, and at $k=6$ a single chain per stage already brings the
realised stopband to $10^{-5}$. The normalised norm stays above $10^{11}$ all the
same, because it is set by the weight $\Lambda_L\ge\Delta$ of every stage and not
by the number of chains.

\section{Deferred proofs}\label{app:proofs}

This appendix collects the proofs that would otherwise interrupt the argument.
Nothing here is used before it is stated in the main text.

\subsection{Proof of Lemma~\ref{lem:transfer}}\label{app:transfer}

\begin{proof}
From~\eqref{eq:recurrence} with $j\ge1$:
\[
  \Theta_{2j}
  = -\lambda \Theta_{2j-1} - x^2 \Theta_{2j-2}
  = -\lambda(-\lambda \Theta_{2j-2}-\Theta_{2j-3}) - x^2 \Theta_{2j-2}
  = (\lambda^2-x^2)\Theta_{2j-2}+\lambda \Theta_{2j-3},
\]
\[
  \Theta_{2j-1} = -\lambda \Theta_{2j-2} - \Theta_{2j-3}.
\]
In matrix form this is exactly~\eqref{eq:transfer}. Iterating from $(\Theta_0,\Theta_{-1})=(1,0)$ gives~\eqref{eq:Tm_int}.
\end{proof}

\subsection{Proof of Lemma~\ref{lem:ch}}\label{app:ch}

\begin{proof}
The closed form follows from the linear recurrence~\eqref{eq:Sk_rec} with characteristic roots $\mu_\pm$ and initial conditions $S_0=0$, $S_1=1$.

Identity~\eqref{eq:CH} is proved by induction on $m$. 

\noindent Base cases:
\begin{itemize}
    \item $m=1$: $S_1 T-x^2 S_0 I = T-0=T=T^1$. 
    \item $m=2$: By Cayley--Hamilton, $T$ satisfies its characteristic equation $T^2-\tau T+x^2 I=0$, i.e.\ $T^2=\tau T-x^2 I=S_2 T-x^2 S_1 I$.
\end{itemize}
\noindent Inductive step ($m\ge2$): $T^{m+1}=T\cdot T^m = T(S_m T-x^2 S_{m-1}I) = S_m T^2-x^2 S_{m-1}T = S_m(\tau T-x^2 I)-x^2 S_{m-1}T = (S_m\tau-x^2 S_{m-1})T - x^2 S_m I = S_{m+1}T-x^2 S_m I$, using $S_{m+1}=\tau S_m-x^2 S_{m-1}$.
\end{proof}

\subsection{Proof of Lemma~\ref{lem:poles}}\label{app:poles}

\begin{proof}
(i) Substituting $S_k=x^{k-1}\widetilde S_k$ turns~\eqref{eq:Sk_rec} into
$\widetilde S_k=2z\widetilde S_{k-1}-\widetilde S_{k-2}$ with $z=\tau/2x$,
$\widetilde S_1=1$, $\widetilde S_2=2z$, so
$\widetilde S_m=U_{m-1}(z)$. The zeros of $U_{m-1}$ are
$z=\cos\vartheta_j$, i.e.\ $\tau=2x\cos\vartheta_j$, i.e.\ $w=w_j$. The
factorisation holds because $\omega_j\bar\omega_j=1$ and
$\omega_j+\bar\omega_j=-2\cos\vartheta_j$, and $S_m$ is monic in $w$ because
$\tau=w-x^2-1$.

(ii) By~\eqref{eq:poles}, $|w_j|$ is the product of the two factors
$|x-\omega_j|=|x+e^{i\vartheta_j}|$ and $|x-\bar\omega_j|=|x+e^{-i\vartheta_j}|$.
Since $|\omega_j|=1$, each factor is at least $|\omega_j|-|x|=1-|x|$. Moreover
$|x+e^{i\vartheta_j}|+|x+e^{-i\vartheta_j}|\ge|e^{i\vartheta_j}-e^{-i\vartheta_j}|=2\sin\vartheta_j$,
so the larger factor is at least $\sin\vartheta_j$. Multiplying the smaller
factor, at least $1-|x|$, by the larger one, at least
$\max\{1-|x|,\sin\vartheta_j\}$, gives~\eqref{eq:polebound2}. On $\mathcal D$ we
have $|w|=|\lambda|^{2}<(1-|x|)^{2}\le|w_j|$, so no zero is reached. The bound is
sharp: as $\vartheta_j\to0$ with $x\in(-1,0)$, the ratio $|w_j|/(1-|x|)^{2}$
tends to $1$.

(iii) $S_{m-1}/S_m=x^{-1}U_{m-2}(z)/U_{m-1}(z)$, and the zeros of
$U_{m-1}$ are simple, so
$U_{m-2}/U_{m-1}=\sum_j\kappa_j/(z-z_j)$ with $\kappa_j=\lim_{z\to z_j}(z-z_j)\frac{U_{m-2}(z)}{U_{m-1}(z)}
=\frac{U_{m-2}(z_j)}{U_{m-1}'(z_j)}$. Using
$U_{m-1}(\cos\vartheta)=\sin(m\vartheta)/\sin\vartheta$ one computes
$U_{m-2}(z_j)=-(-1)^{j}$ and
$U_{m-1}'(z_j)=-m(-1)^{j}/\sin^{2}\vartheta_j$, hence
$\kappa_j=\sin^{2}\vartheta_j/m>0$. And $\sum_{j=1}^{m-1}\sin^{2}\vartheta_j=m/2$ gives
$\sum_j\kappa_j=\tfrac12$. Substituting $z-z_j=(w-w_j)/2x$ yields~\eqref{eq:partialfrac}. For $m=1$ every sum and product is empty:
$S_1=1$ and $R_1=1$.
\end{proof}

\subsection{Proof of Proposition~\ref{prop:eig_eq}}\label{app:eig_eq}

\begin{proof}
The boundary modifies only the last diagonal entry:
$-\lambda \to -\lambda-\Delta$.  Expanding the last row:
\begin{equation}
  \chi_{2m}(\lambda)
  = (-\Delta-\lambda)\Theta_{2m-1} - x^2 \Theta_{2m-2}
  = \Theta_{2m} - \Delta \Theta_{2m-1},
\label{eq:D_boundary}
\end{equation}
where the second equality uses $\Theta_{2m}=-\lambda \Theta_{2m-1}-x^2 \Theta_{2m-2}$.
In vector form, using~\eqref{eq:Tm_int}:
\begin{equation}
  \chi_{2m}(\lambda)
  = \bigl(1,\,-\Delta\bigr)\,T^m\begin{pmatrix}1\\0\end{pmatrix}.
\label{eq:D_vec1}
\end{equation}
Applying~\eqref{eq:CH}:
\begin{align*}
  \chi_{2m}(\lambda)
  &= S_m(\lambda^2+\Delta\lambda-x^2)-x^2 S_{m-1}.
\end{align*}
Setting this to zero and dividing by $S_m$ gives~\eqref{eq:eig_reduced}.
\end{proof}

\subsection{Proof of Proposition~\ref{prop:parity}}\label{app:parity}

\begin{proof}
That $\chi_{2m}$ depends on $x$ only through $x^{2}$ is immediate
from~\eqref{eq:char_poly}: by~\eqref{eq:Sk_rec} each $S_k$ is a polynomial in
$\tau$ and $x^{2}$, and $\tau=\lambda^{2}-x^{2}-1$, while the remaining term
is $-x^{2}S_{m-1}$. Hence $\chi_{2m}(\lambda;x)=\chi_{2m}(\lambda;-x)$ and the
two matrices $E_{2m}(x)$, $E_{2m}(-x)$ are isospectral, so for real $x$ their
ordered eigenvalues agree, $\lambda_j(-x)=\lambda_j(x)$. Any branch determined by
a condition invariant under $x\mapsto-x$ is therefore even. 
\end{proof}

\subsection{Proof of Proposition~\ref{prop:bracket}}\label{app:bracket}

\begin{proof}
Set $F(\lambda):=\lambda(\lambda+\Delta)-x^{2}R_m(\lambda,x)$ on $0\le\lambda\le M/\Delta$. By hypothesis $\lambda^{2}<W$ there, so $(\lambda,x)\in\mathcal D$ and $F$ is continuous by Lemma~\ref{lem:poles}(ii).

\emph{Upper bound.} By~\eqref{eq:monotone}, $x^{2}R_m\le M$ for every such 
$\lambda$, so $F(\lambda)\ge\lambda(\lambda+\Delta)-M$, which is $\ge0$ at 
$\lambda=M/\Delta$. Since $F(0)=-M<0$, $F$ vanishes in $[0,M/\Delta]$, and 
only once, because $\lambda(\lambda+\Delta)$ is strictly increasing there 
and $-x^{2}R_m$ is non-decreasing by Corollary~\ref{cor:threefold}(ii). 
Lying in $[0,M/\Delta]$, this root tends to $0$ with $x$, so near $x=0$ 
it is $\lambda_0^{(m)}$.

\emph{Lower bound.} By~\eqref{eq:monotone} again,
\[
M-x^{2}R_m = x^{2}\,\lambda^{2}\sum_{j}\frac{2\kappa_j}{w_j(w_j-\lambda^{2})}
 \le \frac{x^{2}\,\lambda^{2}}{W\,(W-\lambda^{2})} ,
\]
using $w_j\ge W$ from~\eqref{eq:polebound2} and $\sum_j2\kappa_j=1$. Hence at the root, where $\lambda^{2}+\Delta\lambda=x^{2}R_m$,
\[
\Delta\lambda = M-\lambda^{2}-\bigl(M-x^{2}R_m\bigr)
 \ge M-\lambda^{2}\left(1+\frac{x^{2}}{W(W-\lambda^{2})}\right).
\]
Substituting the upper bound $\lambda\le M/\Delta$ on the right-hand side and dividing by $\Delta$ gives~\eqref{eq:bracket}. Finally $M=\Theta(x^{2m})$ and $W\to1$ as $x\to0$, so the bracket has relative width $O(x^{2m})$.
\end{proof}

\subsection{Proof of Theorem~\ref{thm:series}}\label{app:series}

\begin{proof}
First, by Lemma~\ref{lem:poles}(ii), $R_m$ is analytic in $(\lambda,x)$ on $\mathcal D$, in particular on a neighbourhood of the origin, so $F(\lambda,x)=\lambda^2+\Delta\lambda-x^2R_m(\lambda,x)$ is too, with $F(0,0)=0$ and
\[
\frac{\partial F}{\partial\lambda}\Big|_{(0,0)} =2\lambda-x^2\frac{\partial R_m}{\partial\lambda}\Big|_{(0,0)}+\Delta=\Delta\neq0 .
\]
The implicit function theorem yields a unique analytic branch $\lambda_0^{(m)}$ near $x=0$ with $\lambda_0^{(m)}(0)=0$, the elementary case of the analytic perturbation theory of~\cite{kato1995perturbation}. That characterisation is invariant under $x\mapsto-x$, so Proposition~\ref{prop:parity} makes the branch even and
$\lambda_0^{(m)}(x)=\sum_{j\ge0}b_jx^{2j}$ with $b_0=0$.

Second, write $\lambda:=\lambda_0^{(m)}(x)$ for $x$ real and small. The exact identity~\eqref{eq:monotone} of Corollary~\ref{cor:threefold} gives $x^2R_m=M-\varepsilon(x)$ with $\varepsilon(x)=x^{2}\lambda^{2}\sum_j\frac{2\kappa_j}{w_j(w_j-\lambda^{2})}\ge0$, so the eigenvalue equation $\lambda^2+\Delta\lambda=x^2R_m(\lambda,x)$ reads
\begin{equation}
\Delta\lambda  =  M-\lambda^2-\varepsilon(x),
\qquad
0\le\varepsilon(x)\le \frac{x^{2}\lambda^{2}}{W\,(W-\lambda^{2})},
\quad W=(1-x)^{2}.
\label{eq:exact2sided}
\end{equation}

Third, by Proposition~\ref{prop:bracket}, $0\le\lambda\le M/\Delta$, and $M=\Theta(x^{2m})$, so $\lambda^2=O(x^{4m})$. Since $W\to1$ and $\lambda\to0$ as $x\to0$, the bound in~\eqref{eq:exact2sided} gives $\varepsilon(x)=O(x^{2}\lambda^{2})=O(x^{4m+2})$. Finally, expanding the geometric series,
\[
M=\frac{x^{2m}(1-x^2)}{1-x^{2m}}
=(x^{2m}-x^{2m+2})\sum_{i\ge0}x^{2mi}
=x^{2m}-x^{2m+2}+x^{4m}+O(x^{4m+2}).
\]
Substituting into~\eqref{eq:exact2sided} and using $\lambda^2=\bigl(x^{2m}/\Delta\bigr)^2+O(x^{4m+2})=x^{4m}/\Delta^2+O(x^{4m+2})$,
\begin{equation}
\Delta\,\lambda_0^{(m)}(x) = x^{2m}-x^{2m+2}+\Bigl(1-\frac1{\Delta^2}\Bigr)x^{4m}+O(x^{4m+2}).
\label{eq:lambda_expansion}
\end{equation}

The left-hand side of~\eqref{eq:lambda_expansion} equals $\Delta\sum_{j\ge0}b_jx^{2j}$. Comparing coefficients of $x^{2j}$ for $2j\le 4m$ shows $b_j=0$ for $j<m$, so the series indeed starts at $x^{2m}$ and we may set $c_k^{(m)}:=b_{m+k}$. Then, for $m\ge2$, the three exponents $2m$, $2m+2$ and $4m$ are distinct and all at most $4m$, giving
\[
\Delta c_0^{(m)}=1,\qquad
\Delta c_1^{(m)}=-1,\qquad
\Delta c_k^{(m)}=0 \ \ (1<k<m),\qquad
\Delta c_m^{(m)}=1-\frac1{\Delta^2},
\]
which is the claim.

At $m=1$ the three exponents are no longer distinct: $2m+2=4=4m$, so the $-x^{2m+2}$ and $+\bigl(1-\Delta^{-2}\bigr)x^{4m}$ contributions land on the same power and only their sum is observable. Indeed $R_1=1$ and $M=x^{2}$ exactly, so~\eqref{eq:exact2sided} is the exact quadratic $\lambda^{2}+\Delta\lambda-x^{2}=0$, whose small root $\lambda=\frac{-\Delta+\sqrt{\Delta^{2}+4x^{2}}}{2}$ has $c_0^{(1)}=\Delta^{-1}$ and
\[
c_1^{(1)}=-\frac1{\Delta^{3}}
=\underbrace{-\frac1\Delta}_{\text{the }c_1^{(m)}\text{ term}}
+\underbrace{\Bigl(\frac1\Delta-\frac1{\Delta^{3}}\Bigr)}_{\text{the }c_m^{(m)}\text{ term}} .
\]
\end{proof}

\subsection{Proof of Lemma~\ref{lem:Rflat}}\label{app:Rflat}

\begin{proof}
Since $x^{2}R_m(0,x)=M(x)$, subtracting~\eqref{eq:partialfrac} from its value at $w=0$ gives the exact identity
\[
x^{2}R_m(\lambda,x)-M(x)
 = x^{2}\sum_{j=1}^{m-1}2\kappa_j\left(\frac{1}{w-w_j}+\frac{1}{w_j}\right)
 = x^{2}\,w\sum_{j=1}^{m-1}\frac{2\kappa_j}{w_j\,(w-w_j)},
\qquad w=\lambda^{2}.
\]
Since $|\lambda|<1-r$, one has $|w|=|\lambda|^{2}<(1-r)^{2}$. It suffices to show that every term obeys
\begin{equation}
\frac{2\kappa_j}{|w_j|\,|w-w_j|} \le \frac{2}{m\,\bigl((1-r)^{2}-|w|\bigr)} ,
\label{eq:termbound}
\end{equation}
since summing the $m-1$ terms then gives $\frac{2}{(1-r)^{2}-|w|}$. We may assume $\lambda\neq0$, the case $\lambda=0$ being trivial. In both cases below, $|w-w_j|\ge|w_j|-|w|$.

If $\sin\vartheta_j\ge1-r$, then $|w_j|\ge(1-r)\sin\vartheta_j$ by~\eqref{eq:polebound2}, and $(1-r)\sin\vartheta_j\ge(1-r)^{2}>|w|$, so
\[
|w_j|\,|w-w_j| \ge (1-r)\sin\vartheta_j\bigl((1-r)\sin\vartheta_j-|w|\bigr)>0 .
\]
Since $\kappa_j=\sin^{2}\vartheta_j/m$, inequality~\eqref{eq:termbound} is
therefore implied by
\[
\sin^{2}\vartheta_j\,\bigl((1-r)^{2}-|w|\bigr)\ \le\ (1-r)\sin\vartheta_j\bigl((1-r)\sin\vartheta_j-|w|\bigr).
\]
Dividing by $\sin\vartheta_j>0$ turns this into
$(1-r)^{2}\sin\vartheta_j-|w|\sin\vartheta_j\le(1-r)^{2}\sin\vartheta_j-(1-r)|w|$.
The terms $(1-r)^{2}\sin\vartheta_j$ cancel, and dividing by $|w|>0$ leaves
$1-r\le\sin\vartheta_j$, which is exactly the assumption of this case.

If $\sin\vartheta_j<1-r$, then $|w_j|\ge(1-r)^{2}$ by~\eqref{eq:polebound2} and
\[
\frac{2\kappa_j}{|w_j||w-w_j|}
 \le \frac{2\sin^{2}\vartheta_j/m}{(1-r)^{2}\bigl((1-r)^{2}-|w|\bigr)}
 < \frac{2(1-r)^{2}/m}{(1-r)^{2}\bigl((1-r)^{2}-|w|\bigr)}
 = \frac{2}{m\bigl((1-r)^{2}-|w|\bigr)} .
\]
Every one of the $m-1$ terms is thus at most $2/\bigl(m((1-r)^{2}-|w|)\bigr)$, and summing gives the claim. Finally $|M(x)|\le\frac{r^{2m}(1+r^{2})}{1-r^{2m}}=a_m$ because $|1-x^{2}|\le1+r^{2}$ and $|1-x^{2m}|\ge1-r^{2m}$.
\end{proof}

\subsection{Proof of Theorem~\ref{thm:onestep}}\label{app:onestep}

\begin{proof}
Fix $x$ with $|x|<r$ and apply Rouch\'e's theorem on the circle
$|\lambda|=\rho$ to $f(\lambda):=\Delta\lambda$ and
$g(\lambda):=\lambda^{2}-x^{2}R_m(\lambda,x)$, so that
$F=f+g$ is the eigenvalue equation of
Proposition~\ref{prop:eig_eq}. Since $\rho<1-r$,
Lemma~\ref{lem:Rflat} applies on that circle and gives
$|x^{2}R_m|\le|M(x)|+b(\rho)\rho^{2}$, so
\[
|g(\lambda)|\le\rho^{2}+|M(x)|+b(\rho)\rho^{2}
=|M(x)|+\rho^{2}\bigl(1+b(\rho)\bigr)
<a_m+a_m=2a_m=\Delta\rho=|f(\lambda)| .
\]
The second summand is at most $a_m$, which is exactly~\eqref{eq:onestep_hyp},
and the first is \emph{strictly} below $a_m$: the proof of
Lemma~\ref{lem:Rflat} bounds $|M(x)|$ by $(1+|x|^{2})|x|^{2m}/(1-|x|^{2m})$,
the value at $|x|$ of a function strictly increasing on $(0,1)$ whose value at
$r$ is $a_m$, and $|x|<r$ strictly. That is where the strict inequality Rouch\'e
needs comes from, and it is why~\eqref{eq:onestep_hyp} may be assumed with
equality allowed. Hence $F$ has precisely one zero in $|\lambda|<\rho$, as $f$
does, and it is simple, and since $\chi_{2m}=S_mF$ with $S_m\neq0$ on $\mathcal D$
by Lemma~\ref{lem:poles}(ii), the same holds for the characteristic polynomial.
This zero is an analytic function of $x$, since it equals
$\frac{1}{2\pi i}\oint_{|\lambda|=\rho}\lambda\,\partial_\lambda F/F\,d\lambda$ with
$F$ jointly analytic and zero-free on the contour for every $|x|<r$. Near
$x=0$ it coincides with the germ of Theorem~\ref{thm:series}, both being zeros
of $F$ in $|\lambda|<\rho$ that vanish at $x=0$, so the germ continues
analytically to all of $\{|x|<r\}$ as this zero.

Finally Lemma~\ref{thm:cauchy} gives $|c_k^{(m)}|\le\rho\,r^{-2m-2k}$, and $\rho\,r^{-2m}=\frac{2(1+r^{2})}{\Delta(1-r^{2m})}$.
\end{proof}

\subsection{Proof of Proposition~\ref{prop:m1}}\label{app:m1}

\begin{proof}
With $S_1=1$, $S_0=0$ the characteristic polynomial~\eqref{eq:char_poly} is $\lambda^{2}+\Delta\lambda-x^{2}$, whose root vanishing at $x=0$ is~\eqref{eq:m1}. For $|x|<r\le\Delta/2$ we have $|4x^{2}|<4r^{2}\le\Delta^{2}$, so the argument $\Delta^{2}+4x^{2}$ lies in the open disc of radius $\Delta^{2}$ about $\Delta^{2}$, hence in the right half-plane. Its principal square root is therefore analytic with positive real part, so $|\Delta+\sqrt{\Delta^{2}+4x^{2}}|\ge\Delta$ and the second form of~\eqref{eq:m1} gives $|\lambda_0^{(1)}|\le 2|x|^{2}/\Delta$. Lemma~\ref{thm:cauchy} with $m=1$ converts this into~\eqref{eq:m1_ck}.
\end{proof}

\subsection{Proof of Lemma~\ref{lem:rho_interval}}\label{app:rho}

\begin{proof}
For $\rho<1-r$ the quantity $(1-r)^{2}-\rho^{2}$ is positive, so multiplying $\rho^{2}\bigl(1+\frac{2r^{2}}{(1-r)^{2}-\rho^{2}}\bigr)\le a_m$ through by it and writing $s:=\rho^{2}$ gives the equivalent
\[
s\bigl((1-r)^{2}-s\bigr)+2r^{2}s \le a_m\bigl((1-r)^{2}-s\bigr),
\qquad\text{i.e.}\qquad
s^{2}-\mathsf ps+\mathsf q \ge 0 .
\]
The discriminant is $\mathsf p^{2}-4\mathsf q=\bigl((1-r)^{2}+a_m+2r^{2}\bigr)^{2}-4a_m(1-r)^{2} \ge\bigl((1-r)^{2}+a_m\bigr)^{2}-4a_m(1-r)^{2}=\bigl((1-r)^{2}-a_m\bigr)^{2}\ge0$, with strict inequality since $r>0$. So the quadratic has two real roots $s_-<s_+$ and $s^{2}-\mathsf ps+\mathsf q\ge0$ exactly for $s\le s_-$ or $s\ge s_+$. Evaluating at $s=(1-r)^{2}$,
\[
(1-r)^{4}-(1-r)^{2}\mathsf p+\mathsf q =(1-r)^{4}-(1-r)^{4}-2r^{2}(1-r)^{2}-a_m(1-r)^{2}+a_m(1-r)^{2} =-2r^{2}(1-r)^{2} < 0 ,
\]
so $s_-<(1-r)^{2}<s_+$. The branch $s\ge s_+$ is therefore outside the range $s<(1-r)^{2}$, leaving $s\le s_-$. Positivity of $s_-$ is clear from $\mathsf q>0$.
\end{proof}

\subsection{Proof of Corollary~\ref{cor:Dmin}}\label{app:dmin}

\begin{proof}
With $\rho=2a_m/\Delta$, condition~\eqref{eq:rho_interval} reads $\frac{2a_m}{\Delta}\le\sqrt{s_-}$, i.e.\ $\Delta\ge 2a_m/\sqrt{s_-}$. Substituting the form of $s_-$ from~\eqref{eq:sminus} and using $\mathsf q=a_m(1-r)^{2}$ gives the second expression in~\eqref{eq:Dmin}. Note that $s_-<(1-r)^{2}$ makes the accompanying constraint $\rho<1-r$ of~\eqref{eq:onestep_hyp} automatic.

For the monotonicity, eliminate $\rho$ from~\eqref{eq:onestep_hyp}: with
$\rho=2a_m/\Delta$ the two conditions read
\[
\frac{2a_m}{\Delta}<1-r
\qquad\text{and}\qquad
\Upsilon(a_m):=\frac{4a_m}{\Delta^{2}}
\left(1+\frac{2r^{2}}{(1-r)^{2}-4a_m^{2}/\Delta^{2}}\right)<1 .
\]
On the range where the first holds, $\Upsilon$ is a product of positive factors each strictly increasing in its argument, hence strictly increasing. As $a_m=(1+r^{2})\frac{t}{1-t}$ with $t=r^{2m}$, and $t\mapsto t/(1-t)$ is increasing while $t=r^{2m}$ decreases in $m$, the sequence $a_m$ is strictly decreasing. The conditions therefore hold at every $m>m_0$ as soon as they hold at $m_0$.
\end{proof}

\subsection{Proof of Proposition~\ref{prop:sparsity}}\label{app:sparsity}

\begin{proof}
Write $n=k-j$. By Theorem~\ref{thm:series}, $c_n^{(j)}=0$ whenever $j\ge2$ and $1<n<j$. A nonzero contribution from such a $j$ therefore requires $n\le1$ or $n\ge j$, i.e.\ $j\ge k-1$ or $2j\le k$. This is the stated vanishing. Note that $j=1$ is never excluded, since $1\le\lfloor k/2\rfloor$ for every $k\ge2$: the index $j=1$ always sits in the retained range.

Let $k\ge\max\{3,q+1\}$. Then $k-1\ge2$ and $k-1\ge q$, so Theorem~\ref{thm:series} does apply to the branch $\lambda_0^{(k-1)}$ and gives $c_1^{(k-1)}=-\Delta^{-1}$ exactly. The term $j=k-1$ of~\eqref{eq:alpha_rec} therefore contributes $-\Delta\alpha_{k-1}c_1^{(k-1)}=+\alpha_{k-1}$. Moreover $k\ge3$ makes $k-1>\lfloor k/2\rfloor$, so this term is not among those retained in the sum, and no index is counted twice. The remaining indices $\lfloor k/2\rfloor<j<k-1$ contribute nothing, which is~\eqref{eq:alpha_sparse}.

Both steps fail at $k=2$. There $j=k-1=1$, and $c_1^{(1)}=-\Delta^{-3}$, not $-\Delta^{-1}$. We therefore leave $k=2$ in its original form, which is~\eqref{eq:alpha_k2}.
\end{proof}

\subsection{Proof of Proposition~\ref{prop:alpha_growth}}\label{app:alpha}

\begin{proof}
We first establish a pointwise recurrence: $A_q\le r^{2q}\Delta M_P$ and
\begin{equation}
A_k\le r^{2k}\Delta M_P+r^{2}A_{k-1}+C_r\,\mathcal A_{\lfloor k/2\rfloor},
\qquad k>q .
\label{eq:A_sparse}
\end{equation}
Indeed $\alpha_q=\Delta p_{2q}$ gives $A_q=r^{2q}\Delta|p_{2q}|\le r^{2q}\Delta M_P$. For $k\ge3$, take moduli in the sparse recurrence~\eqref{eq:alpha_sparse}, multiply by $r^{2k}$, and bound the coefficients by~\eqref{eq:cauchy_bound_j}. The three terms rescale as
\[
r^{2k}\,\Delta|p_{2k}|\le r^{2k}\Delta M_P,
\qquad
r^{2k}|\alpha_{k-1}|=r^{2}A_{k-1},
\qquad
r^{2k}\,\Delta|\alpha_j|\frac{C_r}{\Delta\,r^{2(k-j)}}=C_r\,r^{2j}|\alpha_j|=C_rA_j ,
\]
and $\sum_{q\le j\le\lfloor k/2\rfloor}A_j=\mathcal A_{\lfloor k/2\rfloor}$. The remaining case $k=2$ occurs only when $q=1$. There~\eqref{eq:alpha_k2} together with $|c_1^{(1)}|\le\frac{C_r}{\Delta r^{2}}$ from~\eqref{eq:cauchy_bound_j} gives
\[
A_2=r^{4}|\alpha_2|
 \le r^{4}\Delta M_P+r^{4}\,\Delta|\alpha_1|\frac{C_r}{\Delta r^{2}}
 = r^{4}\Delta M_P+C_r\,r^{2}|\alpha_1|
 = r^{4}\Delta M_P+C_r\mathcal A_1 ,
\]
which is weaker than~\eqref{eq:A_sparse} at $k=2$, whose right-hand side carries the additional term $r^{2}A_1\ge0$. So~\eqref{eq:A_sparse} holds for every $k>q$ without exception.

Being a sum of non-negative terms, $\mathcal A$ is non-decreasing. Summing~\eqref{eq:A_sparse} over $k=q+1,\dots,L$ and adding
$A_q\le r^{2q}\Delta M_P$,
\[
\mathcal A_L \le \Delta M_P\!\sum_{k=q}^{L}\!r^{2k}
 + r^{2}\!\!\sum_{k=q+1}^{L}\!\!A_{k-1}
 + C_r\!\!\sum_{k=q+1}^{L}\!\!\mathcal A_{\lfloor k/2\rfloor} .
\]
Here $\sum_{k\ge q}r^{2k}=\frac{r^{2q}}{1-r^{2}}$ and $\sum_{k=q+1}^{L}A_{k-1}=\sum_{k=q}^{L-1}A_k=\mathcal A_{L-1}\le\mathcal A_L$. Also $\lfloor k/2\rfloor\le\lfloor L/2\rfloor$ for $k\le L$, so the last sum is at most $(L-q)\,\mathcal A_{\lfloor L/2\rfloor}$. Collecting the $\mathcal A_L$ terms and dividing by $1-r^{2}>0$,
\begin{equation}
\mathcal A_L \le \frac{\dfrac{r^{2q}\,\Delta M_P}{1-r^{2}}+C_r\,(L-q)\,\mathcal A_{\lfloor L/2\rfloor}}{1-r^{2}},
\qquad L>q .
\label{eq:Asum_rec}
\end{equation}

To reach the closed form~\eqref{eq:Asum_bound}, abbreviate $\nu:=\frac{r^{2q}\Delta M_P}{1-r^{2}}$ and prove $\mathcal A_L\le\nu\Phi_L$ by strong induction on $L$. At $L=q$ we have $\mathcal A_q=A_q\le r^{2q}\Delta M_P\le\nu=\nu\Phi_q$. Let $L>q$ and assume the bound for all smaller indices, in particular $\mathcal A_{\lfloor L/2\rfloor}\le\nu\Phi_{\lfloor L/2\rfloor}$, which also covers the case $\lfloor L/2\rfloor<q$, where $\mathcal A$ vanishes. Then~\eqref{eq:Asum_rec} and $\Phi_{\lfloor L/2\rfloor}\ge1$ give
\[
\mathcal A_L \le \frac{\nu+C_r(L-q)\,\nu\,\Phi_{\lfloor L/2\rfloor}}{1-r^{2}}
 \le \nu\,\frac{1+C_r(L-q)}{1-r^{2}} \Phi_{\lfloor L/2\rfloor}
 = \nu\,\Phi_L ,
\]
which is~\eqref{eq:Asum_bound}.

It remains to bound the growth of $\Phi$. Unrolling~\eqref{eq:Phi}, $\Phi_L$ is a finite product whose index is halved at each step: after $i$ steps it is at most $L/2^{i}$, so the corresponding factor is at most $ \frac{C_r}{1-r^{2}}\,L\,2^{-i}$, since $1+C_r(k-q)\le C_rk$ (as $C_rq\ge1$) and $k\le L/2^{i}$. Writing $\mathsf n$ for the number of factors and taking logarithms,
\[
\ln\Phi_L \le \sum_{i=0}^{\mathsf n-1}\bigl(\ln(\tfrac{C_r}{1-r^{2}} L)-i\ln2\bigr)
 \le -\frac{\ln2}{2}\,\mathsf n^{2}+\ln\bigl(\tfrac{\sqrt2\,C_r}{1-r^{2}}\,L\bigr)\,\mathsf n .
\]
The right-hand side is a concave quadratic in $\mathsf n$, so it is bounded by its maximum over all real $\mathsf n$, namely $\frac{\ln^{2}(\frac{\sqrt2\,C_r}{1-r^{2}} L)}{2\ln2}$, attained at $\mathsf n=\log_2(\frac{\sqrt2\,C_r}{1-r^{2}} L)$ which is~\eqref{eq:Phi_growth}. Since $\ln^{2}(\frac{\sqrt2\,C_r}{1-r^{2}} L)/L\to0$, we get
$\Phi_L^{1/L}\to1$.
\end{proof}

\subsection{Proof of Proposition~\ref{prop:tail_estimate}}\label{app:tail}

\begin{proof}
From $d_k=\sum_{j=q}^{L}\alpha_jc_{k-j}^{(j)}$ and the triangle inequality,
\[
|\mathcal E_L(x)| \le \sum_{j=q}^{L}|\alpha_j|\sum_{k>L}|c^{(j)}_{k-j}|\,|x|^{2k} .
\]
Re-indexing the inner sum with $n=k-j\ge L+1-j$ and using~\eqref{eq:cauchy_bound_j},
\[
\sum_{k>L}|c^{(j)}_{k-j}|\,|x|^{2k}
= |x|^{2j}\!\!\sum_{n>L-j}\!\!|c^{(j)}_n|\,|x|^{2n}
\le \frac{C_r|x|^{2j}}{\Delta}\sum_{n>L-j}\gamma^{n}
= \frac{C_r|x|^{2j}}{\Delta}\cdot\frac{\gamma^{\,L+1-j}}{1-\gamma}.
\]
Since $|x|\le x_0$ we have $|x|^{2j}\gamma^{-j}\le x_0^{2j}\gamma^{-j}=r^{2j}$, so
\[
|\mathcal E_L(x)|\ \le\ \frac{C_r\,\gamma^{\,L+1}}{\Delta(1-\gamma)}\sum_{j=q}^{L}r^{2j}|\alpha_j|
\ =\ \frac{C_r\,\gamma^{\,L+1}}{\Delta(1-\gamma)} \mathcal A_L ,
\]
and $\mathcal A_L\le\frac{r^{2q}\Delta M_P}{1-r^{2}}\Phi_L$ is exactly~\eqref{eq:Asum_bound}. The factor $\Delta$ cancels, leaving~\eqref{eq:tail_final}.
\end{proof}

\subsection{Proof of Proposition~\ref{prop:single}}\label{app:single}

\begin{proof}
Write $\lambda$ for~\eqref{eq:lam1}. It is even, and on $[0,\infty)$ it is
increasing and convex, with
\[
\lambda'(y)=\frac{2y}{\sqrt{\Delta^{2}+4y^{2}}}\in[0,1) ,
\qquad
y-\frac\Delta2\;\le\;\lambda(y)\;\le\;y ,
\]
the two bounds being $2y\le\sqrt{\Delta^{2}+4y^{2}}\le\Delta+2y$. The closed forms
in~\eqref{eq:rho} follow from
$\lambda(\sigma)=2\sigma^{2}/(\Delta+\sqrt{\Delta^{2}+4\sigma^{2}})$ and
$\lambda'(\sigma)=2\sigma/\sqrt{\Delta^{2}+4\sigma^{2}}$, and
$\theta/\varrho=\sqrt{\Delta^{2}+4\sigma^{2}}/2\sigma>1$.

(i) On $\ket{v_x}\otimes\mathbb C^{2}$ the operator $E_2(H_1)$ acts as the matrix
$E_2(x)$, so $\ket{v_x}\otimes\ket{\xi(x)}$ is an eigenvector of $\Sigma_0$ with
eigenvalue $\sigma+\theta_{\rm in}\lambda(d)-\theta_{\rm in}\lambda(x)=f_0(x)$.
Suppose $\Sigma_i\ket{\psi_i}=f_i(x)\ket{\psi_i}$ with
$\ket{\psi_i}=\ket{v_x}\otimes\ket{\Xi_i(x)}$ and put $y=f_i(x)$. Then
$E_2(\Sigma_i)=\Sigma_i\otimes X-\Delta I\otimes\ket1\!\bra1$ acts on
$\ket{\psi_i}\otimes\mathbb C^{2}$ as $E_2(y)$, so
$\ket{\psi_i}\otimes\ket{\xi(y)}=\ket{v_x}\otimes\ket{\Xi_{i+1}(x)}$ is an
eigenvector of $\Sigma_{i+1}$ with eigenvalue $\sigma-\theta\lambda(y)=f_{i+1}(x)$.
No radius enters, because $\ket{\xi(y)}$ exists for every real $y$. For the
locality, expand $\Sigma_i$ in products of Pauli operators. Then
$E_2(\Sigma_i)$ multiplies every term by $X$ on the new qubit and adds a $1$-local
term, so the locality grows by one per stage, from $\varkappa+1$ for $\Sigma_0$.

(ii) Put $u_i:=|f_i(0)|-\sigma$, so that $u_0=\theta_{\rm in}\lambda(d)$. Since
$\lambda(y)/y^{2}=2/(\Delta+\sqrt{\Delta^{2}+4y^{2}})$ decreases on
$(0,\infty)$ and $\theta_{\rm in}>\theta=2\sigma/\lambda(\sigma)$,
\[
u_0\;>\;\theta\lambda(d)\;=\;2\sigma\,\frac{\lambda(d)}{\lambda(\sigma)}\;>\;\frac{2d^{2}}{\sigma} .
\]
If $|y|>\sigma$ then $Q(y)=\sigma-\theta\lambda(|y|)<\sigma-\theta\lambda(\sigma)=-\sigma$,
and $|Q(y)|-\sigma=\theta\bigl(\lambda(|y|)-\lambda(\sigma)\bigr)$ because
$\theta\lambda(\sigma)=2\sigma$. Since $f_0(0)=\sigma+u_0>\sigma$, induction gives
$f_i(0)<-\sigma$ for every $i\ge1$ and
\begin{equation}
u_{i+1}=\theta\bigl(\lambda(\sigma+u_i)-\lambda(\sigma)\bigr)
\qquad(i\ge0) .
\label{eq:u_rec}
\end{equation}
By convexity $\lambda(\sigma+u)-\lambda(\sigma)\ge\lambda'(\sigma)\,u$, so
$u_{i+1}\ge\varrho u_i$ and $u_k\ge\varrho^{k}u_0$.

(iii) The stopband stays in $[-\sigma,\sigma]$. Indeed $f_0$ decreases on
$[d,\sigma]$ from $f_0(d)=\sigma$ to
$f_0(\sigma)=\sigma-\theta_{\rm in}(\lambda(\sigma)-\lambda(d))=-\sigma$, and $Q$ is
even and decreases on $[0,\sigma]$ from $Q(0)=\sigma$ to
$Q(\sigma)=\sigma-\theta\lambda(\sigma)=-\sigma$. So $f_0$ maps $[d,\sigma]$ into
$[-\sigma,\sigma]$ and $Q$ maps $[-\sigma,\sigma]$ into itself, and since every
$f_k$ is even,
\[
|f_k(x)|\le\sigma\quad\text{for } d\le|x|\le\sigma ,
\qquad
f_k(\pm\sigma)=-\sigma ,
\]
the second because $f_0(\sigma)=-\sigma$ and $-\sigma$ is a fixed point of $Q$. Now
$\widehat R_k(0)=1$ by definition, and by~(i) the operator $\Sigma_k/f_k(0)$ acts on
$\ket{v_x}\otimes\ket{\Xi_k(x)}$ as the number $\widehat R_k(x)$. Since
$\widehat R_k(H_1)-\ket m\!\bra m=\sum_{x\in\spec H_1\setminus\{0\}}\widehat R_k(x)\Pi_x$
with orthogonal spectral projectors $\Pi_x$, its norm is
$\max_{x\in\spec H_1\setminus\{0\}}|\widehat R_k(x)|$, and every
$x\in\spec H_1\setminus\{0\}$ has $d\le\delta\le|x|\le\sigma$
by~\eqref{eq:H1spec}, so this norm is at most $\sigma/|f_k(0)|$. If $\sigma=\|H_1\|$
then $\sigma$ or $-\sigma$ lies in $\spec H_1$, where $|f_k|=\sigma$, and the bound
is attained. By~(ii) it is at most
$\sigma/(\sigma+\varrho^{k}u_0)\le(\sigma/u_0)\varrho^{-k}<\frac{\sigma^{2}}{2d^{2}}\varrho^{-k}$,
which is at most $\eta$ once $k\ge\log_\varrho\frac{\sigma^{2}}{2d^{2}\eta}$.

(iv) \emph{The norm.} By~\eqref{eq:E2}, $\|E_2(A)\|\le\|A\|+\Delta$ for every
Hermitian $A$. Hence $\|\Sigma_{i+1}\|\le\theta\|\Sigma_i\|+\sigma+\theta\Delta$, and
since $\theta>2$ this gives
$\|\Sigma_{i+1}\|+\sigma+2\Delta\le\theta\bigl(\|\Sigma_i\|+\sigma+2\Delta\bigr)$.
Also $\|\Sigma_0\|\le\sigma+u_0+\theta_{\rm in}(\sigma+\Delta)$ because
$\|H_1\|\le\sigma$, so
\[
\|\Sigma_k\|\;\le\;\theta^{k}E ,
\qquad
E:=2\sigma+2\Delta+u_0+\theta_{\rm in}(\sigma+\Delta) .
\]
\emph{The marked value.} The norm grows as $\theta^{k}$, and we now show that the
marked value does too, after a slow start whose length is what produces the power
of $\sigma/d$. Recall from~(ii) that $u_i=|f_i(0)|-\sigma$
obeys~\eqref{eq:u_rec}. One stage admits two lower bounds,
\begin{equation}
u_{i+1}\;\ge\;\varrho\,u_i
\qquad\text{and}\qquad
u_{i+1}\;\ge\;\theta\Bigl(u_i-\frac\Delta2\Bigr) .
\label{eq:two_bounds}
\end{equation}
The first is the convexity bound of~(ii). The second follows
from~\eqref{eq:u_rec} with $\lambda(\sigma+u)\ge\sigma+u-\Delta/2$ and
$\lambda(\sigma)\le\sigma$. The first is the better one while $u_i$ is small, and
the second once $u_i$ is large compared with $\Delta$.

\emph{Slow phase.} Let $n$ be the least integer with $\varrho^{n}u_0\ge\Delta$, so
that $n=0$ if $u_0\ge\Delta$ and $n\le1+\log_\varrho(\Delta/u_0)$ otherwise. The
first bound of~\eqref{eq:two_bounds} alone gives $u_i\ge\varrho^{i}u_0$ for every
$i$, and in particular $u_n\ge\Delta$.

\emph{Fast phase.} For $m\ge0$ the first bound gives
$u_{n+m}\ge\varrho^{m}u_n\ge2^{m}\Delta$, because $\varrho>2$. The second bound
then reads
\[
u_{n+m+1}\;\ge\;\theta\,u_{n+m}\Bigl(1-\frac{\Delta}{2u_{n+m}}\Bigr)
\;\ge\;\theta\,u_{n+m}\bigl(1-2^{-m-1}\bigr) ,
\]
so each stage gains a factor $\theta$, up to a correction that tends to $1$
geometrically. All the corrections together cost less than a factor $4$, since
$\prod_{j\ge1}(1-2^{-j})\ge\tfrac12\cdot\tfrac34\bigl(1-\sum_{j\ge3}2^{-j}\bigr)
=\tfrac9{32}>\tfrac14$. Hence
\[
u_k\;\ge\;\tfrac14\,\theta^{k-n}\,u_n\;\ge\;\tfrac14\,\theta^{k-n}\varrho^{n}u_0
\qquad\text{for } k\ge n .
\]

\emph{Both phases.} For $k<n$ the slow phase gives
$u_k\ge\varrho^{k}u_0\ge\theta^{k-n}\varrho^{n}u_0$, since $\varrho<\theta$. So for
every $k\ge0$
\[
|f_k(0)|\;\ge\;u_k\;\ge\;\frac14\Bigl(\frac\varrho\theta\Bigr)^{n}\theta^{k}u_0 .
\]
The marked value thus grows as $\theta^{k}$, and the slow phase costs the factor
$(\varrho/\theta)^{n}$. With $\beta:=\log_\varrho\theta$ one has
$(\varrho/\theta)^{\log_\varrho y}=y^{1-\beta}$, so the bound on $n$ gives
\[
\Bigl(\frac\varrho\theta\Bigr)^{n}\;\ge\;
\frac\varrho\theta\,\min\Bigl\{1,\frac{u_0}{\Delta}\Bigr\}^{\beta-1} .
\]

\emph{The scaling.} Dividing the bound on the norm by the last two displays,
\[
\frac{\|\Sigma_k\|}{|f_k(0)|}\;\le\;\frac{4\theta}{\varrho}\cdot
\frac{E}{u_0\,\min\{1,u_0/\Delta\}^{\beta-1}}
\qquad\text{for every }k\ge0 .
\]
The numbers $\theta$, $\varrho$ and $\beta$ depend only on $\Delta/\sigma$. For
$d\le\sigma/2$ one has
$\theta_{\rm in}\le2\sigma/(\lambda(\sigma)-\lambda(\sigma/2))$ and
$u_0\le\theta_{\rm in}\lambda(\sigma/2)$, so $E$ is at most $\sigma$ times a
function of $\Delta/\sigma$. The denominator increases with $u_0$, and
$u_0>2d^{2}/\sigma$ by~(ii), so it is at least
$2\sigma\min\{1,2\sigma/\Delta\}^{\beta-1}(d/\sigma)^{2\beta}$. This is the claim,
with $C$ depending only on $\Delta/\sigma$.
\end{proof}

\section{The odd gadget}\label{app:odd}

The even chain of Section~\ref{sec:gadget} produces even functions of the input,
so the synthesis of Section~\ref{sec:synthesis} reaches even targets only. Odd
targets are supplied by a companion chain of odd length, and almost nothing has
to be redone for it: the interior of the chain is unchanged, so the minors
$\Theta_k$, the transfer matrix $T$, the Chebyshev convergents $S_k$ and the
ratio $R_n$ are literally the same objects, and
Lemmas~\ref{lem:transfer}, \ref{lem:ch}, \ref{lem:poles}, \ref{lem:Rflat}
and~\ref{lem:rho_interval} hold verbatim. This appendix therefore states the odd
results and proves only the three steps at which the boundary makes a
difference. Everything else is obtained by repeating Sections~\ref{sec:gadget}
and~\ref{sec:branch} word for word, under the dictionary of
Table~\ref{tab:dictionary}.

\subsection{The chain and its eigenvalue equation}\label{app:odd:def}

For each integer $n\ge1$ let $O_{2n+1}(x)$ be the $(2n+1)\times(2n+1)$
tridiagonal Hermitian matrix
\begin{equation}
O_{2n+1}(x) =
\begin{pmatrix}
0 & x & 0 & \cdots & 0 \\
x & 0 & 1 & \ddots & \vdots \\
0 & 1 & 0 & \ddots & 0 \\
\vdots & \ddots & \ddots & \ddots & 1 \\
0 & \cdots & 0 & 1 & \Do x
\end{pmatrix},
\qquad \Do>0 ,
\label{eq:odd-gadget-matrix}
\end{equation}
whose couplings again alternate $x,1,x,1,\dots$ starting with $x$, and which now
has an odd number of sites, so that the chain ends on a unit bond and carries
$n$ bonds of each kind. The boundary term is $\Do x$ rather than the constant
$-\Delta$ of~\eqref{eq:even-gadget-matrix}, and we write $\Do$ for it, never
$\Delta$: the two are different physical quantities, since $\Delta$ is an
on-site energy while $\Do$ multiplies the input. They are independent
parameters, and
Corollary~\ref{cor:Dmax} constrains them in opposite directions. Replacing $x$
by an input Hamiltonian gives
\begin{equation}
O_{2n+1}(H) \;=\; H\otimes\mathsf A_n \;+\; I\otimes\mathsf B_n ,
\qquad
\mathsf A_n=\sum_{k=1}^{n}\bigl(\ket{2k}\bra{2k-1}+\mathrm{h.c.}\bigr)+\Do\ket{2n+1}\bra{2n+1},
\label{eq:odd-split}
\end{equation}
with $\mathsf B_n$ the unit bonds, which is the same shape
$H\otimes\mathsf A+I\otimes\mathsf B$ that Theorem~\ref{thm:realisation} uses. In
the one-excitation encoding of~\eqref{eq:local-form} the boundary term is
$\Do\,H\otimes\tfrac12(1-Z_{2n})$, so the locality is again two more than that
of $H$.

\begin{proposition}[Odd eigenvalue equation]\label{prop:eig_eq_odd}
With $\Theta_k$ the interior minors of Section~\ref{sec:chareq}, which are the
same for both parities,
\begin{equation}
\begin{aligned}
\chi^{\mathrm o}_{2n+1}(\lambda) &:= \det\bigl(O_{2n+1}(x)-\lambda I\bigr)
\;=\; \Theta_{2n+1}+\Do x\,\Theta_{2n}\\
&\;=\; S_n\bigl[(\Do x-\lambda)(\lambda^{2}-x^{2})+\lambda\bigr]-x^{2}(\Do x-\lambda)S_{n-1},
\end{aligned}
\label{eq:char_poly_odd}
\end{equation}
and on $\mathcal D$, where $S_n\neq0$, the reduced eigenvalue equation is
\begin{equation}
\lambda\bigl(1+\Do x\lambda-\lambda^{2}\bigr)
\;=\;(\Do x-\lambda)\,x^{2}R_n(\lambda,x) .
\label{eq:eig_reduced_odd}
\end{equation}
\end{proposition}

\begin{proof}
The boundary modifies only the last diagonal entry, from $-\lambda$ to
$\Do x-\lambda$, and the last coupling is a unit bond, so expanding the last row
gives $\chi^{\mathrm o}_{2n+1}=(\Do x-\lambda)\Theta_{2n}-\Theta_{2n-1}$, and the
unit-bond step of~\eqref{eq:recurrence} at $j=n+1$ reads
$\Theta_{2n+1}=-\lambda\Theta_{2n}-\Theta_{2n-1}$, which is the second
expression. By~\eqref{eq:Tm_int} and~\eqref{eq:CH},
$(\Theta_{2n},\Theta_{2n-1})^{\mathsf T}=T^{n}(1,0)^{\mathsf T}
=\bigl(S_n(\lambda^{2}-x^{2})-x^{2}S_{n-1},\,-\lambda S_n\bigr)^{\mathsf T}$,
and substituting gives the closed form. Dividing by $S_n$ and writing
$S_{n-1}/S_n=R_n-1$ turns it into
$(\Do x-\lambda)(\lambda^{2}-x^{2})+\lambda=(\Do x-\lambda)(x^{2}R_n-x^{2})$, in
which the two terms $\mp(\Do x-\lambda)x^{2}$ cancel. Factoring $\lambda$ out of
the left-hand side is~\eqref{eq:eig_reduced_odd}.
\end{proof}

Two structural facts follow, and they are what everything below rests on.

\emph{Parity.} In~\eqref{eq:char_poly_odd} the quantity $\tau=\lambda^{2}-x^{2}-1$,
hence $S_n$ and $S_{n-1}$, and also $\lambda^{2}-x^{2}$ and $x^{2}$, are
invariant under $(\lambda,x)\mapsto(-\lambda,-x)$, while $\Do x-\lambda$ and
$\lambda$ both change sign. Hence
$\chi^{\mathrm o}_{2n+1}(-\lambda;-x)=-\chi^{\mathrm o}_{2n+1}(\lambda;x)$, so
$\lambda\in\spec O_{2n+1}(x)$ if and only if $-\lambda\in\spec O_{2n+1}(-x)$,
and any branch singled out by a condition invariant under that involution is an
\emph{odd} function of $x$.

\emph{The coefficient of $\lambda$.} Collecting in~\eqref{eq:eig_reduced_odd}
the terms linear in $\lambda$ shows that the role played by the detuning
$\Delta$ in the even equation $\lambda(\lambda+\Delta)=x^{2}R_m$ is played here
by $1+x^{2}R_n$, whose value at $\lambda=0$ is $1+M(x)$. That quantity is
explicit, and this is the identity the appendix turns on:
\begin{equation}
1+M(x)\;=\;1+\frac{x^{2n}(1-x^{2})}{1-x^{2n}}\;=\;\frac{1-x^{2n+2}}{1-x^{2n}} .
\label{eq:oneplusM}
\end{equation}
Its zeros are the $(2n+2)$-th roots of unity, \emph{all of modulus one}.
Consequently, for every $n\ge1$ and every $|x|\le r<1$,
\begin{equation}
\bigl|1+M(x)\bigr|\;\ge\;\frac{1-r^{2n+2}}{1+r^{2n}}\;\ge\;1-r^{2} ,
\label{eq:mubound}
\end{equation}
a bound uniform in $n$ and positive on every disc of radius $r<1$. It is because
of~\eqref{eq:mubound} that no chain length has to be excluded below.

\subsection{The branch}\label{app:odd:series}

\begin{theorem}[Odd branch]\label{thm:series_odd}
There is a unique branch $\lambda_1^{(n)}$ analytic near $x=0$ with
$\lambda_1^{(n)}(0)=0$. It is odd, and for every $n\ge1$
\begin{equation}
\lambda_1^{(n)}(x)\;=\;\Do\,\frac{x^{2n+1}\bigl(1-x^{2}\bigr)}{1-x^{2n+2}}
\;-\;\Do^{3}x^{4n+3}\;+\;O\bigl(x^{4n+5}\bigr) .
\label{eq:odd_expansion}
\end{equation}
In particular, writing $\lambda_1^{(n)}(x)=\sum_{k\ge0}c^{(n)}_kx^{2n+1+2k}$,
\begin{equation}
c^{(n)}_0=\Do,
\qquad
c^{(n)}_1=-\Do,
\qquad
c^{(n)}_k=0 \ \ (1<k\le n),
\qquad
c^{(n)}_{n+1}=\Do-\Do^{3} .
\label{eq:series_odd}
\end{equation}
\end{theorem}

\begin{proof}
Write $F^{\mathrm o}(\lambda,x):=\lambda(1+\Do x\lambda-\lambda^{2})-(\Do x-\lambda)x^{2}R_n$.
By Lemma~\ref{lem:poles}(ii) it is analytic near the origin, with
$F^{\mathrm o}(0,0)=0$ and $\partial_\lambda F^{\mathrm o}(0,0)=1\neq0$, so the
implicit function theorem gives the branch. Its characterisation is invariant
under $(\lambda,x)\mapsto(-\lambda,-x)$, so the branch is odd. Write
$\lambda:=\lambda_1^{(n)}(x)$ and, by Corollary~\ref{cor:threefold}(ii),
$x^{2}R_n=M(x)-\varepsilon$ with $0\le\varepsilon=O(x^{2}\lambda^{2})$. Moving
the terms linear in $\lambda$ to the left,
\begin{equation}
\lambda\bigl(1+M\bigr)\;=\;\Do x\,M\;-\;\Do x\lambda^{2}+\lambda^{3}-\Do x\varepsilon+\lambda\varepsilon .
\label{eq:odd_exact}
\end{equation}

\emph{Leading order.} Let $b_px^{2p+1}$ be the lowest nonzero term of $\lambda$.
Since $1+M=1+O(x^{2n})$, the left side of~\eqref{eq:odd_exact} has order $2p+1$,
while on the right $\Do xM$ has order exactly $2n+1$ and the four remaining terms
have orders $4p+3$, $6p+3$ and at least $4p+5$, all $>2p+1$. The right side
therefore has order $\min\{2n+1,\ \ge4p+3\}$, which is $2p+1$ only if $p=n$, and
matching the coefficient of $x^{2n+1}$ gives $b_n=\Do$.

\emph{The expansion.} Dividing~\eqref{eq:odd_exact} by $1+M$, which is a unit by
\eqref{eq:oneplusM}, the first term becomes
$\Do x\,M/(1+M)=\Do x^{2n+1}(1-x^{2})/(1-x^{2n+2})$ exactly. For the rest,
$\lambda=\Do x^{2n+1}+O(x^{2n+3})$ gives
$-\Do x\lambda^{2}=-\Do^{3}x^{4n+3}+O(x^{4n+5})$, while $\lambda^{3}=O(x^{6n+3})$
and the two $\varepsilon$-terms are $O(x^{4n+5})$, all three being $O(x^{4n+5})$
because $6n+3\ge4n+5$ for $n\ge1$, and multiplying by
$(1+M)^{-1}=1+O(x^{2n+2})$ changes none of this. That is
\eqref{eq:odd_expansion}. Finally
\[
\Do\,\frac{x^{2n+1}(1-x^{2})}{1-x^{2n+2}}
\;=\;\Do\sum_{i\ge0}\Bigl(x^{2n+1+(2n+2)i}-x^{2n+3+(2n+2)i}\Bigr)
\]
is lacunary of period $2n+2$: in the index $k$ of~\eqref{eq:series_odd} its
nonzero coefficients sit at $k\in\{0,1\}+(n+1)\mathbb Z_{\ge0}$, and they equal
$\Do$ and $-\Do$ alternately. Reading off $k=0,1$, the gap $1<k\le n$, and
$k=n+1$, where the lacunary term contributes $\Do$ and the correction
$-\Do^{3}$, gives~\eqref{eq:series_odd}.
\end{proof}

So Theorem~\ref{thm:series} transposes under $\Delta^{-1}\mapsto\Do$, with two
bonuses. The odd family is sparser, in that its run of zeros extends one index
further and~\eqref{eq:odd_expansion} exhibits further runs for which the even
case has no counterpart, and it has no exceptional shortest chain, since
\eqref{eq:series_odd} holds at $n=1$ as it stands whereas
$c^{(1)}_1=-\Delta^{-3}$ spoils the pattern in the even case.

\subsection{Coefficient bounds on the complex disc}\label{app:odd:radius}

Lemma~\ref{lem:Rflat} concerns $R_n$ alone and is available unchanged.
What differs is the Rouch\'e comparison: the natural linear part of
$F^{\mathrm o}$ is $\lambda(1+M(x))$, not $\lambda$, and~\eqref{eq:mubound} is
what keeps it from degenerating.

\begin{theorem}[Odd coefficient bound]\label{thm:onestep_odd}
Let $0<r<1$ and $n\ge1$, let $a_n$ be as in~\eqref{eq:ab}, and put
\begin{equation}
\bar a_n:=\min\Bigl\{a_n,\ \tfrac12\bigl(1-r^{2}\bigr)\Bigr\},
\qquad
\rho:=\frac{8\,\Do\,r\,a_n}{1-r^{2}} .
\label{eq:odd_contour}
\end{equation}
Assume $\rho\le\sqrt{s_-}$, with $s_-$ given by~\eqref{eq:sminus} with $a_m$
replaced by $\bar a_n$. Then $\lambda_1^{(n)}$ is analytic on $\{|x|<r\}$ and
\begin{equation}
\sup_{|x|<r}\bigl|\lambda_1^{(n)}(x)\bigr|\;\le\;\rho ,
\qquad
\bigl|c^{(n)}_k\bigr|\;\le\;\frac{\Do\,C^{\mathrm o}_{r,n}}{r^{2k}},
\qquad
C^{\mathrm o}_{r,n}:=\frac{8(1+r^{2})}{(1-r^{2})(1-r^{2n})} .
\label{eq:onestep_ck_odd}
\end{equation}
\end{theorem}

\begin{proof}
The proof of Lemma~\ref{lem:rho_interval} uses only $a_m>0$, so the lemma holds
with $a_m$ replaced by $\bar a_n$, and it gives $0<s_-<(1-r)^{2}$, so
$\rho\le\sqrt{s_-}<1-r$: the circle $|\lambda|=\rho$ lies in $\mathcal D$,
Lemma~\ref{lem:Rflat} applies on it, and by~\eqref{eq:rho_interval} the
hypothesis is equivalent to $\rho^{2}\bigl(1+b(\rho)\bigr)\le\bar a_n$, with
$b(\rho)=2r^{2}/((1-r)^{2}-\rho^{2})$.

Fix $x$ with $|x|<r$ and split $F^{\mathrm o}=f+g$ with
\[
f(\lambda):=\lambda\bigl(1+M(x)\bigr),
\qquad
g(\lambda):=\Do x\lambda^{2}-\lambda^{3}-\Do x\,x^{2}R_n+\lambda\bigl(x^{2}R_n-M\bigr) .
\]
Then $f$ has exactly one zero in $|\lambda|<\rho$, it is simple, and
$|f|\ge\rho\,(1-r^{2})$ on the circle by~\eqref{eq:mubound}. By
Lemma~\ref{lem:Rflat}, $|x^{2}R_n-M|\le b(\rho)\rho^{2}$ and
$|x^{2}R_n|\le a_n+b(\rho)\rho^{2}$ there, so
\[
|g|\;\le\;\Do r\rho^{2}+\rho^{3}+\Do r\bigl(a_n+b\rho^{2}\bigr)+\rho\,b\rho^{2}
\;=\;\Do r\Bigl[a_n+\rho^{2}\bigl(1+b\bigr)\Bigr]+\rho^{3}\bigl(1+b\bigr) .
\]
Using $\rho^{2}(1+b)\le\bar a_n\le a_n$ in the bracket and
$\rho^{3}(1+b)\le\rho\,\bar a_n\le\tfrac12\rho(1-r^{2})$ in the last term,
\[
|g|\;\le\;2\Do r\,a_n+\tfrac12\rho\bigl(1-r^{2}\bigr)
\;=\;\tfrac14\rho\bigl(1-r^{2}\bigr)+\tfrac12\rho\bigl(1-r^{2}\bigr)
\;=\;\tfrac34\rho\bigl(1-r^{2}\bigr)\;<\;|f| ,
\]
the middle equality because $\rho=8\Do ra_n/(1-r^{2})$ makes
$2\Do ra_n=\tfrac14\rho(1-r^{2})$. Rouch\'e therefore gives exactly one zero of
$F^{\mathrm o}$ in $|\lambda|<\rho$, as $f$ has, and that zero is simple. As in
the proof of Theorem~\ref{thm:onestep} it is an analytic function of $x$, being
$\frac{1}{2\pi i}\oint_{|\lambda|=\rho}\lambda\,\partial_\lambda F^{\mathrm o}/F^{\mathrm o}\,d\lambda$
with $F^{\mathrm o}$ jointly analytic and non-vanishing on the contour, and it
agrees near $x=0$ with the germ of Theorem~\ref{thm:series_odd}, so it is
$\lambda_1^{(n)}$. Finally the proof of Lemma~\ref{thm:cauchy} applies with
$x^{2m}$ replaced by $x^{2n+1}$, giving
$|c^{(n)}_k|\le\rho\,r^{-2n-1-2k}$, and
$\rho\,r^{-2n-1}=8\Do a_n/\bigl((1-r^{2})r^{2n}\bigr)$, which is
$\Do\,C^{\mathrm o}_{r,n}$.
\end{proof}

\begin{corollary}[Maximum boundary coupling]\label{cor:Dmax}
Theorem~\ref{thm:onestep_odd} applies at radius $r$ for every
\begin{equation}
\Do\;\le\;\Do^{\max}(n,r):=\frac{\bigl(1-r^{2}\bigr)\sqrt{s_-}}{8\,r\,a_n}\;>\;0 .
\label{eq:Dmax}
\end{equation}
In particular \emph{no chain length is excluded}: every $n\ge1$ is admissible at
every radius $r<1$, provided $\Do$ is small enough. Moreover
$\Do^{\max}(\cdot,r)$ is increasing, so $\Do\le\Do^{\max}(n_0,r)$ certifies
every $n\ge n_0$ at once: as in the even case, the shortest chain in the basis
binds.
\end{corollary}

\begin{proof}
Equation~\eqref{eq:Dmax} is $\rho\le\sqrt{s_-}$ solved for $\Do$, and $s_->0$
because $\bar a_n>0$. For the monotonicity write $\varpi_n:=8ra_n/(1-r^{2})$,
so that $\rho=\varpi_n\Do$ and the hypothesis reads
$\Do^{2}\le\bar a_n\big/\bigl[\varpi_n^{2}\bigl(1+b(\varpi_n\Do)\bigr)\bigr]$.
The ratio $\bar a_n/\varpi_n^{2}$ equals $(1-r^{2})^{2}/(64r^{2}a_n)$ when
$\bar a_n=a_n$ and $(1-r^{2})^{3}/(128r^{2}a_n^{2})$ otherwise. The two agree
where the cases meet and both decrease strictly in $a_n$, which decreases
strictly in $n$ by Corollary~\ref{cor:Dmin}. So $\bar a_n/\varpi_n^{2}$ is
non-decreasing in $n$, while $b(\varpi_n\Do)$ is non-increasing because $b$
increases in its argument and $\varpi_n$ decreases. The right-hand side of
the hypothesis is therefore non-decreasing in $n$.
\end{proof}

The direction of~\eqref{eq:Dmax} is the one substantive difference between the
two parities, and it is not an artefact of the estimate. For the even gadget the
detuning must be \emph{large}, $\Delta\ge\Delta_\star(r)$, because $\Delta$ is
the coefficient of $\lambda$ and it is what pins the branch near the origin. For
the odd gadget that role is played by $1+M(x)$, which \eqref{eq:mubound} already
bounds below, and $\Do$ instead multiplies the source term. Since
$c^{(n)}_0=\Do$, the branch grows linearly in $\Do$, so $\Do$ must be
\emph{small}. The cost is in the energy scale, not in the accuracy: the
triangular diagonal is $\Do$, so the synthesis weights scale like $\Do^{-1}$ and
$\|\Sigma\|$ with them, exactly as they scale like $\Delta$ in the even case,
while $\Do$ cancels out of the error bound as $\Delta$ does. And because
$\Delta$ and $\Do$ are independent parameters, an on-site energy on the even
chains and a coupling to the input on the odd ones, the two requirements never
conflict.

\subsection{The dictionary}\label{app:odd:transpose}

Let $P$ be an odd polynomial, $P(x)=\sum_{k=p}^{\ell}p_{2k+1}x^{2k+1}$ with
$p_{2p+1}\neq0$ and $p\ge1$, a linear term $p_1x$ being carried by $p_1H$
itself, which is $\varkappa$-local and needs no gadget, and put
$P^{\mathrm o}_L:=\sum_{n=p}^{L}\beta_n\lambda_1^{(n)}$. Since
$c^{(n)}_0=\Do\neq0$, the matching conditions again form a lower-triangular
system with constant diagonal, and Section~\ref{sec:synthesis} goes through
verbatim under the substitutions of Table~\ref{tab:dictionary}, with the
back-substitution
\begin{equation}
\beta_k=\frac1{\Do}\Bigl(p_{2k+1}-\sum_{n=p}^{k-1}\beta_n\,c^{(n)}_{k-n}\Bigr),
\qquad k=p,\dots,L .
\label{eq:beta_rec}
\end{equation}

\begin{table}[t]
\centering
\small
\renewcommand{\arraystretch}{1.25}
\begin{tabular}{@{}lll@{}}
\toprule
& \textbf{even}, Sections~\ref{sec:gadget}--\ref{sec:synthesis} & \textbf{odd}, this appendix\\
\midrule
chain & $E_{2m}$, $2m$ sites, detuning $-\Delta$ & $O_{2n+1}$, $2n+1$ sites, boundary $\Do x$\\
eigenvalue equation & $\lambda(\lambda+\Delta)=x^{2}R_m$ & $\lambda(1+\Do x\lambda-\lambda^{2})=(\Do x-\lambda)x^{2}R_n$\\
coefficient of $\lambda$ & $\Delta$ & $1+M(x)=\dfrac{1-x^{2n+2}}{1-x^{2n}}$\\
branch & even, order $x^{2m}$ & odd, order $x^{2n+1}$\\
$c_0,\ c_1$ & $\Delta^{-1},\ -\Delta^{-1}$ ($c_1$ for $m\ge2$) & $\Do,\ -\Do$\\
zeros of $c_k$ & $1<k<m$ & $1<k\le n$, and further runs\\
next coefficient & $c_m=\Delta^{-1}-\Delta^{-3}$ & $c_{n+1}=\Do-\Do^{3}$\\
Rouch\'e contour & $\rho=2a_m/\Delta$ & $\rho=8\Do ra_n/(1-r^{2})$\\
admissibility & $\Delta\ge\Delta_\star(r)$ & $\Do\le\Do^{\max}(n,r)$\\
coefficient constant & $C_r=\dfrac{2}{1-r^{2}}$ & $C^{\mathrm o}_r=\dfrac{8(1+r^{2})}{(1-r^{2})^{2}}$\\
triangular diagonal & $\Delta^{-1}$ & $\Do$\\
\bottomrule
\end{tabular}
\caption{The two parities side by side. Every statement of
Sections~\ref{sec:gadget}--\ref{sec:synthesis} transposes under these
substitutions. Here $C^{\mathrm o}_r$ dominates $C^{\mathrm o}_{r,n}$
of~\eqref{eq:onestep_ck_odd} for every $n\ge1$.}
\label{tab:dictionary}
\end{table}

Three checks make the transposition exact rather than merely plausible.
\emph{Sparsity:} Theorem~\ref{thm:series_odd} kills $c^{(n)}_{k-n}$ whenever
$1<k-n\le n$, a range strictly larger than the one used in
Proposition~\ref{prop:sparsity}, and $c^{(k-1)}_1=-\Do$ holds for \emph{every}
$k\ge2$, so the term $n=k-1$ of~\eqref{eq:beta_rec} contributes
$+\beta_{k-1}$ with no exceptional case, and there is no counterpart of the
carve-out~\eqref{eq:alpha_k2}. \emph{Growth and truncation:} the factor
$\Do^{-1}$ of~\eqref{eq:beta_rec} meets the factor $\Do$
of~\eqref{eq:onestep_ck_odd}, so the two cancel, leaving $C^{\mathrm o}_r$
exactly where $C_r$ stood and $r^{2p}M_P/\Do$ where $r^{2q}\Delta M_P$ stood.
Propositions~\ref{prop:alpha_growth} and~\ref{prop:tail_estimate} and
Corollary~\ref{cor:L_needed} hold with those replacements, and in particular the
admissible input norm again matches the analyticity radius $r$.
\emph{Realisation:} by~\eqref{eq:odd-split} the odd gadgets have the form
$H\otimes\mathsf A+I\otimes\mathsf B$ with $\mathsf A,\mathsf B$ supported on one
chain, which is all the proof of Theorem~\ref{thm:realisation} uses, so odd
chains commute with one another and with the even ones. For a target
$P=p_0+p_1x+P_{\mathrm e}+P_{\mathrm o}$, split into its two parities with
$P_{\mathrm e}$ starting at degree two and $P_{\mathrm o}$ at degree three, the
single Hamiltonian
\begin{equation}
\Sigma\;=\;p_0I+p_1H+\sum_{j}\alpha_j\,E_{2j}(H)\;+\;\sum_{n}\beta_n\,O_{2n+1}(H)
\label{eq:Sigma_mixed}
\end{equation}
therefore carries $p_0+p_1x+P^{e}_L+P^{\mathrm o}_L$ on the joint encoded subspace, as
soon as $\Delta\ge\Delta_\star(r)$ and $\Do\le\Do^{\max}(n_{\min},r)$.

One thing does not transpose, and we have not needed it: the real-axis bracket
of Proposition~\ref{prop:bracket}. An analogue follows the same way from
Corollary~\ref{cor:threefold}, but the series of Theorem~\ref{thm:series_odd}
was read off the exact identity~\eqref{eq:odd_exact} instead.

\section*{Acknowledgements}

JR acknowledges support from the Belgian Fonds de la Recherche Scientifique - FNRS under Grant O.0013.22 (EoS CHEQS).
JC acknowledges support from the Maison du Quantique de Nouvelle-Aquitaine ``HybQuant'', 
as part of the HQI initiative and France 2030, under the French National Research Agency (ANR) grant ANR-22-PNCQ-0002.
AB was supported by the French PEPR integrated project HQI (ANR-22-PNCQ-0002) and 
the Momentum MSCA Programme which is co-funded by the European Commission through 
the HORIZON-MSCA-2023-COFUND programme and the Secretariat of the Hungarian Academy 
of Sciences (MTA), Grant Agreement No. 101179854.

\bibliographystyle{plain}
\bibliography{refs}

\end{document}